%% file: main.tex
\documentclass[final]{siamltex704}
\usepackage[utf8]{inputenc}
\usepackage{microtype}
\usepackage{srcltx}

\usepackage{amsmath}
\usepackage{amssymb}
\usepackage{amsfonts}
\usepackage{pgfplots}

\newcommand{\refbubble}[1]{}

\let\accentvec\vec

\let\vec\accentvec

\let\originalleft\left
\let\originalright\right
\renewcommand{\left}{\mathopen{}\mathclose\bgroup\originalleft}
\renewcommand{\right}{\aftergroup\egroup\originalright}

\def\LC{\ensuremath{\textnormal{\textsf{LC}}}}
\def\LCY#1{\ensuremath{\textnormal{\textsf{LCY}}(#1)}}

\def\CS#1{\ensuremath{\textnormal{\textsf{CG}}(#1)}}

\def\Csmall{\ensuremath{\textnormal{\textsf{C}}_\textnormal{\textit{small}}}}
\def\LL{\ensuremath{\textnormal{\textsf{LL}}}}
\def\HT{\ensuremath{\textnormal{\textsf{HT}}}}
\def\MOS{\ensuremath{\textnormal{\textsf{MCOG}}}}
\def\SP{\ensuremath{\textnormal{\textsf{SP}}(L)}}
\def\RSP{\ensuremath{\textnormal{\textsf{RSP}}(L)}}
\def\PWT#1{\ensuremath{\textnormal{\textsf{PWT}}(#1)}}
\def\RWT#1{\ensuremath{\textnormal{\textsf{RWT}}(#1)}}
\def\MCWT#1{\ensuremath{\textnormal{\textsf{MCWT}}(#1)}}
\def\MC#1{\ensuremath{\textnormal{\textsf{MC}}(#1)}}
\def\WT#1{\ensuremath{\textnormal{\textsf{WT}}(#1)}}

\newcommand{\GG}{\mathcal{G}}

\newcommand{\RR}{\mathcal{Z}}
\newcommand{\HH}{\mathcal{H}}
\newcommand{\FF}{\mathcal{F}}
\newcommand{\rmTheta}{\mathrm{\Theta}}
\newcommand{\PAT}{P\v{a}tra\c{s}cu}
\newcommand{\ex}{\textrm{ex}}

\newcommand{\E}{\textnormal{E}}

\newcommand{\bG}{\ensuremath{G}}
\newcommand{\uG}{\ensuremath{G}}

\newcommand{\qed}{}
\def\url#1{{\tt #1}}

\newtheorem{claim}[theorem]{Claim}
\newtheorem{remark}[theorem]{Remark}

\DeclareMathOperator*{\argmin}{arg\,min}    

\title{A Simple Hash Class with Strong Randomness Properties in 
Graphs and Hypergraphs%
\thanks{This work was supported in part by DFG grant DI~412/10-1, DFG
grant DI~412/10-2, and by a Discovery Grant from the National Sciences and
Research Council of Canada (NSERC). Part of this work was done during a visit of
the second author to MPI Saarbr{\" u}cken, Germany, and during Dagstuhl seminars
07391, 08381, and 11121.
This work 
encompasses the results of the conference papers
\cite{DW2003a}, \cite{Woelfel06}, and---to some extent---\cite{AumullerDW12}, and adds
many new insights.}
}

\author{Martin Aum\"{u}ller\thanks{
IT University of Copenhagen, 2300 København, Denmark. ({\tt
maau@itu.dk})}\and
Martin
Dietzfelbinger\thanks{
Fakult\"{a}t f\"{u}r Informatik und Automatisierung, Technische
Universit\"{a}t Ilmenau, 98694 Ilmenau, Germany. ({\tt martin.dietzfelbinger@tu-ilmenau.de})} \and Philipp Woelfel\thanks{
Department of Computer Science,
University of Calgary, Calgary, Alberta T2N 1N4, Canada. ({\tt
woelfel@cpsc.ucalgary.ca})
}}

\begin{document}
\maketitle

\begin{abstract}
  We study randomness properties of graphs and hypergraphs generated by
  simple hash functions.  Several hashing applications can be analyzed by
  studying the structure of $d$-uniform random ($d$-partite) hypergraphs
  obtained from a set $S$ of $n$ keys and $d$ randomly chosen hash functions
  $h_1,\dots,h_d$ by associating each key $x\in S$  with a hyperedge
  $\{h_1(x),\dots, h_d(x)\}$.
  Often it is assumed that $h_1,\dots,h_d$ exhibit a high degree of independence.
  We present a simple construction of a hash class whose hash functions have
  small constant evaluation time and can be stored in sublinear space. We
  devise general techniques to analyze the randomness properties of the graphs and
  hypergraphs generated by these hash functions, and we show that they can
  replace other, less efficient constructions in cuckoo hashing (with and
  without stash), the simulation of a uniform hash function, the construction of 
  a perfect hash function, generalized cuckoo hashing
	and different load balancing scenarios.

\end{abstract}

\begin{keywords}
  Hashing, cuckoo hashing, randomized algorithms, random graphs, load balancing
\end{keywords}

\begin{AMS}   	
68P05, 68R10, 68W20, 05C80
\end{AMS}

 \section{Introduction}

We study randomness properties of graphs and hypergraphs generated by hash functions
of a particulary simple structure.
Consider a set $S$ of $n$ keys chosen from a finite set $U$, and a sequence 
$\vec{h} = (h_1,\dots,h_d)$, $d\geq2$, of random hash functions 
$h_i\colon U\rightarrow [m]=\{0,\dots,m-1\}$ for some positive integer $m$.
Then $S$ and $\vec{h}$ naturally define a $d$-partite $d$-uniform hypergraph $\uG(S,\vec{h}):=(V,E)$ with $V=V_{m,d}$, where
$V_{m,d}$ is the union of $d$ disjoint copies of $[m]$ and
$E=\bigl\{\bigl(h_1(x),\dots,h_d(x)\bigr)\mid x\in S\bigr\}$.

Properties of such (hyper-)graphs are essential in the analysis of a number of
randomized algorithms from a variety of applications, such as balanced
allocation \cite{SchickingerS00,stemann96}, shared memory and PRAM simulations
\cite{MadHSS96,KarpLH96}, perfect hashing \cite{FoxHCD92,MajewskiWHC96}, and
recent hashing based dictionaries \cite{FotakisPSS05,stash,cuckoo_hashing_pagh,Eppstein14,Khosla13}.

Often such algorithms are analyzed under the idealized \emph{uniform hashing assumption}, which says that every hash function $h$ employed by the algorithm is a truly random function.
For example, if all functions $h_1,\dots,h_d$ are truly random hash functions, then the graph $\uG(S,\vec{h})$ 
can be analyzed using the vast array of tools from random graph theory \cite{random_graphs}.
Since it is infeasible to store truly random functions, it is preferable to
use small sets of hash functions (called \emph{hash classes}), and sample
random hash functions from those sets.
These sets provide some weaker randomness guarantees.
Carter and Wegman \cite{carter_universal_hashing} defined a \emph{universal hash class} $\HH$ as one which guarantees that any two distinct keys $x,x'\in U$ are mapped by a random function $h\in \HH$ to the same function value only with a probability of $O(1/m)$.
A stronger notion is that of a \emph{$k$-wise independent} hash class $\HH$ \cite{WegmanC79}, which says that for any $k$ distinct keys $x_1,\dots,x_k$ and a random hash function $h\in \HH$ the vector $\bigl(h(x_1),\dots,h(x_k)\bigr)$ is uniformly distributed over $[m]^k$.  The canonical representation of a $k$-independent hash class
    is the class of all degree $k-1$ polynomials over some prime field. For
    the representation of such a polynomial, we just store its $k$ coefficients
    ($k$ words). The evaluation is possible in time $O(k)$.

Sometimes, ad-hoc analyses show that such limited randomness properties are sufficient for algorithms to exhibit the desired behavior.
For example, Pagh, Pagh, and Ru{\v z}i{\' c} \cite{pagh_lin_probing} showed that
closed hashing with linear probing works well using only 5-wise independence.
On the other hand, insufficient randomness can have subtle and unexpected
negative effects on some algorithms. For example, P{\v a}tra{\c s}cu and Thorup
\cite{PatrascuT15} showed that linear probing behaves badly when used with a certain 
artificial class of 4-wise
independent hash functions. In the same vein, it was experimentally observed in
\cite{cuckoo_hashing_pagh} and formally proved by Dietzfelbinger and Schellbach
in \cite{DS09b, DS09a} that cuckoo hashing using the multiplicative class of
hash functions from~\cite{DietzfelbingerHKP97}, although universal, 
does not work with high probability.  For many other applications, such as
    cuckoo hashing \cite{cuckoo_hashing_pagh} and $\varepsilon$-minwise independent hashing
    \cite{Indyk01}, we know that
    a logarithmic degree of independence suffices (in the size of the key set for the former, 
    in $1/\varepsilon$ for the latter). In that case, polynomials use
    logarithmic space and evaluation time. If one aims for constant evaluation
    time, there is the construction by Siegel \cite{Siegel04}---although 
    Siegel states that his construction has constant albeit
    impractical evaluation time---and, more recently, the more efficient
    constructions by Thorup \cite{Thorup13} and Christiani, Pagh, 
		and Thorup \cite{ChristianiPT15}.

Several techniques to circumvent the uniform hashing assumption have been proposed.
The most general one is to ``simulate'' uniform hashing.
The idea is to generate a class $\HH$ of hash functions at random such that 
for arbitrary given $S\subseteq U$ with
high probability $\HH$ is ``uniform'' on $S$, 
which means that a random
hash function $h\in \HH$ restricted to the domain $S$ is a truly random function.
Such a simulation was presented by Pagh and Pagh in
\cite{pagh_uniform_conference,pagh_uniform}, Dietzfelbinger and Rink \cite{DietzfelbingerR09}, and in the precursor work
\cite{DW2003a}. 
However, such simulations require at least a linear (in $|S|\cdot\log m$) number
of bits of additional space, which is often undesirable.

An alternative is the so-called \emph{split-and-share} technique \cite{FotakisPSS05,Dietzfelbinger07,DietzfelbingerR09}, in which $S$ is first partitioned by a top-level hash function into smaller sets of keys, called bins.
Then, a problem solution is computed for each bin, but all bins share the same hash functions.
Since the size of each bin is significantly smaller than the size of $S$, it is possible to use a hash function that behaves like a truly random function on each bin.
Finally, the problem solution of all bins is combined to a solution of the original problem.
This technique cannot be applied uniformly to all applications, as ad-hoc
algorithms depending on the application are required to merge the individual solutions for each bin to a solution of the original problem.
In some scenarios, e.g., balanced allocation with high loads, the small deviations in the bin sizes incurred by the top-level hash function are undesirable. Moreover, additional costs in space and time are caused by the top-level splitting hash function and by compensating for a larger failure probability in each of the smaller bins.

 Another perspective on uniform hashing is to assume that the key set $S = \{x_1, \ldots, x_n\}$ itself is ``sufficiently random''.
Specifically, Mitzenmacher and Vadhan showed in \cite{MitzenmacherV08} that
when the distribution that governs $\{x_1, \ldots,x_n\}$ has a low enough
collision probability, then even using a hash function $h$  from a $2$-wise
independent hash class $\mathcal{H}$
makes the sequence  $(h, h(x_1), \ldots, h(x_n))$ distributed close to the
uniform distribution on $\mathcal{H} \times R^n$ (see also
\cite{Dietzfelbinger12}).  

Besides these general techniques to circumvent the uniform hashing assumption, some research focuses on particular hash classes and their properties.
P{\v a}tra{\c s}cu and Thorup \cite{PatrascuT12} studied simple tabulation hashing, where each key is a tuple $(x_1,\dots,x_c)$ which is mapped to the hash value\footnote{$\oplus$ denotes the bit-wise XOR operation} $\bigl(f_1(x_1)\oplus\dots\oplus f_c(x_c)\bigr)\bmod m$   by $c$ uniform random hash functions $f_1,\dots,f_c$, each with a domain of cardinality $\lceil|U|^{1/c}\rceil$.
The authors showed that simple tabulation hashing has striking randomness
properties, and several applications (for example cuckoo hashing) exhibit good
behavior if such hash functions are used. One year later, the same authors
    introduced ``twisted tabulation hashing'' \cite{PatrascuT13}, which gives
    even stronger randomness properties in many applications. Furthermore, Dahlgaard and
    Thorup proved that twisted tabulation is $\varepsilon$-minwise independent
    \cite{DahlgaardT14}. Very recently, 
    Dahlgaard, Knudsen, Rotenberg, and Thorup 
    extended the use of simple tabulation hashing to load
    balancing \cite{DahlgaardKRT16}, showing that simple tabulation suffices for
    sequential load balancing with two choices. Simple tabulation hashing provides
    constant evaluation time with a description
    length that is polynomial (with exponent smaller than 1) in the size of the key set,
		just as with the hash functions studied in the present paper. 
		Each application of tabulation hashing requires its own analysis. 
	There are other approaches that trade higher evaluation time for 
    smaller description length. For example, 
    Reingold, Rothblum, and Wieder \cite{ReingoldRW14} showed that a class of hash
    functions introduced by Celis \emph{et al.} \cite{CelisRSW13} has strong
    enough randomness properties for running a slightly modified version of
    cuckoo hashing and sequential load balancing with two choices. While the
    hash class has non-constant evaluation time, its description length is
    notably smaller than what one gets using the standard polynomial approach
    for $\log n$-wise independence ($O(\log n \log \log n)$ vs. $O(\log^2 n)$
    bits) or tabulation hashing. 

    \paragraph{The Contribution} In this paper we focus on the properties of random graphs $\uG(S,\vec{h})$ generated by
simple hash functions. These hash functions have been described before by Aumüller, Dietzfelbinger, and Woelfel in \cite{AumullerDW14}. 
A function from their hash class, called $\RR$, combines simple $k$-independent hash functions with lookups in random tables. It can be evaluated efficiently in constant time, using a few arithmetic operations and table lookups.
Each hash function can be stored in sublinear space, more precisely using $O(n^\gamma)$ bits for some $\gamma<1$.
    To put our contribution in perspective,
    we first review some background. Building upon the
    work of Dietzfelbinger and Meyer auf der Heide \cite{DietzfelbingerH90},
    Aumüller, Dietzfelbinger, and Woelfel \cite{AumullerDW14} showed that hash functions from class $\RR$ 
		have randomness properties strong enough to run 
    cuckoo hashing with a stash with guarantees only known for fully random hash
    functions. To prove this result, Aumüller, Dietzfelbinger, and
    Woelfel studied the randomness properties of $G(S, h_1, h_2)$ when 
		the hash function pair $(h_1,h_2)$ is chosen randomly from $\RR$. They showed that the connected
    components of this graph behave, in some technical sense, very close to what
    is expected of the graph $G(S, h_1, h_2)$ when $(h_1, h_2)$ is fully random.
    
    Our contribution is that we provide a 
    general framework that allows us to analyze applications whose analysis 
    is based on arguments on the random graph described above when hash functions from $\RR$ are used
    instead of fully random hash functions. To argue that
    the hash class can run a certain application, only random graph theory
    is applied, no details of the actual hash class need to be considered.
    Using this framework, we show 
    that hash functions from $\RR$ 
    have randomness properties strong enough for many
    different applications, e.g., cuckoo hashing with a stash as described by
    Kirsch, Mitzenmacher, and Wieder in 
    \cite{stash:journal:09}; generalized cuckoo hashing as proposed by Fotakis,
    Pagh, Sanders, and Spirakis in 
     \cite{FotakisPSS05} with two recently discovered insertion
    algorithms due to Khosla \cite{Khosla13} and Eppstein, Goodrich, Mitzenmacher and Pszona 
    \cite{Eppstein14} (in a sparse setting); the construction of a perfect hash function of Botelho,
    Pagh and Ziviani
    \cite{BotelhoPZ13}; the simulation of a uniform hash function of Pagh and Pagh \cite{pagh_uniform};
    different types of load balancing as
    studied by Schickinger and Steger \cite{SchickingerS00}. 
    The analysis is done in a unified way which
    we hope will be of independent interest. We will find sufficient conditions under which it is possible
    to replace the full randomness assumption of a sequence of hash functions
    with explicit hash functions.  

    \paragraph{The General Idea} 

The analysis of hashing applications is often concerned with bounding (from
above) the probability that random hash functions $h_1,\dots,h_d$ map a given
set $S\subseteq U$ of keys to some ``bad'' configuration of hash function values.
These undesirable events can often be described by certain properties exhibited
by the random graph $\uG(S,\vec{h})$. (Recall the notation $\vec{h}=(h_1,\dots,h_d)$.)
For example, in cuckoo hashing a bad event occurs
when $\bG(S,h_1,h_2)$ contains a very long simple path or a connected
component with at least two cycles \cite{cuckoo_hashing_pagh,devroye}.

If $h_1,\dots,h_d$ are uniform hash functions, often a technique called
\emph{first moment method} (see, e.g., \cite{random_graphs}) is employed to
bound the probability of undesired events: In the standard analysis, one
calculates the expectation of the random variable $X$ that counts the number of
subsets $T\subseteq S$ such that the subgraph $\uG\bigl(T,\vec{h}\bigr)$ forms a
``bad'' substructure, as e.g., a connected component with two or more cycles.
This is done by summing  the probability that the
subgraph $\uG\bigl(T,\vec{h}\bigr)$ forms a ``bad'' substructure over all subsets $T\subseteq S$. One then shows that
$\E(X) = O(n^{-\alpha})$ for some $\alpha > 0$ and concludes that $\Pr(X >
0)$---the probability that an undesired event happens---is at most
$O(n^{-\alpha})$ by Markov's inequality. 

We state sufficient conditions allowing us to replace uniform hash
functions $h_1,\dots,h_d$ with hash function sequences from $\RR$ without
significantly changing the probability of the occurrence of certain undesired
substructures  $\uG\bigl(T,\vec{h}\bigr)$.  On a high level, the idea is  
as follows: We assume that for each $T \subseteq
U$ we can split $\RR$ into two disjoint parts: hash function sequences being
$T$-\emph{good}, and hash function sequences being $T$-\emph{bad}. Choosing
$\vec{h}=(h_1,\dots,h_d)$ at random from the set of $T$-good hash functions
ensures that the hash values $h_i(x)$ with $x\in T$ and $1 \leq i \leq d$ are
distributed fully randomly.
Fix some set $S \subseteq U$. We
identify some ``exception set'' $B_S \subseteq \RR$ (intended to be very small)
such that for all $T \subseteq S$ we have: If $G(T,\vec{h})$ has an undesired
property (e.g., a connected component with two or more cycles) and $\vec{h}$ is
$T$-bad, then $\vec{h} \in B_{S}$.

For $T \subseteq S$, disregarding the hash functions from $B_{S}$ will allow
us to calculate the probability that $G(T,\vec{h})$ has an undesired property as if
$\vec{h}$ were a sequence of fully random hash functions. 
Specifically, in \cite[Lemma 2]{AumullerDW14} it was already shown that 
$$\Pr\nolimits_{\vec{h} \in \RR}(X > 0) \leq \E(X) + \Pr\nolimits_{\vec{h} \in \RR}(B_S),$$
where the expectation is calculated assuming that $\vec{h}$ is a pair
of fully random hash functions. So, it is critical to find subsets $B_{S}$ of sufficiently small probability.
Whether or not this is possible depends on the substructures we are interested in. Here, we deviate from \cite{AumullerDW14} and provide general criteria that allow us to bound the size of $B_{S}$ from
above entirely by using graph theory. This means that details about the
hash function construction need not be known to argue that random hash functions from
$\RR$ can be used in place of uniform random hash functions for certain
applications.

\paragraph{Outline and Suggestions} Section~\ref{hashing:sec:basics} introduces
the considered class $\RR$ of hash functions and provides the general framework of our
analysis. Because of its abstract nature, the details of the framework might be
hard to understand. A simple application of the framework is 
provided in Section~\ref{hashing:sec:example:c:h}. There, we will
discuss the use of hash class $\RR$ in static cuckoo hashing.  
The reader might find it helpful to study the example first to get 
a feeling of how the framework is applied. 
Another way to approach the framework is to first read the
paper \cite{AumullerDW14}. This paper discusses one example of the framework with
an application-specific focus, which might be easier to understand. 

The following sections deal with applications of the hash function construction. 
Because these applications are quite diverse, the background of each one will be
provided in the respective subsection right before the analysis.

Section~\ref{hashing:sec:applications:simple:graphs}  deals with randomness
properties of $\RR$ on (multi-)graphs.  Here,
Subsection~\ref{hashing:sec:leafless} provides some groundwork for bounding the
impact of using $\RR$ in our applications.
The subsequent subsections discuss the use of $\RR$
in cuckoo hashing (with a stash), the simulation of a uniform hash function,
the construction of a perfect hash function, and the behavior of $\RR$ on 
connected components of $G(S, h_1, h_2)$. 

The next section (Section~\ref{hashing:sec:applications:hypergraphs}) deals with
applications whose analysis builds upon hypergraphs. As an introduction, we
study generalized cuckoo hashing with $d \geq 3$ hash functions when
the hash table load is low. Then, we consider two
recently described alternative insertion algorithms for generalized cuckoo hashing. Finally, we
prove that hash class $\RR$  provides randomness properties strong enough for
many different load balancing schemes.

In Section~\ref{hashing:sec:generalization} we show how our analysis generalizes
to the case that we use more involved hash functions as building blocks of hash
class $\RR$, which lowers the total number of hash functions needed and the
space consumption.

    \input{sec-2-setup}

    \input{sec-3-graphs}

    \input{sec-4-hypergraphs}

    \input{sec-5-generalization}

    \section{Conclusion and Open Questions}  We have
    described a general framework for analyzing hashing-based algorithms 
    and data structures whose
    analysis depends on properties of the random graph $G(S, \vec{h})$, where
    $\vec{h}$ comes from a certain class $\RR$ of simple hash functions. This class
    combined lookups in small random tables with the evaluation
    of simple $2$-universal or $2$-independent hash functions. 

    We developed a framework that allowed us to consider what happens to certain 
    hashing-based algorithms or data structures when fully random hash functions
    are replaced by hash functions from hash class $\RR$. If the analysis
    works using the so-called first-moment method in the fully random case, 
    the framework makes it possible to analyze the situation without exploiting details
    of the hash function construction. Thus, it requires no knowledge of the 
    hash function and only expertise in random graph theory.
    
    Using this framework we showed that hash functions from class $\RR$ can be used in such diverse applications
    as cuckoo hashing (with a stash), generalized cuckoo hashing, the simulation
    of uniform hash functions, the construction of a perfect hash function, and
    load balancing. Particular choices for the parameters to set
    up hash functions from $\RR$ provide hash functions that can be evaluated efficiently. 
    
    We collect some pointers for future work. Our method is tightly connected
    to the first moment method.
    Unfortunately, some properties of random graphs cannot be proven
    using this method. For example, the classical proof that the
    connected components of the random graph $G(S, h_1, h_2)$ for $m
    = (1+\varepsilon) |S|$, for $\varepsilon > 0$, with fully random
    hash functions have size $O(\log n)$
    uses a Galton-Watson process (see, e.g., \cite{random_graphs}).
    From previous work \cite{DietzfelbingerH90,DietzfelbingerH92} 
    we know that hash class $\RR$ has some classical properties 
    regarding the balls-into-bins game. In the hypergraph setting this
    translates to a degree distribution of the vertices 
    close to the fully random 
    case. It would be very interesting to see if such a framework is also possible for
        other hash function constructions such as \cite{PatrascuT12,CelisRSW13}.
	The analysis of generalized cuckoo hashing could succeed (asymptotically)
        using hash functions from $\RR$. For this, one has to extend the analysis
        of the behavior of $\RR$ on small connected hypergraphs to connected 
        hypergraphs with super-logarithmically many edges.
        Witness trees are another approach to tackle the analysis of
            generalized cuckoo hashing. 
            We presented initial results in Section~\ref{sec:load_balancing}.
            It is open whether this approach yields good bounds
            on the space utilization of generalized cuckoo hashing.
        In light of the new constructions of Thorup \cite{Thorup13} and
        Christiani, Pagh, and Thorup \cite{ChristianiPT15}, it would
        be interesting to see whether or not highly-independent hash classes
        with constant evaluation time are efficient in practice.  Moreover, it 
        would be nice to prove 
        that hash class $\RR$ allows running linear probing robustly
        or to show that it is
        $\varepsilon$-minwise independent (for suitable
        $\varepsilon$).
    
\bibliographystyle{siam}
\bibliography{lit}

\appendix

\include{appendix}

\end{document}

%% file: sec-2-setup.tex
\section{Basic Setup and Groundwork}\label{hashing:sec:basics}
Let $U$ and $R$ be two finite sets with $1 < |R| \leq |U|$.
A \emph{hash function with range} $R$ is a mapping from $U$ to $R$. In our applications,
a hash function
is applied on some key set $S \subseteq U$ with $|S| = n$. Furthermore, the range of
the hash function is the set $[m] = \{0, \ldots, m - 1\}$ where often $m = \Theta(n)$.
In measuring space, we always assume that $\log |U|$ is a term so small
that it vanishes in big-Oh notation when compared with terms depending on $n$.
If this is not the case, one first applies
a hash function to collapse the universe to some size polynomial in $n$~\cite{Siegel04}.
We say that a pair $x,y \in U, x \neq y$ \emph{collides under a hash function} $g$ if $g(x)
= g(y)$.

The term \emph{universal hashing}, introduced by Carter and Wegman in
\cite{carter_universal_hashing}, refers to the technique of
choosing a hash function  at random from a \emph{hash class}
$\mathcal{H}_m \subseteq \{h \mid h\colon U \rightarrow [m]\}$.

\begin{definition}[\cite{carter_universal_hashing,CarterW79}]
    For a constant $c \geq 1$,
    a hash class $\HH$ with functions from $U$ to $[m]$ is called
$c$\emph{-universal} if for an arbitrary distinct pair of keys $x, y \in U$ we
have $$\Pr\nolimits_{h \in \HH}\bigl(h(x) = h(y)\bigr) \leq c/m.$$
\end{definition}%
In our constructions we will use $2$-universal classes of hash functions.
Examples for
$c$-universal hash classes can be found for example in
\cite{carter_universal_hashing,DietzfelbingerHKP97,Woelfel99}.
In the following, $\FF^c_m$
denotes an arbitrary $c$-universal hash class with domain $U$ and range $[m]$.

\begin{definition}[\cite{WegmanC79,WegmanC81}]
For an integer $\kappa\ge2$, a hash class $\HH$ with functions from $U$ to  $[m]$ is called a
$\kappa$\emph{-wise independent} hash class if for arbitrary distinct keys
$x_1,\ldots,x_\kappa\in U$ and for arbitrary $j_1,\ldots,j_\kappa\in [m]$ we have
$$\Pr\nolimits_{h\in\HH}\bigl(h(x_1)=j_1 \wedge \ldots \wedge h(x_\kappa)=j_\kappa\bigr) =
{1}/{m^\kappa}.$$
\end{definition}%
In other terms, choosing a hash function uniformly at random
from a $\kappa$-wise independent class of hash functions guarantees that the hash
values $h(x)$ are uniform in $[m]$ and that
that the hash values of an arbitrary set of at most
$\kappa$ keys are independent.  The classical construction of a $\kappa$-wise
independent hash class is based on polynomials of degree
$\kappa-1$ over a finite field~\cite{WegmanC79}. Another approach
is to use tabulation-based hashing, see \cite{ThorupZ12,KW2012a,PatrascuT12} for constructions using this approach.
Tabulation-based constructions
are often much faster  in
practice than polynomial-based hashing (\emph{cf.} \cite{ThorupZ12}) at the cost
of using slightly more memory. Throughout this work, $\HH^\kappa_m$ denotes an
arbitrary $\kappa$-wise independent hash class with domain $U$ and range $[m]$.

We remark that Section~\ref{sec:hash:class} and Section~\ref{sec:graph:properties} are quite natural generalizations of the basic definitions and observations made in \cite{AumullerDW14}
for pairs of hash functions. We give a full account for the convience of the reader and to provide consistent notation.

\subsection{The Hash Class}\label{sec:hash:class}
The hash class presented in this work draws ideas from many different papers. So, we first
give a detailed overview of related work and key concepts.

Building upon the work
on $k$-independent hash classes and two-level hashing strategies, e.g.,
the FKS-scheme of Fredman \emph{et al.} \cite{FredmanKS84},
Dietzfelbinger and Meyer
auf der Heide studied in \cite{DietzfelbingerH90,DietzfelbingerH92} randomness properties
of hash functions from $U$ to $[m]$ constructed in the following way: For given $k_1, k_2, m, n \geq 2$,
and $\delta$ with $0 < \delta < 1$, set $\ell = n^\delta$.
Let $f\colon U \to [m]$ be chosen from a $k_1$-wise independent hash class, and
let $g\colon U \to [\ell]$ be chosen from a $k_2$-wise independent hash class. Fill
a table $z[1..\ell]$ with random values from $[m]$. Given a key $x$, the hash function is evaluated
as follows:
\begin{align*}
    h(x) = f(x) + z[g(x)] \mod m.
\end{align*}
For $m = n$, the hash class of \cite{DietzfelbingerH90} had many randomness
properties that were only known to hold for fully random hash functions: When
throwing $n$ balls into $n$ bins, where each candidate bin is chosen by
``applying the hash function to the ball'', the expected maximum bin load is
$O(\log n/\log \log n)$, and conditioned on a ``good event'' that occurs
with probability $1-\frac{1}{\text{poly}(n)}$ the probability that a bin contains $i \geq 1$
balls decreases exponentially with $i$. Other explicit hash classes that share
this property were discovered by {\PAT} and Thorup \cite{PatrascuT12}
and Celis \emph{et al.} \cite{CelisRSW13} only about two decades later.

Our work studies randomness properties of the same hash class as
considered in Aumüller, Dietzfelbinger, and Woelfel \cite{AumullerDW14}.
For the convience of the reader we define this class
$\RR$ next. It is a generalization of the hash class proposed in \cite{DietzfelbingerH90}, modified so as to
obtain pairs $(h_1, h_2)$ of hash functions.
One could choose two $f$-functions (from a $k_1$-wise independent class), two $z$-tables, but \emph{only one}
$g$-function (from a $k_2$-wise independent class) that is shared among $h_1$ and $h_2$.
In~\cite{AumullerDW14}, this idea is further generalized so that
for a given $c \geq 1$ one uses $2c$ $z$-tables and $c$ $g$-functions.



We restrict the $f$-functions and the $g$-functions to be from very simple,
2-wise independent and 2-universal hash classes, respectively. This modification has two
effects: it simplifies the analysis and it seems to yield faster hash functions
in practice (see \cite[Section 7]{AumullerDW14}).

\begin{definition}
    Let $c\ge1$ and $d\ge2$.  For integers $m$, $\ell\ge 1$, and given
    $f_1,\ldots,f_d\colon U\to [m]$, $g_1,\ldots,g_c\colon U \to [\ell]$, and
    $d$ two-dimensional tables $z^{(i)}[1..c,0..\ell-1]$ with elements from $[m]$ for
    $i\in\{1,\ldots,d\}$, we let
    $\vec{h} = (h_1,\ldots,h_d) = (h_1,\ldots,h_d)\langle
    f_1, \allowbreak\ldots, f_d,g_1,\ldots,g_c,z^{(1)},\ldots, z^{(d)}\rangle$,
    where $$ {h_i(x) = \Bigl(f_i(x) +\sum_{1\le j \le c}
    z^{(i)}[j, g_j(x)]\Bigr) \bmod m\text{, for }x\in U, i\in\{1,\ldots,d\}.} $$

    \smallskip
    \noindent Let $\FF^2_\ell$ be an arbitrary two-universal class of hash functions from $U$ to $[\ell]$,
    and let $\HH^2_m$ be an arbitrary two-wise independent hash class from $U$ to $[m]$.
    Then $\RR^{c,d}_{\ell,m}(\FF^2_\ell, \GG^2_m)$ is the class of all sequences $(h_1,\ldots,h_d)\langle f_1, \ldots, f_d,
    g_1, \ldots, g_c, z^{(1)}, \ldots, z^{(d)}\rangle$ for $f_i \in \HH^2_m$ with $1 \leq i \leq d$ and
    $g_j \in \FF^2_\ell$ with $1 \leq j \leq c$.
    \label{def:family:R}
\end{definition}

If arbitrary
$\kappa$-wise independent hash classes are used as building blocks for the
functions $f_i$, for $1 \leq i \leq d$, and $g_j$, for $1 \leq j \leq c$,
one obtains the construction from~\cite{AumullerDW14}. However,
the simpler hash functions are much easier to deal with in the proofs of
this section. We defer the discussion of the general situation
to Section~\ref{hashing:sec:generalization}.

While this is not reflected in the notation, we consider $(h_1,\ldots,h_d)$ as a
structure from which
the components $g_1,\ldots,g_c$ and $f_i,z^{(i)}$, $i\in\{1,\ldots, d\}$, can be read off again.
It is class $\RR = \RR^{c,d}_{\ell,m}(\FF^2_\ell, \GG^2_m)$ for some $c \geq 1$ and $d \geq 2$, made
into a probability space by the uniform distribution, that we will study in the
following.
We usually assume that $c$ and $d$ are fixed and that $m$ and $\ell$ are known. Also,
the hash classes $\FF^2_\ell$ and $\GG^2_m$ are arbitrary (if providing
the necessary degree of universality or independence) and will not be mentioned explicitly below.

We will now discuss some randomness properties of hash class $\RR$. The central
lemma is identical to \cite[Lemma 1]{AumullerDW14}, but the proof is notably simpler because
of the restriction to simpler hash functions as building blocks. 

\begin{definition}
For $T \subseteq U$, define the random variable
$d_T$, the ``deficiency'' of $\vec{h} = (h_1,\ldots,h_d)$ with respect to $T$, by
$d_T(\vec{h})=|T| - \max\{|g_1(T)|,\ldots, |g_c(T)|\}$.
Further, let
\begin{enumerate}
    \item[\textnormal{(i)}] \mbox{\emph{$\text{bad}_T$} be the event that $d_T >
	1$\emph{;}}
    \item[{\textnormal{(ii)}}] \emph{$\text{good}_T$}{} be \emph{$\overline{\text{bad}_T}$},
	           i.e., the event that $d_T \le 1$\emph{;}
    \item[{\textnormal{(iii)}}] \emph{$\text{crit}_T$} be the event that $d_T =
	1$.
\end{enumerate}
Hash function sequences $(h_1,\ldots,h_d)$ in these events are called
\emph{``$T$-bad''}, \emph{``$T$-good''}, and \emph{``$T$-critical''}, respectively.\label{def:T:bad}
\end{definition}

It will turn out that if at least one of the functions $g_j$ is injective on a
set $T \subseteq U$,
then all hash values
on $T$
are independent. The deficiency $d_T$ of a sequence $\vec{h}$ of hash functions
measures how far away the hash function sequence is from this ``ideal'' situation.
 If $\vec{h}$ is $T$-bad,
then for each component $g_j$ there are at least two ``collisions'' on $T$,
i.e., there are at least two distinct pairs of keys from $T$ that collide. If $\vec{h}$ is
$T$-good, then there exists a $g_j$-component with at most one collision on $T$.
A hash function $\vec{h}$ is $T$-critical if for all
functions $g_j$ there is at least one collision and there exists at least one
function $g_j$ such that $g_j$ has exactly one collision on $T$. Note that the deficiency only
depends on the $g_j$-components of a hash function. In the following, we will first fix
these $g_j$-components when choosing a hash function. If $d_T(\vec{h}) \leq 1$ then the
unfixed parts of the hash function, i.e., the entries in the tables
$z^{(i)}$ and the $f$-functions,  are sufficient to guarantee strong randomness
properties of the hash function on $T$.

Our framework will build on the randomness properties of hash class $\RR$ that
are summarized in the next lemma. It comes in two parts. The first part
makes the role of the deficiency of a hash function sequence from $\RR$ precise,
as described above. The second part states that for a fixed set $T \subseteq S$
three parameters govern the probability of the events $\text{crit}_T$ or
$\text{bad}_T$ to occur: The size of $T$, the range $[\ell]$ of the
$g$-functions, and their number. To be precise, this probability is at most
$(|T|^2 / \ell)^c$, which yields two consequences.  When $|T|$ is much
smaller than $\ell$, the factor $1/\ell^c$ will make the probability of a hash
function behaving badly on a small key set vanishingly small. But when $|T|^2$ is
larger than $\ell$, the influence of the failure term of the hash class is
significant. We will see later how to tackle this problem.

\begin{lemma}
  Assume $d\ge2$ and $c\ge 1$.
  For $T\subseteq U$ the following holds\emph{:}
  \begin{enumerate}
      \item[\textnormal{(a)}] Conditioned on {$\textnormal{good}_T$} {(}or on {$\textnormal{crit}_T$}{)},
  the hash values  $(h_1(x), \ldots, h_d(x))$, $x\in T$, are
  distributed uniformly and independently in $[m]^d$.
      \item[\textnormal{(b)}] {$\Pr(\textnormal{bad}_T\cup\textnormal{crit}_T) \le \bigl(\left|T\right|^{2}/\ell\bigr)^c$}.
  \end{enumerate}\label{lem:random}
\end{lemma}

\begin{proof}
    Part (a): If $|T| \leq 2$, then $h_1,\dots,h_d$ are fully random on $T$
    simply because $f_1,\dots,f_d$ are drawn independently from a $2$-wise
    independent hash class. So suppose $|T| > 2$. First, fix an arbitrary
    $g$-part of $(h_1,\ldots,h_d)$ so that $\text{crit}_T$ occurs. (The
    statement follows analogously for $\text{good}_T$.) Let $j_0 \in
    \{1,\ldots,c\}$ be such that there occurs exactly one collision of keys in
    $T$ using $g_{j_0}$. Let $x,y \in T$, $x \neq y$, be this pair of keys (i.e.,
    $g_{j_0}(x) = g_{j_0}(y)$).  Arbitrarily fix all values in the tables
    $z^{(i)}[j, k]$ with $i \in \{1, \ldots, d\}$, $j \neq j_0$, and $0 \leq
    k \leq \ell - 1$. Furthermore, fix $z^{(i)}[j_0, g_{j_0}(x)]$ with $i \in
    \{1,\dots,d\}$.
    The hash functions $(h_1,\ldots,h_d)$ are fully random on $x$ and $y$ since $f_1,\dots,f_d$ are
    $2$-wise independent. Furthermore, the function $g_{j_0}$ is injective on
    $T - \{x, y\}$ and for each $x' \in (T - \{x, y\})$  the table
    cell $z^{(i)}[j_0, g_{j_0}(x')]$
    is yet unfixed, for $i \in \{1, \ldots, d\}$. Thus, the hash values
    $h_1(x'),\ldots,h_d(x')$, $x' \in T - \{x,y\}$, are distributed fully randomly
    and are independent of the hash values of $x$ and $y$.

    Part (b): Assume $|T| \geq 2$. (Otherwise the events $\text{crit}_T$
    or $\text{bad}_T$ cannot occur.) Suppose $\text{crit}_T$ (or $\text{bad}_T$)
    is true. Then for each component $g_i$, $1 \leq i \leq c$, there are keys
    $x,y \in T$, $x \neq y$, such that $g_i(x) = g_i(y)$. Since $g_i$ is chosen
    uniformly at random from a $2$-universal hash class, the probability
    that such a pair exists is at most $\binom{|T|}{2} \cdot 2/
    \ell\leq |T|^2/\ell$.
    Since all $g_i$-components are chosen independently, the statement follows.\qquad
\end{proof}

\subsection{Graph Properties and the Hash Class}\label{sec:graph:properties}
Here we describe how (hyper)graphs are built from a set of keys and a tuple
of hash functions. The notation is superficially different from \cite{AumullerDW14}, the
proof of the central lemma is identical.

We assume that the notion of a simple bipartite multigraph is known to the reader.
A nice introduction to graph theory is given by Diestel \cite{diestel}. We also
consider \emph{hypergraphs} $(V,E)$, which extend the notion of a graph by allowing edges
to consist of more than two vertices, i.e., the elements of $E$ are subsets of $V$
of size 2 or larger. For an integer $d \geq 2$,
a hypergraph is called \emph{$d$-uniform}
if each edge contains exactly $d$ vertices. It is called \emph{$d$-partite} if $V$ can
be split into $d$ sets $V_1, \ldots, V_d$ such that no edge contains
two vertices from the same class. A hypergraph $(V', E')$ is a \emph{subgraph}
of a hypergraph $(V, E)$ if $V' \subseteq V$ and if 
there is a one-to-one mapping $\varrho\colon E' \to E$ such that
$e'\subseteq \varrho(e')$ for each $e'\in E'$. 
(The reader should be aware that while this definition gives the usual
subgraphs for graphs, in the hypergraph setting it differs from
the standard notion of a subhypergraph. To obtain a subgraph of a hypergraph in our sense,
we may remove nodes and whole edges, but also delete nodes from single edges.) 
More notation for graphs and hypergraphs will be
provided in Section~\ref{hashing:sec:example:c:h} and
Section~\ref{hashing:sec:applications:hypergraphs}, respectively.

We build graphs and hypergraphs from a set of keys $S = \{x_1,\ldots,x_n\}$ and a sequence
of hash functions $\vec{h} = (h_1,\ldots,h_d)$, where $h_i\colon U \rightarrow [m]$ for $1\leq i \leq d$, in the following
way: The $d$-partite hypergraph $G(S, \vec{h}) = (V,E)$ has $d$ copies of $[m]$ as
vertex set and edge set $E = \{(h_1(x),\ldots,h_d(x)) \mid x \in S\}$.\footnote{In this paper,
whenever we refer to a graph or a hypergraph we mean a multi-graph or multi-hypergraph, i.e.,
the edge set is a multiset.
We also use the words ``graph'' and ``hypergraph'' synonymously in this section.
Finally, note that
our edges are tuples instead of sets to avoid problems with regard to the fact that
the hash functions use the same range. The tuple notation $(j_1, \ldots, j_d)$ for edges
is to be read as follows: $j_1$ is a vertex in the first copy of $[m]$, $\ldots$,
$j_d$ is a vertex in the $d$-th copy of $[m]$.}
Also, the edge
$(h_1(x_i),\ldots,h_d(x_i))$ is labeled ``$i$''.%
\footnote{We assume (w.l.o.g.) that the universe $U$ is ordered and that each set $S \subseteq
U$ of $n$ keys is represented as $S= \{x_1,\dots,x_n\}$ with $x_1 < x_2 < \dots
< x_n$.}   Since keys correspond to edges, the graph $G(S, \vec{h})$ has $n$ edges
and $d \cdot m$ vertices, which is the standard notation from a data structure point
of view, but is non-standard in graph theory. For a set $S$ and an
edge-labeled graph $G$, we let $T(G) = \{x_i \mid x_i \in S,
\text{$G$ contains an edge labeled $i$}\}$.

In the following, our main objective is to prove that with high probability certain subgraphs
do not occur in $G(S,\vec{h})$. Formally, for $n, m, d \in \mathbb{N}$, $d \geq 2$,
let $\mathcal{G}^d_{m, n}$ denote the set of all $d$-partite
hypergraphs with vertex set $[m]$ in each class of the
partition whose edges are labeled with distinct labels from $\{1,\ldots,n\}$.
A set $\mathsf{A} \subseteq \mathcal{G}^d_{m, n}$ is called a
\emph{graph property}. If for a graph $G$ we have that $G \in \mathsf{A}$, we say that
\emph{$G$ has property $\mathsf{A}$}. We shall always disregard isolated vertices.

For a key set $S$ of size $n$, a sequence $\vec{h}$ of
hash functions from $\RR$, and a graph property $\mathsf{A} \subseteq
\mathcal{G}^d_{m,n}$, we define the following random variables: For each $G \in \mathsf{A}$,
let $I_G$ be the indicator random variable that indicates whether $G$ is a subgraph
of $G(S, \vec{h})$ or not. (We demand the edge labels to coincide.) Furthermore, the
random variable $N^\mathsf{A}_S$ counts the number of graphs $G \in \mathsf{A}$ which
are subgraphs of $G(S, \vec{h})$, i.e., $N^{\mathsf{A}}_S = \sum_{G \in \mathsf{A}} I_G$.

Let $\mathsf{A}$ be a graph property.
Our main objective is then to estimate (from below) the probability that no subgraph
of $G(S,\vec{h})$ has property $\mathsf{A}$. Formally, for given $S \subseteq U$
we wish to bound (from above)
\begin{align}\label{eq:10}
    \Pr\nolimits_{\vec{h} \in \RR}\left(N^\mathsf{A}_S > 0\right).
\end{align}
In the analysis of a randomized algorithm, bounding \eqref{eq:10} is often
a classical application of the
\emph{first moment method}, which says that
\begin{align}\label{eq:11}
    \Pr\nolimits_{\vec{h} \in \RR} \left(N^\mathsf{A}_S > 0 \right) \leq
    \E_{\vec{h} \in \RR} \left(N^\mathsf{A}_S\right) = \sum_{G \in \mathsf{A}}
    \Pr\nolimits_{\vec{h} \in \RR}\left(I_G = 1\right).
\end{align}
However, we cannot apply the first moment method directly to bound
\eqref{eq:10}, since hash functions from $\RR$ do not guarantee full
independence on the key set, and thus the right-hand side of \eqref{eq:11} is
hard to calculate. However, we will
prove an interesting connection to the expected number of subgraphs
having property $\mathsf{A}$  when the hash function sequence $\vec{h}$ is
fully random.

To achieve this, we will start by collecting ``bad'' sequences of hash functions.
Intuitively, a sequence $\vec{h}$ of hash functions is \emph{bad} with respect
to a key set $S$ and a graph property $\mathsf{A}$ if $G(S, \vec{h})$ has a subgraph
$G$ with $G \in \mathsf{A}$ and for the keys $T \subseteq S$ which form $G$ the $g$-components
of $\vec{h}$ distribute $T$ ``badly''. (Recall the formal definition of ``bad''
from Definition~\ref{def:T:bad}.)
\begin{definition}
For $S \subseteq U$ and a graph property $\mathsf{A}$ let $B^{\mathsf{A}}_{S}
\subseteq \RR$  be the event
\begin{align*}
    \bigcup_{G \in \mathsf{A}} \left(\left\{I_G = 1\right\} \cap \mbox{\emph{bad}}_{T(G)}\right).
\end{align*}
\label{def:exception_set}
\end{definition}
This definition is slightly different from the corresponding definition in the paper
\cite[Definition 3]{AumullerDW14}, which considers one application of hash class $\RR$ with an application-specific
focus.%
\footnote{%
In \cite{AumullerDW14} we defined
$B^\mathsf{A}_S = \bigcup_{T \subseteq S} \left(\{G(T, \vec{h}) \text{ has property $A$}\} \cap
\text{bad}_T\right)$. This works well in the case that we only consider randomness properties
of the graph $G(S, h_1, h_2)$.
In the hypergraph setting,  ``important'' subgraphs of $G(S, \vec{h})$ often occur not in terms of
the graph $G(T, \vec{h})$, for some set $T \subseteq S$, but by removing some vertices from
the edges of $G(T, \vec{h})$. In Definition~\ref{def:exception_set},
we may consider exactly such subgraphs of $G(T, \vec{h})$ by defining $\mathsf{A}$ appropriately.
The edge labels of a graph are used to identify which keys of
$S$ form the graph.
}

In addition to the probability space $\RR$ together with the uniform
distribution, we also consider the probability space in which we use $d$ fully
random hash functions from $U$ to $[m]$, chosen independently.
From here on, we will denote probabilities of events and expectations of random variables in the
former case by $\Pr$ and $\E$; we will use ${\Pr}^\ast$ and
$\E^\ast$ in the latter. The next lemma shows that for
bounding $\Pr\left(N^\mathsf{A}_S > 0\right)$ we can use $\E^\ast\left(N^\mathsf{A}_S\right)$, i.e.,
the expected number of subgraphs having property $\mathsf{A}$  in the fully random case, and have to add
the probability that the event $B^\mathsf{A}_S$ occurs. We call this
additional summand the \emph{failure term of $\RR$ on $\mathsf{A}$}.


\begin{lemma} Let $S \subseteq U$ be given. For an arbitrary graph property $\mathsf{A}$
we have
\begin{equation}
    \Pr\left(N^\mathsf{A}_S > 0\right) \le \Pr\left(B^{\mathsf{A}}_{S}\right)
+\E^\ast\left(N^\mathsf{A}_S\right).
\label{eq:1000}
\end{equation}\label{lem:good:bad}
\end{lemma}
\begin{proof}
    We calculate:
    \begin{align*}
        \Pr\left(N^\mathsf{A}_S > 0\right) \leq \Pr\left(B^\mathsf{A}_S\right) +
        \Pr\left(\left\{N^\mathsf{A}_S > 0\right\} \cap \overline{B^{\mathsf{A}}_{S}}\right).
    \end{align*}
    We only have to focus on the second term on the right-hand side. Using the union bound,
    we continue as follows:
    \begin{align*}
        \Pr\left(\left\{N^\mathsf{A}_S > 0\right\} \cap \overline{B^{\mathsf{A}}_S}\right)
        &\leq \sum_{G \in \mathsf{A}} \Pr\left(\left\{I_G = 1\right\} \cap \overline{B^{\mathsf{A}}_S}\right)\\
        &= \sum_{G \in \mathsf{A}} \Pr\left(\left\{I_G = 1\right\} \cap \left(\bigcap_{G' \in \mathsf{A}}
        \left(\left\{I_{G'} = 0\right\} \cup \text{good}_{T(G')}\right)\right)\right)\\
        &\leq \sum_{G \in \mathsf{A}} \Pr\left(\left\{I_G = 1\right\} \cap \text{good}_{T(G)}\right)\\
        &\leq \sum_{G \in \mathsf{A}} \Pr\left(I_G = 1 \mid \text{good}_{T(G)}\right)\\
        &\stackrel{\text{(i)}}{=} \sum_{G \in \mathsf{A}} \Pr\nolimits^\ast\left(I_G = 1\right) = \E^\ast\left(N^\mathsf{A}_S\right),
    \end{align*}
where (i) holds by Lemma~\ref{lem:random}(b). \qquad
\end{proof}

This lemma provides the strategy for bounding $\Pr(N^\mathsf{A}_S > 0)$. The second summand in (\ref{eq:1000}) can be
calculated assuming full randomness and is often already known from the literature
if the original analysis was conducted using the first moment method.
The task of bounding the first summand is tackled separately in the next
subsection.

\subsection{A Framework for Bounding the Failure Term}
As we have seen, using hash class $\RR$ gives an additive failure term
(\emph{cf}.~\eqref{eq:1000}) compared to the case that we bound
${\Pr}^\ast\left(N^\mathsf{A}_S > 0\right)$ by the first moment method in the fully random
case. Calculating $\Pr\left(B^\mathsf{A}_S\right)$ looks difficult since we have to
calculate the probability that there exists a subgraph $G$ of $G(S, \vec{h})$
that has property $\mathsf{A}$ and where $\vec{h}$ is $T(G)$-bad. Since we know
the probability that $\vec{h}$ is $T(G)$-bad from Lemma~\ref{lem:random}(b), we
could tackle this task by calculating the probability that there exists such
a subgraph $G$ under the condition that $\vec{h}$ is $T(G)$-bad, but then we
cannot assume full randomness of $\vec{h}$ on $T(G)$ to obtain
a bound that a certain subgraph is realized by the hash values.
We avoid this difficulty by taking another approach. We will find suitable events that contain $B^\mathsf{A}_S$
and where $\vec{h}$ is guaranteed to behave well on the key set in question.

Observe the following relationship that is immediate from
Definition~\ref{def:exception_set}.

\begin{lemma}
    Let $S \subseteq U$, $|S| = n$, and let
    $\mathsf{A} \subseteq \mathsf{B} \subseteq \mathcal{G}^d_{m, n}$.
    Then $\Pr\left(B^\mathsf{A}_S\right) \leq
    \Pr\left(B^{\mathsf{B}}_S\right)$. \qed\label{lem:fail_superset}
\end{lemma}

We will now introduce two concepts that will allow us to bound
the failure probability of $\RR$ for ``suitable'' graph properties $\mathsf{A}$.

\begin{definition}[Peelability]
    A graph property $\mathsf{A}$ is called \textbf{peelable} if for all $G=(V,E)
    \in \mathsf{A}$ with $|E| \geq 1$ there exists an edge $e \in E$ such that $(V,E-\{e\}) \in \mathsf{A}$.
\label{def:peelability}
\end{definition}

An example for a peelable graph property for
bipartite graphs, i.e., in the case $d = 2$, is the set of all connected
bipartite graphs (disregarding isolated vertices), because removing an edge
that lies on a cycle or an edge incident to a vertex of degree $1$ does not
destroy connectivity.

Peelable graph properties will help us in the following sense: Assume that
$B^{\mathsf{A}}_S$ occurs, i.e., for the chosen $\vec{h} \in \RR$ there exists
some graph $G \in \mathsf{A}$ that is a subgraph of $G(S,\vec{h})$ and
$\vec{h}$ is $T(G)$-bad. Let $T = T(G)$.
In terms of the ``deficiency'' $d_T$ of
$\vec{h}$ (\emph{cf.}  Definition~\ref{def:T:bad}) it holds that $d_T(\vec{h}) > 1$.
If $\mathsf{A}$ is peelable, we can iteratively remove edges from $G$
such that the resulting graphs still have property $\mathsf{A}$. Let $G'$ be a graph
that results from $G$ by removing a single edge. Then $d_{T(G)} - d_{T(G')} \in \{0,1\}$.
Eventually, because $d_\emptyset = 0$, we will obtain a subgraph $G' \in \mathsf{A}$ of $G$ such that $\vec{h}$ is
$T(G')$-critical. In this case, we can again make use of Lemma~\ref{lem:random}(b) and
bound the probability that $G'$ is realized by the hash function sequence by assuming
that the hash values are fully random.

However, peelability does not suffice to obtain low
enough bounds for failure terms $\Pr\left(B^\mathsf{A}_S\right)$; we
need the following auxiliary concept, whose idea will become clear in the proof
of the next lemma.

\begin{definition}[Reducibility]
    Let $c \in \mathbb{N}$, and let $\mathsf{A}$ and $\mathsf{B}$ be graph properties.
$\mathsf{A}$ is called
\textbf{$\mathsf{B}$-$2c$-reducible} if for all graphs $(V,E) \in \mathsf{A}$ and sets $E^\ast \subseteq E$
with $|E^\ast| \leq 2c$ we have the following: There exists an edge set $E'$
with $E^\ast \subseteq E' \subseteq E$ such that $(V,E') \in \mathsf{B}$.
\label{def:reducible}
\end{definition}

If a graph property $\mathsf{A}$ is $\mathsf{B}$-$2c$-reducible, we say that
$\mathsf{A}$ \emph{reduces to} $\mathsf{B}$. The parameter $c$ shows the connection to hash class $\RR$:
it is the same parameter as the number of $g_j$-functions in hash class $\RR$.

 To shorten notation, we let
\begin{align*}
    \mu^\mathsf{\mathsf{A}}_t := \sum_{\substack{G \in \mathsf{A}, |E(G)| = t}} {\Pr}^\ast\left(I_G = 1\right)
    \end{align*}
    be the expected number of subgraphs with exactly $t$ edges having
    property $\mathsf{A}$ in the fully random case.
    The following lemma is the central result of this section and encapsulates our
overall strategy to bound the additive failure term introduced by using hash
class $\RR$ instead of fully random hash functions.
\begin{lemma}
    Let $c \geq 1$, $S\subseteq U$ with $|S| = n$, and let $\mathsf{A}$,
    $\mathsf{B}$, and $\mathsf{C}$ be graph properties such that
    $\mathsf{A} \subseteq \mathsf{B}$, $\mathsf{B}$ is a peelable graph
    property, and $\mathsf{B}$ reduces to $\mathsf{C}$. Then
    $$\Pr\left(B^{\mathsf{A}}_S\right) \leq \Pr\left(B^{\mathsf{B}}_S\right) \leq
    \ell^{-c} \cdot \sum_{t = 2}^{n} t^{2c} \cdot
    \mu^\mathsf{C}_t.$$
\label{lem:fail_prob}
\end{lemma}

\begin{proof}
    By Lemma~\ref{lem:fail_superset} we have $\Pr\left(B^\mathsf{A}_S\right) \leq
    \Pr(B^\mathsf{B}_S) = \Pr\left(\bigcup_{G \in \mathsf{B}} (\{I_G = 1\} \cap
    \text{bad}_{T(G)})\right)$. Assume that $\vec{h}$ is such that $B^\mathsf{B}_S$ occurs.
    Then there exists a subgraph $G$ of $G(S,\vec{h})$ such that
    $G \in \mathsf{B}$ and $d_{T(G)}(\vec{h}) > 1$. Fix such a graph.

    Since $\mathsf{B}$ is peelable, we iteratively remove edges from $G$ until we obtain
    a graph $G' = (V, E')$ such that $G' \in \mathsf{B}$ and
    $\text{crit}_{T(G')}$ occurs. The latter is guaranteed, since
    $d_\emptyset(\vec{h}) = 0$ and since for two graphs $G$ and $G'$, where $G'$
    results from $G$ by removing a single edge, it holds that $d_{T(G)}(\vec{h}) -
    d_{T(G')}(\vec{h}) \in \{0 ,1\}$. Since
    $\text{crit}_{T(G')}$ happens, for each $g_i$-component of $\vec{h}, 1 \leq i
    \leq c$, there is at least one collision on $T(G')$. Furthermore, there exists
    one component $g_{j_0}$ with $j_0 \in \{1,\dots,c\}$ such that exactly one
    collision on $T(G')$ occurs. For each $g_i$, $i \in \{1,\dots,c\}$, let
		$x_i$ and $y_i$ be two distinct keys such that $G'$ contains the edges
		$e_{x_i}$ and $e_{y_i}$ labeled with these keys and such that $x_i$ and $y_i$ collide under $g_i$.
    Let $E^\ast = \bigcup_{1 \leq i \leq c} \{e_{x_i},e_{y_i}\}$.

    By construction $|E^\ast| \leq 2c$. Since
    $\mathsf{B}$ reduces to $\mathsf{C}$, there exists some set $E''$
    with $E^\ast \subseteq E'' \subseteq E'$
    such that $G'' = (V, E'') \in \mathsf{C}$. By
    construction of $E^\ast$, each $g_i$-component has at least one collision on
    $T(G'')$. Moreover, $g_{j_0}$ has exactly one collision on $T(G'')$.
    Thus, $\vec{h}$ is $T(G'')$-critical.

    We calculate:
    \begin{align*}
	\Pr\left(B^\mathsf{A}_S\right) {\leq} \Pr\left(B^\mathsf{B}_S\right)
	&{=} \Pr\left(\bigcup_{G \in \mathsf{B}} \!\left(\left\{I_G = 1\right\} \cap
    \text{bad}_{T(G)}\right)\right)\!
    \stackrel{\text{(i)}}{\leq} \Pr\left(\bigcup_{G' \in \mathsf{B}} \left(\left\{I_{G'} = 1\right\} \cap
    \text{crit}_{T(G')}\right)\right)\\
    &\stackrel{\text{(ii)}}{\leq} \Pr\left(\bigcup_{G'' \in \mathsf{C}} \!\left(\left\{I_{G''} = 1\right\} \cap
    \text{crit}_{T(G'')}\right)\right) {\leq} \sum_{G'' \in \mathsf{C}} \!\Pr\left(\left\{I_{G''} = 1\right\} \cap
    \text{crit}_{T(G'')}\right)\\
    &= \sum_{G'' \in \mathsf{C}} \Pr\left(I_{G''} = 1\mid
\text{crit}_{T(G'')}\right) \cdot \Pr\left(\text{crit}_{T(G'')}\right)\\
&\stackrel{\text{(iii)}}{\leq}\ell^{-c} \cdot \sum_{\substack{G'' \in \mathsf{C}}}
{\Pr}^\ast\left(I_{G''} = 1\right) \cdot |T(G'')|^{2c}\\
&= \ell^{-c}\cdot \sum_{t = 2}^{n} \Bigl(t^{2c}\cdot\!\!\sum_{\substack{G'' \in
\mathsf{C}\\|E(G'')| = t}} {\Pr}^\ast\left(I_{G''} = 1\right)\Bigr) = \ell^{-c} \cdot
\sum_{t = 2}^{n} t^{2c} \cdot \mu^{\mathsf{C}}_t,
    \end{align*}
    where (i) holds because $\mathsf{B}$ is peelable, (ii) is due to reducibility, and (iii) follows by Lemma~\ref{lem:random}. \qquad
\end{proof}

We summarize the results of Lemma~\ref{lem:good:bad} and Lemma~\ref{lem:fail_prob} in the
following proposition.
\begin{proposition}
    Let $c \geq 1$, $m \geq 1$, $S\subseteq U$ with $|S| = n$, and let $\mathsf{A}$,
    $\mathsf{B}$, and $\mathsf{C}$ be graph properties such that
    $\mathsf{A} \subseteq \mathsf{B}$, $\mathsf{B}$ is a peelable graph
    property, and $\mathsf{B}$ reduces to $\mathsf{C}$. Assume that
    there are constants $\alpha$, $\beta$ such that
    \begin{align}
        \label{rr:bound:1}
        \E^\ast\left(N^\mathsf{A}_S\right) := \sum_{t = 1}^n \mu^\mathsf{A}_t = O\left(n^{-\alpha}\right),
    \end{align}
    and
    \begin{align}
        \label{rr:bound:2}
        \sum_{t = 2}^n t^{2c} \mu^\mathsf{C}_t = O\left(n^\beta\right).
    \end{align}
    Then setting $\ell = n^{(\alpha + \beta)/c}$ and choosing $\vec{h}$ at random from $\RR^{c,d}_{\ell,m}$ yields
    \begin{align*}
        \Pr\left(N^\mathsf{A}_S > 0\right) = O\left(n^{-\alpha}\right).
    \end{align*}
    \label{prop:rr:bound}
\end{proposition}
\begin{proof} The proposition follows immediately by plugging the failure probability bound from Lemma~\ref{lem:fail_prob}
into Lemma~\ref{lem:good:bad}. \qquad
\end{proof}

\begin{remark}
    In the statement of Lemma~\ref{lem:good:bad} and Proposition~\ref{prop:rr:bound} graph properties $\mathsf{B}$ and
    $\mathsf{C}$ can be the same graph properties, since every graph property reduces to itself.
    \label{remark:peelable}
\end{remark}

Proposition~\ref{prop:rr:bound} shows the power of our framework. The
conditions of this lemma can be checked without looking at the details of the
hash functions, only by finding suitable graph properties that have a low enough expected number of
subgraphs in the fully random case. Let us compare properties \eqref{rr:bound:1}
and \eqref{rr:bound:2}. Property \eqref{rr:bound:1} is the standard first moment
method approach. So, it can often be checked from the literature
whether a particular application seems suitable for an analysis with our
framework or not.
Property \eqref{rr:bound:2} seems very
close to a first moment method approach, but there is one important difference
to \eqref{rr:bound:1}. The additional factor $t^{2c}$, coming from the randomness properties
of the hash class, means that to obtain low enough bounds for \eqref{rr:bound:2},
the average number of graphs with property $\mathsf{C}$ must decrease rapidly, e.g.,
exponentially, fast in $t$. This will be the case for almost all graph properties considered in this thesis.

In the analysis for application scenarios, we will use Lemma~\ref{lem:good:bad}
and Lemma~\ref{lem:fail_prob} instead of Proposition~\ref{prop:rr:bound}. Often, one
auxiliary graph property suffices for many different applications and we think
it is cleaner to first bound the failure term of $\RR$ on this graph property
using Lemma~\ref{lem:fail_prob}; then we only have to care about the fully
random case and apply Lemma~\ref{lem:good:bad} at the end.

This concludes the development of the theoretical basis of this paper.

\subsection{Step by Step Example: Analyzing Static Cuckoo Hashing}
\label{hashing:sec:example:c:h}

We start by fixing graph-related notation for usual graphs:  We call an edge
that is incident to a vertex of degree 1 a \emph{leaf edge}. We call an edge a
\emph{cycle edge} if removing it does not disconnect any two nodes. A connected 
graph is called \emph{acyclic} if it does not contain cycles. It is called
\emph{unicyclic} if it contains exactly one cycle. 

Cuckoo hashing \cite{cuckoo_hashing_pagh} is a well known dictionary algorithm that stores
a (dynamically changing) set $S \subseteq U$ of size $n$ in two hash tables,
$T_1$ and $T_2$, each of size 
$m \geq (1 + \varepsilon)n$ for some $\varepsilon > 0$. 
It employs two hash functions $h_1$ and $h_2$ with $h_1, h_2\colon U \to [m]$. 
A key $x$ can be stored either in $T_1[h_1(x)]$ or in $T_2[h_2(x)]$, and all keys are stored in distinct table cells.
Thus, to find or remove a key it suffices to check these two possible locations. For
details on the insertion procedure we refer the reader to \cite{cuckoo_hashing_pagh}.

In this section, we deal with the static setting.
Here the question is whether or not a key set $S$ of 
size $n$ can be stored in the two tables of size $(1+\varepsilon) n$ each, for
some $\varepsilon > 0$, using a pair of hash functions $(h_1, h_2)$ according to 
the cuckoo hashing rules. To this end,
we look at the bipartite graph $G(S, h_1, h_2)$ built from $S$ and $(h_1, h_2)$. 
Recall that the vertices of $G$ are two copies of $[m]$ and that each key 
$x_i \in S$ gives rise to an edge $(h_1(x), h_2(x))$ labeled $i$. If $(h_1, h_2)$
allow storing $S$ according to the cuckoo hashing rules, i.e., independent
of the insertion algorithm, we call $(h_1,h_2)$ \emph{suitable} for $S$.


This section is meant as an introductory example for applying the framework. Already Pagh and Rodler
showed in \cite{cuckoo_hashing_pagh} that using a $\Theta(\log n)$-wise
independent hash class suffices to run cuckoo hashing. 
Since standard cuckoo hashing is a special 
case of cuckoo hashing with a stash, the results here can also be proven 
using the techniques presented in \cite{AumullerDW14}. However,
the proofs here are notably simpler than the proofs needed for the analysis
of cuckoo hashing with a stash. The central insight the reader should take away from the 
proof is that the analysis can be carried out by applying Lemma~\ref{lem:fail_prob} alone,
without reference to the inner structure of the hash functions. 

We will prove the following theorem:
\begin{theorem}
    Let $\varepsilon > 0$ and $0 < \delta <1$ be
given. Assume $c\ge  2/\delta$.  
    For $n\ge 1$ consider $m\ge(1+\varepsilon)n$ and $\ell=n^\delta$.  
    Let $S \subseteq U$ with $|S|=n$.
    Then for $(h_1,h_2)$ chosen at random from $\RR=\RR^{c,2}_{\ell,m}$ the
following holds:
    \begin{align*}
    \Pr\left( \left( h_1, h_2 \right) \text{ is not suitable for $S$} \right) = O(1/n).
\end{align*}
    \label{hashing:thm:cuckoo:hashing}
\end{theorem}
In the following, all statements of lemmas and claims use the parameter settings of Theorem~\ref{hashing:thm:cuckoo:hashing}.

It is not hard to see that $(h_1, h_2)$ is suitable for $S$ if and only if every
connected component of $G(S, h_1, h_2)$ has at most one cycle \cite{devroye}.
So, if $(h_1, h_2)$ is not suitable then $G(S,h_1,h_2)$ has a connected component with
more than one cycle. This motivates considering the following graph property. 
\begin{definition}
Let $\mathsf{MOG}$ (``minimal
obstruction graphs'') be the set of all labeled graphs from $\mathcal{G}^2_{m, n}$ (disregarding isolated
vertices) that 
form either a cycle with a chord or two cycles connected by a path of length $t \geq 0$.
\end{definition}

These two types of graphs form minimal connected graphs with more than one cycle, see
Figure~\ref{hashing:fig:obstruction:sets}.
So, if $(h_1, h_2)$ is not suitable for $S$, then $G(S,h_1,h_2)$
contains a subgraph with property $\text{MOG}$.
\begin{figure}[t]
    \centering
    \scalebox{0.9}{\includegraphics{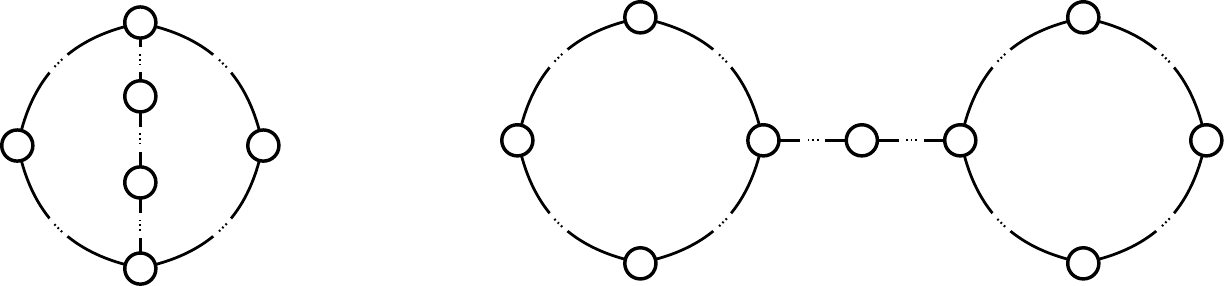}}
    \caption{The minimal obstruction graphs for cuckoo hashing, see also \cite{devroye}.}
    \label{hashing:fig:obstruction:sets}
\end{figure}
We summarize: 
\begin{align}
    \Pr\left( \left(h_1, h_2\right) \text{ is  not suitable for $S$}\right) = \Pr\left(N^\mathsf{MOG}_S > 0\right).
    \label{hashing:eq:cuckoo:hashing:1}
\end{align}
According to Lemma~\ref{lem:good:bad}, the probability on the right-hand side of \eqref{hashing:eq:cuckoo:hashing:1}
is at most
\begin{align}
    \Pr\left(N^\mathsf{MOG}_S > 0\right) \leq
    \Pr\left(B^\mathsf{MOG}_S\right)
    + \E\phantom{}^{\ast}\left(N^\mathsf{MOG}_S\right).
    \label{hashing:eq:cuckoo:hashing:2}
\end{align}
We first study the expected number of minimal obstruction graphs in the fully random case. 

\paragraph{Bounding $\E^{\ast}\left(N^\mathsf{MOG}_S\right)$}
The expected number of minimal obstruction graphs in the fully random case is well known from other work, 
see, e.g., \cite{cuckoo_hashing_pagh}. The proof is included for the convenience of the reader and follows
\cite{cuckoo_hashing_pagh}.
\begin{lemma}
    \begin{align*}
        \E^\ast\left(N^\mathsf{MOG}_S\right) = O(1/m).
    \end{align*}
    \label{hashing:lem:cuckoo:hashing:fully:random}
\end{lemma}

\begin{proof} We start by counting unlabeled graphs with exactly $t$ edges that
    form a minimal obstruction graph.  Every minimal obstruction graphs consists
    of a simple path of exactly $t - 2$ edges and two further edges which connect the
    endpoints of this path with vertices on the path. Since a minimal
    obstruction graph with $t$ edges has exactly $t - 1$ vertices, there are no
    more than $(t-1)^2$ unlabeled minimal obstruction graphs having exactly $t$
    edges. Fix an unlabeled minimal obstruction graph $G$. First, there are two
    ways to split the vertices of $G$ into the two parts of the bipartition.
    When this is fixed, there are no more than $m^{t-1}$ ways to label the
    vertices with labels from $[m]$, and there are no more than $n^{t+1}$ ways
    to label the edges with labels from $\{1, \ldots, n\}$. Fix such a fully
    labeled graph $G'$. 

    Now draw $t$ labeled edges\footnote{The labels of these edges
    are equivalent to the edge labels of $G'$.} at random from $[m]^2$.
    The probability that these edges realize $G'$ is exactly $1/m^{2t}$.
    We calculate: 

\begin{align*}
    \E^\ast\left(N^\mathsf{MOG}_S\right) &\leq \sum_{t=3}^n
    \frac{2n^{t} \cdot m^{t - 1} \cdot (t-1)^2}{m^{2t}} 
    \leq \frac{2}{m}\cdot\sum_{t = 3}^n \frac{t^2 n^{t}}{m^{t}} =\frac{2}{m}\cdot
    \sum_{t=3}^n  \frac{t^2}{(1 + \varepsilon)^t}
    =O\Bigl{(}\frac{1}{m}\Bigr{)},
\end{align*}
where the last step follows from the convergence of the series $\sum_{t =
0}^{\infty} t^2/q^t$ for every $q > 1$. \qquad
\end{proof}

\noindent Combining Lemma~\ref{hashing:lem:cuckoo:hashing:fully:random} and 
\eqref{hashing:eq:cuckoo:hashing:2} we obtain:
\begin{align}
    \Pr\left(N^\mathsf{MOG}_S > 0\right) \leq
    \Pr\left(B^\mathsf{MOG}_S\right)
    + O\left(\frac{1}{m}\right). 
\end{align}
It remains to bound the failure term $\Pr\left(B^\mathsf{MOG}_S\right)$.

\paragraph{Bounding $\Pr\left(B^\mathsf{MOG}_S\right)$}
In the light of Definition~\ref{def:peelability}, we first note that $\mathsf{MOG}$ is not peelable. So, we first
find a peelable graph property that contains $\mathsf{MOG}$. Since paths are peelable, and a minimal
obstruction graph is ``almost path-like'' (\emph{cf.} proof of Lemma~\ref{hashing:lem:cuckoo:hashing:fully:random}),
we relax the notion of a minimal obstruction graph in the following way.

\begin{definition}
    Let $\mathsf{RMOG}$ (``relaxed minimal obstruction graphs'') consist of all graphs in $\GG^2_{m,n}$ that form
    either (i) a minimal obstruction graph, (ii) a simple path, or (iii) a simple path and exactly one edge which connects an
    endpoint of the path with a vertex on the path. (We disregard isolated vertices.)
    \label{hashing:def:rmog}
\end{definition}

\noindent By the definition, we obviously have that $\mathsf{MOG} \subseteq \mathsf{RMOG}$.
\begin{lemma}
    $\mathsf{RMOG}$ is peelable.
\end{lemma}
\begin{proof}
    Let $G \in \mathsf{RMOG}$. We may assume that $G$ has at least two edges.
    We distinguish three cases:

    \noindent Case 1: $G$ is a minimal obstruction graph. Let $G'$ be the graph that results from 
    $G$ when we remove an arbitrary cycle edge incident to a vertex of degree $3$ 
		or degree $4$ in $G$. Then
    $G'$ has property (iii) of Definition~\ref{hashing:def:rmog}.

    \noindent Case 2: $G$ has property (iii) of Definition~\ref{hashing:def:rmog}. 
    Let $G'$ be the graph that
    results from $G$ when we remove an edge in the following way: If $G$
    contains a vertex of degree $3$ then remove an arbitrary cycle edge incident
    to this vertex of degree $3$, otherwise remove an arbitrary cycle edge. Then
    $G'$ is a path and thus has property (ii) of
    Definition~\ref{hashing:def:rmog}. 

    \noindent Case 3: $G$ is a simple path. Let $G'$ be the graph that results from $G$ when we remove
    an endpoint of $G$ with the incident edge. $G'$ is a path and has property (ii) of Definition~\ref{hashing:def:rmog}. \qquad
\end{proof}

Standard cuckoo hashing is an example where we do not need every
component of our framework, because there are ``few enough'' 
graphs having property $\mathsf{RMOG}$ to obtain low enough failure probabilities. 

\begin{lemma}
\begin{align*}
    \Pr\left(B^\mathsf{MOG}_S\right) = O\left(\frac{n}{\ell^c}\right).
\end{align*}
\label{hashing:lem:mog:fail:prob}
\end{lemma}

\emph{Proof}.
    We aim to apply Lemma~\ref{lem:fail_prob}, where $\mathsf{MOG}$ takes the role of $\mathsf{A}$ and $\mathsf{RMOG}$ takes
    the role of $\mathsf{B}$ and $\mathsf{C}$ (\emph{cf.} Remark~\ref{remark:peelable}), respectively,
    in the statement of that lemma.

    \begin{claim}
	For $t \geq 2$, we have 
        \begin{align*}
            \mu^{\mathsf{RMOG}}_t \leq \frac{6mt^2}{(1 + \varepsilon)^t}.
        \end{align*}
	\label{hashing:claim:rmog}
    \end{claim}
    \emph{Proof}.
        We first count labeled graphs with exactly $t$ edges having property
        $\mathsf{RMOG}$.  From the proof of
        Lemma~\ref{hashing:lem:cuckoo:hashing:fully:random} we know that there
        are fewer than $2 \cdot t^2 \cdot n^t \cdot m^{t - 1}$ labeled graphs
        that form minimal obstruction graphs ((i) of
        Def.~\ref{hashing:def:rmog}). Similarly, there are not more than $2
        \cdot n^t \cdot m^{t+1}$ labeled paths ((ii) of
        Def.~\ref{hashing:def:rmog}), and not more than $2 \cdot t \cdot n^t \cdot m^t$
        graphs having property (iii) of Def.~\ref{hashing:def:rmog}. Fix a labeled graph $G$
        with property $\mathsf{RMOG}$ having exactly $t$ edges. Draw $t$ labeled edges at random from
        $[m]^2$. The probability that these $t$ edges realize $G$ is exactly $1/m^{2t}$. We calculate:
        \begin{align*}
            \mu^\mathsf{RMOG}_t \leq \frac{6t^2n^t m^{t+1}}{m^{2t}} = \frac{6mt^2}{(1+\varepsilon)^t} \; . \qquad \endproof
        \end{align*}

    Using Lemma~\ref{lem:fail_prob}, we proceed as follows:
    \begin{align*}
        \Pr\left(B^\mathsf{MOG}_S\right) \leq \ell^{-c} \cdot \sum_{t = 2}^{n}t^{2c} \cdot \mu^{\mathsf{RMOG}}_t \leq 
        \ell^{-c} \cdot \sum_{t=2}^{n}\frac{6mt^{2(c+1)}}{(1+\varepsilon)^t} = O\left(\frac{n}{\ell^{c}}\right).\qquad \endproof
    \end{align*}

\paragraph{Putting Everything Together} Plugging the results of Lemma~\ref{hashing:lem:cuckoo:hashing:fully:random}
and Lemma~\ref{hashing:lem:mog:fail:prob} into \eqref{hashing:eq:cuckoo:hashing:2} gives:

\begin{align*}
    \Pr\left(N^\mathsf{MOG}_S > 0\right) \leq
    \Pr\left(B^\mathsf{MOG}_S\right)
    + \E\phantom{}^{\ast}\left(N^\mathsf{MOG}_S\right) = O\left(\frac{n}{\ell^c}\right) + O\left(\frac1m\right).
\end{align*}
Using that $m = (1+ \varepsilon) n$ and setting $\ell = n^\delta$ and $c \geq
2/\delta$ yields Theorem~\ref{hashing:thm:cuckoo:hashing}.  

This example gives insight into the situation in which our
framework can be applied. The graph property under consideration
($\mathsf{MOG}$)  is such that the expected number of subgraphs with
this property is polynomially small in $n$. The peeling
process, however, yields graphs which are much more likely to occur, e.g.,
paths of a given length. The key in our analysis is finding suitable graph
properties of ``small enough'' size. (That is the reason why the concept of
``reducibility'' from Definition~\ref{def:reducible} is needed in other
applications: It makes the
number of graphs that must be considered smaller.) The $g$-components of the hash functions
from $\RR$ provide a boost of $\ell^{-c}$, which is then used to make the
overall failure term again polynomially small in $n$. 

The reader might find it instructive  
to apply Proposition~\ref{prop:rr:bound} directly. Then, graph property
$\mathsf{MOG}$ plays the role of graph property $\mathsf{A}$ in that proposition;
graph property $\mathsf{RMOG}$ plays the role of $\mathsf{B}$ \emph{and}
$\mathsf{C}$.

%% file: sec-3-graphs.tex
\section{Applying the Framework to Graphs}
\label{hashing:sec:applications:simple:graphs}

In this section, we will study different applications of our hash class in algorithms and data
structures whose analysis relies on properties of the graph $G(S,
h_1, h_2)$. We shall study four different applications: 
\begin{itemize}
    \item A variant of cuckoo hashing
called \emph{cuckoo hashing with a stash} introduced by Kirsch, Mitzenmacher, and Wieder in
\cite{stash}.
    \item A construction for the simulation of a uniform hash function due to
        Pagh and Pagh \cite{pagh_uniform}.
    \item  A construction of a (minimal) perfect hash function as described by Botelho, Pagh,
        and Ziviani
        \cite{BotelhoPZ13}.
    \item The randomness properties of hash class $\RR$ on connected components
        of $G(S, h_1, h_2)$.
\end{itemize}
We start by studying the failure term of the hash class on a graph property that plays a 
central role in the applications.

\subsection{Randomness Properties of Leaf\/less Graphs}\label{hashing:sec:leafless}
    In this section we study the additive failure term of hash functions from $\RR$ 
    on a graph property that will be a key ingredient in the
	applications to follow. First,
    we recall some graph notation and present a counting argument from \cite{AumullerDW14}.  
    Subsequently, we study the failure term of $\RR$ on the class of 
    graphs which contain no leaf edges, so-called ``leaf\/less graphs''.
    
    Leaf\/less graphs are at the core of the analysis of many randomized
    algorithms and data structures, such as cuckoo hashing (note that the
    minimal obstruction graphs from Figure~\ref{hashing:fig:obstruction:sets}
    have no leaves), the simulation of a uniform hash function as described by
    Pagh and Pagh in \cite{pagh_uniform}, and the construction of a perfect
    hash function from Botelho, Pagh, and Ziviani \cite{BotelhoPZ13}. As we
    shall demonstrate in the subsequent sections, analyzing the case that we replace
    fully random hash functions by hash functions from $\RR$ in these
    applications becomes easy when the behavior of the additive failure term of
    $\RR$ on leaf\/less graphs is known. The main result of this section says
    that the additive failure term of hash class $\RR^{c,2}_{\ell, m}$ on
    leaf\/less graphs is $O(n/\ell^c)$, and can thus be made as small as
    $O(n^{-\alpha})$ for $\ell = n^\delta$, $0 < \delta < 1$, and $c =
    \Theta(\alpha)$. On the way, we will use all steps of our framework
    developed in Section~\ref{hashing:sec:basics}. 

    We recall some graph notation. The cyclomatic number $\gamma(G)$ is the dimension of the \emph{cycle space} of a graph $G$.
It is equal to the smallest number of edges we have to remove from $G$ such that
the remaining 
graph is a forest (an acyclic, possibly disconnected graph) \cite{diestel}.
Also, let $\zeta(G)$ denote the number of connected components of $G$ (ignoring
isolated vertices).
\begin{definition}
Let $N(t,\ell, \gamma, \zeta)$ be the number of unlabeled (multi-)graphs with
$\zeta$ connected components and cyclomatic number $\gamma$ that have $t - \ell$ inner edges and 
$\ell$ leaf edges.
\end{definition}

The following bound is central in our analysis; it is taken from \cite[Lemma 4]{AumullerDW14}.
\begin{lemma}
$N(t,\ell,\gamma,\zeta) = t^{O(\ell + \gamma + \zeta)}.$
\label{lem:num_graphs}
\end{lemma}

    We let $\LL \subseteq \GG^2_{m,n}$ consist of all bipartite graphs
    that contain no leaf edge.  
    It will turn out that for all our applications $\LL$ will 
    be a suitable ``intermediate'' graph property, i.e., for the graph
    property $\mathsf{A}$ interesting for the application it will hold
    $\mathsf{A} \subseteq \LL$, which will allow us to apply
    Lemma~\ref{lem:fail_superset}. (For example, graph property $\LL$ could have 
	been used instead of graph property $\mathsf{RMOG}$ 
    in the example of the previous section.) Hence our goal in this section is to show
    that there exists a constant
    $\alpha > 0$, which depends on the parameters $\ell$ and $c$ of the hash
    class $\RR^{c,2}_{\ell,m}$, such that 
    \begin{align*}
	\Pr\nolimits_{(h_1, h_2) \in \RR}\left(B^\LL_S \right) =
	O\left(n^{-\alpha}\right).
    \end{align*}
    Luckily, bounding $\Pr\left(B^\LL_S\right)$ is an example par
    excellence for applying Lemma~\ref{lem:fail_prob}. To use this 
    lemma we have to find a suitable peelable graph property (note that 
    $\LL$ is not peelable) and a suitable further graph property to which 
    that graph property reduces.

    We let $\LC$ consist of all graphs $G$ from $\GG^2_{m,n}$ that contain
    at most one connected component that has leaves, disregarding isolated vertices. 
    If such a component exists, we call it the 
    \emph{leaf component} of $G$.

    \begin{lemma}
        $\LC$ is peelable.
    \end{lemma}

    \begin{proof}
        Suppose $G \in \LC$ has at least one edge. If $G$ has no leaf
        component then all edges are cycle edges, and removing an arbitrary
        edge leaves a cycle component or creates a leaf component. So, the resulting graph has
        property $\LC$. If $G$ has a leaf component $C$, remove a 
        leaf edge. This makes the component smaller, but maintains 
        property $\LC$.  So, the
    resulting graph has again property $\LC$. \qquad
    \end{proof}
 
    We will also need the following auxiliary graph property:

\begin{definition}
    For $K \in \mathbb{N}$, let $\LCY{K} \subseteq \GG^2_{m,n}$ be the set of all
    bipartite graphs $G=(V,E)$ with the following properties (disregarding isolated vertices):
    \begin{enumerate} 
	\item at most one connected component of $G$ contains leaves (i.e., $\LCY{K}\subseteq \LC$);
	\item the number $\zeta(G)$ of connected components is bounded by $K$;
	\item if present, the leaf component of $G$ contains at most $K$ leaf and cycle edges;
	\item the cyclomatic number $\gamma(G)$ is bounded by $K$.
    \end{enumerate}
\end{definition}

\begin{lemma}
    Let $c \geq 1$. Then $\LC$ is $\LCY{4c}$-$2c$-reducible.
\label{lem:lcy:reducibility}
\end{lemma}

\begin{proof}
    Consider an arbitrary graph $G = (V,E) \in \LC$ and an arbitrary edge set 
    $E^\ast \subseteq E$ with $|E^\ast| \leq 2c$. We say that an edge that belongs to $E^\ast$ is \emph{marked}.
    $G$ satisfies
    Property $1$ of graphs from $\LCY{4c}$. We process $G$ in three stages:

\emph{Stage} 1: Remove all components of $G$ without marked edges. Afterwards
at most $2c$ components are left, and $G$ satisfies Property $2$.

\emph{Stage} 2: If $G$ has a leaf
component $C$, repeatedly remove unmarked leaf and cycle edges 
from $C$, while $C$ has such edges. The remaining leaf and cycle edges in $C$ are marked, 
and thus their number is at most $2c$; Property 3 is satisfied.

\emph{Stage} 3: If there is a leaf component $C$ with $z$ marked edges (where $z\le2c$),
then at least one of them is a leaf edge, and hence
$\gamma(C) \leq z - 1$. Now consider an arbitrary leaf\/less component $C'$ with cyclomatic number $z$.
We construct a suitable subgraph $C''$ of $C'$. For this, 
we need the following graph theoretic claim:
\begin{claim}\label{claim:leafless:reduce}
Every leaf{\kern0.5pt}less connected graph with $i$ marked edges has a leaf{\kern0.5pt}less connected subgraph with 
cyclomatic number $\le i {+} 1$ that contains all marked edges. 
\label{lem:cyclic_cyclo_peeling}
\end{claim}

\begin{proof}
Let $G=(V,E)$ be a leaf\/less connected graph with $i$ marked edges. If $\gamma(G) \le i+1$, there is nothing to prove. 
So suppose $\gamma(G) \ge i+2$. Choose an arbitrary spanning
tree $(V,E_0)$ of $G$.

There are two types of edges in $G$:
\emph{bridge edges} and \emph{cycle edges}. A bridge edge is an edge whose
deletion disconnects the graph,
cycle edges are those whose deletion does not disconnect the graph.

Clearly, all bridge edges are in $E_0$.  
Let $E_{\text{mb}}\subseteq E_0$ denote the set of marked bridge edges. 
Removing the edges of $E_{\text{mb}}$ from $G$ splits $V$ into $|E_{\text{mb}}|+1$ connected components  $V_1,\ldots,V_{|E_\text{mb}|+1}$;
removing the edges of $E_{\text{mb}}$ from the spanning tree $(V,E_0)$ will give exactly the same components.
For each \emph{cyclic} component $V_j$ we choose one edge $e_j\notin E_0$ that connects two nodes in $V_j$.
The set of these $|E_\text{mb}|+1$ edges is called $E_1$. Now each marked
bridge edge lies on a path connecting two cycles in $(V,E_0 \cup E_1)$.

Recall from graph theory~\cite{diestel} the notion of a fundamental cycle:
Clearly, each edge $e\in E-E_0$ closes a unique cycle with $E_0$.
The cycles thus obtained are called the fundamental cycles of $G$ w.\,r.\,t. the spanning tree $(V,E_0)$.
Each cycle in $G$ can be obtained as an XOR-combination of fundamental cycles.
(This is just another formulation of the standard fact that the fundamental cycles form a basis
of the ``cycle space'' of $G$, see~\cite{diestel}.)
From this it is immediate that every cycle edge of $G$ lies on some fundamental cycle. 
Now we associate an edge $e'\notin E_0$ with each marked cycle edge $e\in E_{\text{mc}}$.
Given $e$, let
$e'\notin E_0$ be such that $e$ is on the fundamental cycle of $e'$.
Let $E_2$ be the set of all edges $e'$ chosen in this way. Clearly, each $e \in
E_{\text{mc}}$ is a cycle edge in $(V,E_0\cup E_2)$.

Now let $G' =(V,E_0 \cup E_1 \cup E_2)$. Note that
$|E_1\cup E_2| \leq(|E_\text{mb}| + 1) +|E_\text{mc}| \le i+1$ and thus $\gamma(G') \leq i + 1$. 
In $G'$, each marked edge is on a cycle or on a path that connects two cycles. 
If we iteratively remove leaf edges from $G'$ until no leaf is left, none of the marked edges will be affected. 
Thus, we obtain the desired leaf\/less subgraph $G^\ast$ with $\gamma(G^\ast)=\gamma(G')\le i+1$. \qquad
\end{proof}

Claim~\ref{claim:leafless:reduce} gives us a leaf\/less subgraph $C''$ of our leafless component $C'$ 
with $\gamma(C'') \leq z+1$ that contains all marked edges of $C'$. 
We remove from $G$ all vertices and edges of $C'$ that are not in $C''$.
Doing this for all leaf\/less components
yields the final graph $G$. 
Summing contributions to the cyclomatic number of $G$
over all (at most $2c$) components, we see that $\gamma(G) \leq 4c$;
thus Property 4 is satisfied. \qquad 
\end{proof}

\noindent We now bound the additive failure term $\Pr\left(B^{\LL}_S\right)$.

\begin{lemma}
    Let $S \subseteq U$ with $|S| = n,$ $\varepsilon > 0,$ $c \geq
    1$, and let $\ell \geq 1$. Assume $m \geq (1
    + \varepsilon) n$. If $(h_1, h_2)$ are chosen at random from $\RR^{c,2}_{\ell,m}$, then
    \begin{align*}
        \Pr\left(B^\LL_S\right) \leq \Pr\left(B^{\LC}_S\right) = O\left(n/\ell^c\right).
    \end{align*}
\label{lem:failure_term_bound}
\end{lemma}

\emph{Proof.}
    According to Lemma~\ref{lem:fail_prob} and Lemma~\ref{lem:lcy:reducibility} it holds that
    \begin{align*}
	\Pr\left(B^{\LL}_S\right) \leq \Pr\left(B^{\LC}_S\right) \leq \ell^{-c}\cdot\sum_{t = 2}^n t^{2c} \cdot  \mu^{\LCY{4c}}_t.
    \end{align*}

    \begin{claim}
        \begin{align*}
            \mu^{\LCY{4c}}_t = \frac{2 n \cdot t^{O(1)}}{(1+ \varepsilon)^{t-1}}.
        \end{align*}
    \end{claim}

\emph{Proof.} 
By Lemma~\ref{lem:num_graphs}, there are at most $t^{O(c)} =
t^{O(1)}$ ways to choose a bipartite graph $G$ in $\LCY{4c}$ with $t$ edges. 
Graph $G$ cannot have more than $t+1$ nodes, since cyclic
components have at most as many nodes as edges, and in the single leaf component, if 
present, the number of nodes is at most one bigger than the number of edges.  
In each component of $G$, there are two ways to assign the vertices to the two sides of the
bipartition. After such an assignment is fixed, there are at most $m^{t+1}$ ways
to label the vertices with elements of $[m]$, and there are not more than $n^t$ ways
to label the $t$ edges of $G$ with labels from $\{1,\ldots,n\}$. Assume now such labels have been chosen for $G$.
Draw $t$ labeled edges according to the labeling of $G$ from $[m]^2$ uniformly at random.
The probability that they exactly fit the labeling of nodes and edges of $G$ is $1/m^{2t}$.  
Thus, 
\begin{align*}
    \mu^{\LCY{4c}}_t \leq \frac{2\cdot m^{t+1}\cdot n^t \cdot t^{O(1)}}{m^{2t}} \leq \frac{2n \cdot
    t^{O(1)}}{(1+\varepsilon)^{t-1}}.\qquad\endproof
\end{align*}
We use this claim to finish the proof of Lemma~\ref{lem:failure_term_bound} by the 
following calculation: 
\begin{align*}
    \Pr(B^\LC_S) &\leq \ell^{-c} 
    \sum_{t = 2}^n  t^{2c} \cdot \mu^{\LCY{4c}}_t
     \leq \frac{2n}{\ell^{c}} \cdot \sum_{t = 2}^n \frac{t^{O(1)} }{ (1
     + \varepsilon)^{t-1} } = O\left(\frac{n}{\ell^{c}}\right).\qquad\endproof
\end{align*}
    
We now turn our focus to cuckoo hashing with a stash. We reprove a result
from \cite{AumullerDW14} to demonstrate the power of the framework. 

\subsection{Cuckoo Hashing (with a Stash)}
Kirsch, Mitzenmacher, and Wieder \cite{stash} proposed augmenting the 
cuckoo hashing tables with a \emph{stash}, an
additional segment
of storage that can hold up to $s$ keys for some (constant) parameter $s$. They
showed that using a stash of size $s$ reduces the rehash probability to
${\rmTheta}(1/n^{s+1})$. For details of the algorithm, see \cite{stash:journal:09}.
 
We focus on the question whether the pair $(h_1,h_2)$ allows storing the key set $S$
in the two tables with a stash of size $s$.  
It is known from \cite{stash:journal:09,AumullerDW14} that a single parameter of $G=G(S,h_1,h_2)$
determines whether a stash of size $s$ 
is sufficient to store $S$ using $(h_1,h_2)$, namely 
the \emph{excess} $\text{ex}(G)$.
    The excess $\ex(G)$ of a graph $G$ is defined as the minimum number of edges
one has to remove from $G$ so that all connected components 
of the remaining graph are acyclic or unicyclic. In \cite{stash:journal:09} it is shown 
that the excess of a graph $G=(V,E)$ is $\ex(G) = \gamma(G) -
\zeta_{\textrm{cyc}}(G),$ where $\zeta_{\textrm{cyc}}(G)$ is the number of
cyclic connected components in $G$. The connection between the excess of a graph 
and the failure probability of cuckoo hashing with a stash is that $(h_1, h_2)$ 
are suitable for a key set $S$ if and only
if \emph{$\textnormal{ex}(G(S,h_1,h_2))\le s$}. 

The following theorem shows that one can replace the 
full randomness assumption of \cite{stash:journal:09} by hash functions
from hash class $\RR$. 

\begin{theorem}[\cite{AumullerDW14}]
    Let $\varepsilon > 0$ and $0 < \delta <1$, let  $s\ge0$  be
given. Assume $c\ge  (s+2)/\delta$.  
    For $n\ge 1$ consider $m\ge(1+\varepsilon)n$ and $\ell=n^\delta$.  
    Let $S \subseteq U$ with $|S|=n$.
    Then for $(h_1,h_2)$ chosen at random from $\RR=\RR^{c,2}_{\ell,m}$ the
following holds:
    $$\Pr(\ex(G(S,h_1,h_2)) \geq s+1) = O(1/n^{s+1}).$$
    \label{thm:r_excess_prob_zz}\samepage
\end{theorem}
\begin{proof}
As in \cite{AumullerDW14}, we define an \emph{excess-$(s+1)$ core graph} as 
a leaf\/less graph $G$ with excess
exactly $s + 1$ in which all connected components have at least two cycles.
By $\CS{s+1}$ we denote the set of all excess-$(s+1)$ core graphs in
$\mathcal{G}^2_{m, n}$. (For an illustration, see Figure~1 in \cite{AumullerDW14}.)

From \cite[Lemma 6]{AumullerDW14} we know that each $G = G(S, h_1, h_2)$ with $\ex(G) \geq s + 1$
contains an excess-$(s+1)$ core graph. Hence 
$\Pr(\ex(G(S, h_1, h_2)) \geq s + 1) \leq \Pr\left(N^{\CS{s+1}}_S > 0\right)$.
To prove Theorem~\ref{thm:r_excess_prob_zz}, it suffices to  show that
$\Pr\left(N^{\CS{s+1}}_S > 0\right) = O\left(1/n^{s+1}\right)$.
By Lemma~\ref{lem:good:bad}, we know that 
\begin{align}\label{eq:1005}
    \Pr\left(N^{\CS{s+1}}_S > 0\right) \leq \Pr\left(B^{\CS{s+1}}_S\right) +
\E^\ast\left(N^{\CS{s+1}}_S\right).
\end{align}
From \cite[Lemma 7]{AumullerDW14} we know that $\E^\ast\left(N^{\CS{s+1}}_S\right) = O(1/n^{s + 1})$. 
Since $\CS{s+1} \subseteq \LL$, we may apply Lemma~\ref{lem:failure_term_bound} and write
\begin{align}\label{eq:1010}
    \Pr\left(N^{\CS{s+1}}_S > 0\right) &\leq O\left(\frac{n}{\ell^c}\right) +
    O\left(\frac{1}{n^{s + 1}}\right)
= O\left(\frac{1}{n^{s+1}}\right),
\end{align}
for the parameters used in Theorem~\ref{thm:r_excess_prob_zz}. \qquad
\end{proof}

\subsection{Simulation of a Uniform Hash Function}

Consider a universe $U$ of keys and a finite set $R$. Suppose 
we want to construct a hash function that takes on fully random
values from $R$ on 
a key set $S \subseteq U$ of size $n$. 
The naïve construction just assigns a random hash value
to each key $x \in S$ and stores the key-value pair in a hash table that supports
lookup in constant time and construction in expected time $O(n)$, e.\,g., 
cuckoo hashing (with a stash). For information theoretical reasons, this 
construction needs space at least $n \log |R|$. (See, e.\,g., \cite[Lemma 5.3.1]{Rink14}.)
We will now see that we can achieve much more in (asymptotically) almost the same space.

By the term ``\emph{simulating uniform hashing for $U$ and $R$}'' we mean an
algorithm that does the following.  On input $n\in\mathbb{N}$, a randomized
procedure sets up a data structure DS$_n$ that represents a hash function
$h\colon U\to R$, which can then be evaluated efficiently for keys in $U$.  For
each set $S\subseteq U$ of cardinality $n$ there is an event $B_S$ that occurs 
with small probability such that conditioned on $\overline{B_S}$ the values $h(x)$, $x\in S$, are
fully random.  So, in contrast to the naïve construction from above, one $h$ can
be shared among many applications and works on each set $S \subseteq U$ of size
$n$ with high probability.  The quality of the algorithm is determined by the
space needed for DS$_n$, the evaluation time for $h$, and the probability of the
event $B_S$, which we call the \emph{failure probability of the construction}.
It should be possible to evaluate $h$ in constant time.  Again, the information
theoretical lower bound implies that at least $n\log|R|$ bits are needed to
represent DS$_n$.

The first constructions that matched this space bound up to constant factors were 
proposed independently by Dietzfelbinger and Woelfel \cite{DW2003a} and Östlin and Pagh 
\cite{pagh_uniform_conference}. 
In the following, let $R$ be the range of the hash function to be constructed, 
and assume that $(R,\oplus)$ is a commutative group. (For example, we 
could use $R=[t]$ with addition mod $t$.) We sketch the construction 
of \cite{pagh_uniform_conference} next.


The construction described in \cite{pagh_uniform_conference} builds upon 
the graph $G(S, h_1, h_2)$. Each vertex $v$ of $G(S, h_1, h_2)$ is associated with a random
element $x_v$ from $R$.  The construction uses a third hash function $h_3\colon
U \rightarrow R$.  All three hash functions have to be chosen from a
$n^\delta$-wise independent class.  Let $x \in U$ be an
arbitrary key and let $(v,w)$ be the edge that corresponds to $x$ in $G(S, h_1,
h_2)$. The hash value of $x$ is  $h(x) = x_v \oplus x_w \oplus h_3(x)$.  This
construction uses $8n \cdot \log |R| + o(n) + O(\log \log |U|)$ bits of space
and achieves a failure probability of $O(1/n^s)$ for each $s \geq 1$. (The
influence of $s$ on the description length of the data structure is in the
$o(n) + O(\log \log |U|)$ term. It is also in the construction time of the hash
functions $h_1, h_2, h_3$.) The evaluation time is dominated by the evaluation
time of the three highly-independent hash functions.  The construction of
\cite{pagh_uniform_conference} runs in time $O(n)$. In their full paper
\cite{pagh_uniform},  a general method to reduce the
description length of the data structure to $(1+\varepsilon) n \log |R| +
o(n) + O(\log \log |U|)$ bits was presented. This is essentially optimal. This technique
adds a summand of $O(1/\varepsilon^2)$ to the evaluation time. 

Another essentially space-optimal construction was presented by Dietzfelbinger and Rink
in~\cite{DietzfelbingerR09}. It is based on results of Calkin~\cite{Calkin97} and
the ``split-and-share'' approach. 
It uses $(1+\varepsilon)n \log |R| + o(n) + O(\log \log |U|)$ bits of space and
has evaluation time $O(\max\{\log^2(1/\varepsilon), s^2\})$ for failure probability
$O(n^{1-(s+2)/9})$.

The construction presented here is a modification of the construction
in~\cite{pagh_uniform}. We replace
the highly independent hash functions with functions from hash class $\RR$.
The data structure consists of a hash function pair $(h_1,h_2)$ 
from our hash class, two tables of size $m=(1+\varepsilon)n$ each, filled with random elements from $R$,
a two-wise independent hash function with range $R$,
$O(s)$ small tables with entries from $R$, and $O(s)$ two-independent hash functions 
to pick elements from these tables.
The evaluation time of $h$ is $O(s)$, 
and for $S\subseteq U$, $|S|=n$, the event $B_S$ occurs with probability $O(1/n^{s+1})$.
The construction requires roughly twice as much space 
as the most space-efficient solutions~\cite{DietzfelbingerR09,pagh_uniform}. 
However, it seems to be a good compromise combining simplicity and fast evaluation time with moderate space consumption.
\begin{theorem}
Given $n \geq 1, 0 < \delta < 1$, $\varepsilon > 0$, and $s\ge0$,
we can construct a data structure \emph{DS$_n$} that allows
us to compute a function $h\colon U \rightarrow R$ such that\emph{:}\\
\makebox[2em][r]{\textnormal{(i) }}For each $S \subseteq U$ of size $n$ there is an event
$B_S$ of probability $O(1/n^{s+1})$ \linebreak 
  \phantom{\makebox[2em][r]{\textnormal{(i) }}}such that conditioned on $\overline{B_S}$ the function $h$ is distributed uniformly on $S$.\\
\makebox[2em][r]{\textnormal{(ii) }}For arbitrary $x \in U$, $h(x)$ can be evaluated in time
$O(s/\delta)$.\\
\makebox[2em][r]{\textnormal{(iii) }}\emph{DS$_n$} 
comprises $2(1+\varepsilon)n \log |R| + o(n) + O(\log \log |U|)$ bits. 
\label{thm:uniform_hashing}
\end{theorem}
\begin{proof}
Choose an arbitrary integer $c\ge (s+2)/\delta$. Given $U$ and $n$, set up 
DS$_n$ as follows.  
Let $m=(1+\varepsilon)n$ and $\ell=n^\delta$, and choose and store a hash function pair
$(h_1,h_2)$ from $\RR=\RR^{c,2}_{\ell,m}$,
with component functions $g_1,\ldots,g_c$ from $\FF^{2}_\ell$.
In addition, choose two random vectors $t_1,t_2 \in R^m, c$ random vectors
$y_1,\dots,y_c \in R^\ell$, and choose $f$ at random from a $2$-wise
independent
class of hash functions from $U$ to $R$. 

Using DS$_n$, the mapping $h\colon U\to R$ is defined as follows: 
$$
h(x) = t_1[h_1(x)] \oplus t_2[h_2(x)] \oplus f(x) \oplus y_1[g_1(x)] \oplus \ldots \oplus y_c[g_c(x)].
$$  
DS$_n$ satisfies (ii) and (iii) of
Theorem~\ref{thm:uniform_hashing}. (If the universe is too large, it must 
be collapsed to size $n^{s+3}$ first.) We show that it satisfies (i) as well. For this, let $S \subseteq U$ with
$|S| = n$ be given.

First, consider only the hash functions $(h_1,h_2)$ from $\RR$.
By Lemma~\ref{lem:failure_term_bound} we have $\Pr(B^{\LL}_{S}) =
O(n/\ell^{c}) = O(1/n^{s+1})$.
Now fix $(h_1,h_2)\notin {B^{\LL}_{S}}$, which includes fixing the components $g_1,\ldots,g_c$.
Let $T\subseteq S$ be such that $G(T,h_1,h_2)$ is the 2-core of $G(S,h_1,h_2)$, i.e., the 
maximal subgraph with minimum degree at least two.
The graph $G(T,h_1,h_2)$ is leaf\/less, 
and since $(h_1,h_2)\notin {B^{\LL}_{S}}$, we have that $(h_1,h_2)$ is $T$-good.
Now we note that the part $f(x) \oplus \bigoplus_{1\le j \le c}y_j[g_j(x)]$ of
$h(x)$
acts exactly as one of our hash functions $h_1$ and $h_2$, where
$f$ and $y_1, \ldots, y_c$ are yet unfixed.
So, arguing as in the proof of Lemma~\ref{lem:random}
we see that $h$ is fully random on $T$.

Now assume that $f$ and the entries in the tables $y_1,\ldots,y_c$ are fixed. 
Following \cite{pagh_uniform}, we show that the random entries in $t_1$ and $t_2$ alone make
sure 
that $h(x)$, $x\in S-T$, is fully random. For an idea of the proof let $(x_1,\ldots,x_p)$ be the keys in $S \setminus T$,
ordered in such a way that the edge corresponding to $x_i$ is a leaf edge in 
$G(T \cup \{x_1,\ldots,x_i\}, h_1, h_2)$, for each $i \in \{1, \ldots, p\}$. 
To obtain such an ordering, repeatedly remove leaf edges  
from $G = G(S,h_1,h_2)$, as long as this is possible. The sequence of corresponding keys
removed in this way is $x_p, \ldots, x_1$. 
In \cite{pagh_uniform} it is shown by an induction argument that $h$ is uniform on $T \cup \{x_1,\ldots, x_p\}$. \qquad
\end{proof}

When this construction was first described in \cite{AumullerDW12}, it was the easiest to implement
data structure to simulate a uniform hash function in almost optimal space. Nowadays, 
the construction of Pagh and Pagh can use the 
highly-independent hash class construction of Thorup \cite{Thorup13} or 
Christiani, Pagh, and Thorup \cite{ChristianiPT15} instead 
of Siegel's construction. However, in the original analysis of Pagh and Pagh \cite{pagh_uniform}, 
the hash functions are required to be from an $n^\delta$-wise independent hash class. 
It remains to be demonstrated by experiments that the construction
of Pagh and Pagh in connection with the constructions mentioned above is efficient. 
We believe that using
hash class $\RR$ is much faster.

Applying the same trick as in \cite{pagh_uniform}, the data structure 
presented here can be extended to use only $(1+ \varepsilon)n$ words from $R$. The
evaluation time of this construction is  
$O\left(\max\{\frac{1}{\varepsilon^2}, s\}\right)$.

\subsection{Construction of a (Minimal) Perfect Hash Function}

A hash function $h\colon U \rightarrow [m]$ is perfect on $S
\subseteq U$ if it is injective (or $1$-on-$1$) on $S$.  A perfect hash function
is \emph{minimal} if $|S| = m$. Here, $S$ is assumed to be a static set.
Perfect hash functions are usually applied when a large set of items is
frequently queried; they allow fast retrieval and efficient memory storage in
this situation.  
We refer the reader to \cite{BotelhoPZ13,Dietzfelbinger07} for
surveys of constructions for (minimal) perfect hash functions. 

The first explicit practical construction of a
(minimal) perfect hash function which needs
only $O(n)$ bits is due to Botelho, Pagh, and Ziviani \cite{BotelhoPZ07} 
(full version \cite{BotelhoPZ13}) and 
is based on a hypergraph approach.  
Our construction is based on this work. 
In \cite{BotelhoPZ13}, explicit hash functions based on the ``split-and-share'' approach were used. 
This technique builds upon a general strategy described by Dietzfelbinger in
\cite{Dietzfelbinger07} and Dietzfelbinger and Rink in
\cite{DietzfelbingerR09} to make the ``full randomness assumption'' feasible in the construction
of a perfect hash function. Botelho \emph{et al.} showed in experiments that their construction is 
practical, even when realistic hash functions are used.  Our goal is to show that
hash functions from class $\RR$ can be used in a specific version of their construction as well. 
After proving the main result, we will speculate about differences in running time between the 
split-and-share approach of \cite{BotelhoPZ13} and hash class $\RR$. 

The construction of \cite{BotelhoPZ13} to build a perfect hash function mapping keys from a key set $S$
to $[2m]$ with $m = (1+\varepsilon)n$ works as follows. First, the graph $G(S, h_1, h_2)$ is built. 
If this graph contains a cycle,
new hash functions are chosen and the graph is built anew. If the graph is acyclic, a peeling algorithm 
is used to construct a one-to-one mapping $\sigma\colon S \to [2m]$
with $\sigma(x)\in\{h_1(x),m+h_2(x)\}$ for all $x\in S$.
The result of the peeling procedure also makes it possible
to construct two tables $t_1[0..m-1]$ and  $t_2[0..m-1]$ storing bits,
with the property that
\begin{equation*}
\sigma(x)=\genfrac{\{}{\}}{0pt}{0}{h_1(x)}{m+h_2(x)} \qquad \Longleftrightarrow \qquad  
t_1[h_1(x)] \oplus t_2[h_2(x)] = \genfrac{\{}{\}}{0pt}{0}{0}{1}\text{, for }x\in S.
\end{equation*}
It is then obvious how $\sigma(x)$ can be calculated in constant time,
given $t_1[0..m-1]$ and  $t_2[0..m-1]$.

If $(h_1, h_2)$ are fully random hash functions, the probability that the graph is acylic,  i.e., the probability that the construction succeeds, for $m = (1+ \varepsilon) n $ is 
\begin{align}\label{hashing:eq:50002}
    \sqrt{1 - \left(\frac{1}{1+\varepsilon}\right)^2}. 
\end{align}

We replace the pair $(h_1, h_2)$ of hash functions by functions from $\RR$.
The next lemma shows that for $m \geq 1.08n$ we can build a perfect hash function
for a key set $S$ by trying the construction of Botelho \emph{et al.}
a constant number of times (in expectation).
\begin{lemma}
    Let $S \subseteq U$ with $S = n$. Let $\varepsilon \geq 0.08$, and let $m \geq (1
    + \varepsilon) n$. Set $\ell = n^\delta$ and $c \geq 1.25/\delta $. Then for a randomly chosen pair
    $(h_1,h_2) \in \RR^{c,2}_{\ell,m}$ we have
    \begin{align}\label{hashing:eq:50001}
	\Pr(G(S,h_1,h_2) \text{ is acyclic}) 
	\geq 1 + \frac{1}{2} \ln \left(1 - \left(
	\frac{1}{1+\varepsilon}\right)^2\right) - o(1).
    \end{align}\label{lem:cycles}
\end{lemma}
Figure~\ref{fig:comparison:cycle:success:prob} depicts the difference between
the probability bounds in \eqref{hashing:eq:50002} and in \eqref{hashing:eq:50001}. The theoretical bound using a first moment approach
is close to the behavior in a random graph when $\varepsilon \geq 1$.

\begin{figure}
    \centering
    \scalebox{0.8}{
\begin{tikzpicture}
    \selectcolormodel{gray}
    \begin{axis}
	[enlarge x limits=false,
	domain=0.075:4,
	width = 14cm,
	height = 6cm,
	xmin = 0.075,
    xtick = {0.08, 0.5, 1, 1.5, 2, 2.5, 3, 3.5, 4},
    x tick label style = { /pgf/number format/.cd,
            fixed,
            fixed zerofill,
            precision=2,
        /tikz/.cd},
	xlabel = $\varepsilon$,
	ylabel = Success probability,
    legend cell align = left,
	legend style = {anchor = south east, at={(0.95,.03)}, draw=none}]
	\addplot[smooth]{sqrt(1 - (1/(1+x))^2)};
	\addplot[smooth,dashed]{1 + 0.5 * ln(1 - (1/(1+x))^2)};
    \legend{\raisebox{2.5ex}{$\sqrt{1 - \left(\frac{1}{1+\varepsilon}\right)^2}$},\raisebox{2.5ex}{$1 +
	\frac12 \ln\left(1
    - \left(\frac{1}{1+\varepsilon}\right)^2\right)$}}
\end{axis}
\end{tikzpicture}}
    \caption{Comparison of the probability of a random graph being acyclic and 
	the theoretical bound following from a first moment approach for values $\varepsilon \in [0.08, 4]$.}
    \label{fig:comparison:cycle:success:prob}
\end{figure}
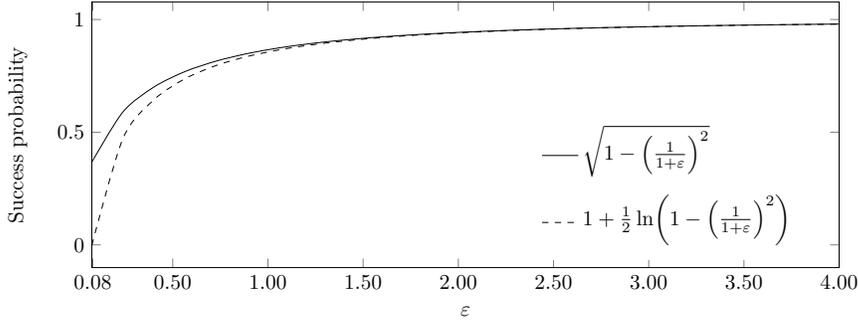

%

\begin{proof}[Of Lemma~\ref{lem:cycles}]
    Let
	$\mathsf{CYC}$ be the set of all cycles in $\GG^2_{m, n}$.
    Note that all these cycles have even length, since we consider bipartite
    graphs.
    By Lemma~\ref{lem:good:bad}, we may bound
    $\Pr(N^{\mathsf{CYC}}_S > 0)$ by $\Pr(B^\mathsf{CYC}_S) +
    \E^\ast(N^\mathsf{CYC}_S)$. Since $\mathsf{CYC} \subseteq \LL$, we
    know that
    $\Pr(B^\mathsf{CYC}_S) = O(n/\ell^{c})$, see
    Lemma~\ref{lem:failure_term_bound}. For the parameter choices
    $\ell = n^\delta$ and $c\geq 1.25/\delta$ we have $\Pr(B^\mathsf{CYC}_S) =
    o(1)$. We now
    focus on the second summand and calculate (as in \cite{BotelhoPZ13}):
    \begin{align*}
        \E^\ast\left(N^\mathsf{CYC}_S\right) &= \sum_{t = 1}^{n/2}
        \mu^\mathsf{CYC}_{2t} \leq \sum_{t = 1}^{n/2}\frac{\binom{n}{2t}(2t)! \cdot m^{2t}}{2t \cdot m^{2 \cdot 2t}}
	=\sum_{t = 1}^{n/2} \frac{\binom{n}{2t} \cdot (2t)! }{2t
    \cdot m^{2t}} \leq \sum_{t=1}^{n/2} \frac{n^{2t}}{2t \cdot m^{2t}}\\
    &= \sum_{t = 1}^{n/2}\frac{1}{2t \cdot (1+\varepsilon)^{2t}} \leq \sum_{t =
1}^\infty \frac{1}{2t \cdot (1 + \varepsilon)^{2t}} = -\frac{1}{2} \ln \left(1
- \left(\frac{1}{1+\varepsilon}\right)^2\right),
    \end{align*}
    where the last step is the Maclaurin expansion. \qquad
\end{proof}

According to Lemma~\ref{lem:cycles}, we can build a perfect hash function with 
range $[2.16n]$ with a constant number of constructions of $G(S, h_1, h_2)$ (in expectation).
To store the data structure we need $2.16n$ bits (to store $\texttt{g}_1$ and $\texttt{g}_2$), and
$o(n)$ bits to store the pair $(h_1,h_2)$ from $\RR$. For example, for a set of $n = 2^{32}$ keys,
i.e., about $4.3$ billion keys, the pair $(h_1, h_2)$ 
may consist of ten tables with $256$ entries each, five $2$-universal hash functions, 
and two $2$-independent
hash functions, see Lemma~\ref{lem:cycles} with parameters $c = 5$ and $\delta = 1/4$. This seems to 
be more practical than the split-and-share approach from \cite{BotelhoPZ13} which
uses more and larger tables per hash function, \emph{cf.} \cite[Section 4.2]{BotelhoPZ13}.
However, 
it remains future work to demonstrate in experiments how both approaches compare to each other. 
To obtain a minimal perfect hash function, one has to compress the perfect hash function
further. This roughly doubles the description length, see \cite{BotelhoPZ13} for details.

In their paper \cite{BotelhoPZ13}, Botelho \emph{et al.}
showed that minimal space usage is achieved when using three hash functions $h_1, h_2, h_3$
to build the hypergraph $G(S, h_1, h_2, h_3)$. In this case, one can construct a 
perfect hash function with range $[1.23n]$ with high probability. Since the $g$-values
must then index three hash functions, $1.23n \cdot \log_23 \approx 1.95n$ bits are needed 
to store the bit vectors. According to \cite{BotelhoPZ13}, the minimal perfect hash function 
needs about $2.62n$ bits. Our results with regard to hypergraphs do not lead to a construction
that can compete. 

\subsection{Connected Components of $G(S, h_1, h_2)$ are small}
As is well known from the theory of 
random graphs, for a key set $S \subseteq U$ of size $n$ and  
$m = (1 + \varepsilon)n$, for $\varepsilon > 0$, and fully 
random hash functions $h_1, h_2: U \rightarrow [m]$ the graph 
$G(S, h_1, h_2)$ w.h.p. contains  only components of at most logarithmic size
which are trees or unicyclic. (This is the central argument for 
standard cuckoo hashing to work.) We 
show here that hash class $\RR$ can provide this behavior if one is willing
to accept a density that is smaller by a constant factor. Such situations have been
considered in the seminal work of Karp \emph{et al.}
\cite{klmadh96} on the simulation of shared memory
in distributed memory machines.

We give the following result as a corollary. It has first appeared in \cite{klmadh96}, for 
a different class of ($\sqrt{n}$-wise independent) hash functions. Here we prove it for hash class $\RR$.
\begin{corollary}[{\cite[Lemma 6.3]{klmadh96}}]
    Let $S \subseteq U$ with $|S| = n$. Let $m \geq 6n$.  Then for each  
$\alpha \geq 1$, there are $\beta, \ell, c, s \geq 1$ such that for $G=G(S,h_1,h_2)$
with $(h_1,h_2) \in \RR^{c,2}_{\ell, m}$ we have that 
\begin{enumerate}
    \item[\textnormal{(}a\textnormal{)}] $\Pr(G \text{ has a
	connected component with at least $\beta \log n$ vertices}) =
	O(n^{-\alpha})$.
    \item[\textnormal{(}b\textnormal{)}] $\Pr(G \text{ has a
	connected component with $k$ vertices and $\geq k + s - 1$ edges}) =
	O(n^{-\alpha})$.
\end{enumerate}
\label{cor:connected:components}
\end{corollary}

\begin{proof}
    We start with the proof of (b). If $G = G(S,h_1,h_2)$ has a connected
    component $A$ with $k$ vertices and at least $k + s - 1$ edges, then $\ex(G) \geq s-1$.
    According to
    Theorem~\ref{thm:r_excess_prob_zz}, the probability that such a component appears 
    is $O(1/n^{\alpha})$, for $s = \alpha$ and $c \geq 2(\alpha + 1)$. 
    
    The proof of Part (a) requires more care. For this part, we may focus on the probability that $G$ contains a tree with
    $k = \beta \log n$ vertices. 
    We let $\mathsf{T}$ consist of all trees with $k$ vertices in $\GG_{m,n}$ and
    apply Lemma~\ref{lem:good:bad} to get
    \begin{align}\label{eq:10050}
        \Pr\left(N^\mathsf{T}_S > 0\right) \leq \Pr\left(B^{\mathsf{T}}_S\right)
	+ \E^\ast\left(N^\mathsf{T}_S\right).
    \end{align}
    Since $\mathsf{T} \subseteq \mathsf{LCY}$, we have that $\Pr(B^\mathsf{T}_S) =
    O(n/\ell^c)$, see Lemma~\ref{lem:failure_term_bound}. We now bound the second summand of
    \eqref{eq:10050}. Note that the calculations are essentially the same as the ones made in \cite{klmadh96} 
    to prove their Lemma 6.3. 
    By Cayley's formula we know that there are
    $k^{k-2}$ labeled trees with vertex set $\{1,\ldots,k\}$. Fix such a tree
    $T^\ast$. 
    We can label the edges of $T^\ast$ with $k - 1$ keys from $S$ in $\binom{n}{k-1} \cdot (k-1)!$ many ways. Furthermore,
    there are two ways to fix which vertices of $T^\ast$ belong to which side of the
    bipartition. After this, there are not more than $\binom{2m}{k}$ ways to assign 
    the vertices of $T^\ast$ to
    vertices in the bipartite graph $G(S,h_1,h_2)$. Once all these labels of $T^\ast$ are
    fixed, the probability that the hash values of $(h_1,h_2)$ realize
    $T^\ast$ is $1/m^{2(k-1)}$. We can thus calculate:
    \begin{align*}
	\E^\ast(N^\mathsf{T}_S) & \leq
	\frac{\binom{n}{k-1} \cdot k^{k-2} \cdot 2 \cdot (k-1)! \cdot
    \binom{2m}{k}}{m^{2(k-1)}}
    \leq \frac{2^{k+1}\cdot m^2 \cdot k^{k-2}}{6^k \cdot k!}\\
    &\leq \frac{2^{k+1}\cdot m^2 \cdot k^{k-2}}{6^k \cdot \sqrt{2\pi k}
\cdot (k/e)^{k}} \leq 2m^2 \cdot \left(\frac{e}{3}\right)^k\\
    \end{align*}
    Part ($a$) follows for $k = \Omega(\log n)$.  \qquad
\end{proof}

The authors of \cite{klmadh96} use the hash functions from \cite{DietzfelbingerH90} (precursor of $\RR$) 
combined with functions from Siegel's
hash class to obtain a hash class with high ($\sqrt{n}$-wise) independence.
They need this high level of independence in the proof of their Lemma 6.3, which
states properties of the connected components in the graph built from the key set
and these highly independent hash functions. Replacing Lemma $6.3$ in
\cite{klmadh96} with our Corollary~\ref{cor:connected:components} 
immediately implies that the results of \cite{klmadh96}, in particular, their Theorem~6.4, 
also hold when (only) hash functions from $\RR$ are used. 
In particular, in \cite{klmadh96} the sparse setting where $m$ is at least  $6n$ was considered as well. We remark
that statement (b) of Corollary~\ref{cor:connected:components} holds for $m \geq (1 + \varepsilon)n$. 

Moreover, this result could be applied to prove results for cuckoo hashing (with
a stash), and de-amortized cuckoo hashing of Arbitman \emph{et al.}
\cite{ArbitmanNS09}. However, note that while in the fully random case the
statement of Corollary~\ref{cor:connected:components} holds for $m = (1
+ \varepsilon) n$, here we had to assume $m \geq 6n$, which yields only very low
hash table load. We note that this result cannot be improved to $(1+
\varepsilon) n$ using the first moment approach inherent in our approach and the approach of
\cite{klmadh96} (for $\sqrt{n}$-wise independence), since
the number of unlabeled trees that have to be considered in the first moment
approach is too large \cite{Otter48}.  It remains open to show that graphs
built with our class of hash functions have small connected components for all
$\varepsilon > 0$.

%% file: sec-4-hypergraphs.tex
\section{Applying the Framework on Hypergraphs}\label{hashing:sec:applications:hypergraphs}

In this section we will discuss some applications of hash class $\RR$
in the setting with more than two hash functions, i.e., 
each edge of $G(S, \vec{h})$ contains at least three vertices. 
We will study three different applications: Generalized cuckoo hashing 
with $d \geq 3$ hash functions as proposed by Fotakis, Pagh, Sanders, and Spirakis \cite{FotakisPSS05}, 
two recently described insertion algorithms 
for generalized cuckoo hashing due to Khosla \cite{Khosla13} and Eppstein, Goodrich, Mitzenmacher, and Pszona \cite{Eppstein14}, and
different schemes for load balancing as studied by Schickinger and Steger \cite{SchickingerS00}. 

For applications regarding generalized cuckoo hashing, 
we will study the failure term of $\RR$ on the 
respective graph properties directly. We will show that $\RR$ allows
running these applications efficiently. However, we have to assume that 
the load of the hash table is rather low. For the application with regard
to load balancing schemes, the failure term of $\RR$ will be analyzed by
means of  a very general graph property. 
However, it requires higher parameters when setting
up a hash function from $\RR$, which degrades the performance of these
hash functions. We will start by introducing some notation and making a small
generalization to the framework. 

\paragraph{Hypergraph Notation} A hypergraph extends the notion of an undirected 
graph by allowing edges 
to consist of more than two vertices. 
We use the hypergraph notation from \cite{ShamirS85,karonski02}.  
A hypergraph is called \emph{$d$-uniform}
if every edge contains exactly $d$ vertices. Let $H = (V, E)$ be a hypergraph. 
A \emph{hyperpath from $u$ to $v$} in $H$ 
is a sequence $(u = u_1, e_1, u_2, e_2, \ldots, $ $e_{t - 1}, u_t = v)$
such that $e_i \in E$ and $u_i, u_{i+1} \in e_i$, for $1 \leq i \leq t - 1$.
The hypergraph $H$ is \emph{connected} if for each pair of vertices $u,v \in V$ there exists a hyperpath
from $u$ to $v$.

The \emph{bipartite representation} of a hypergraph $H$ is 
the bipartite graph $\text{bi}(H)$
where vertices of $H$ are the vertices on the right side of the bipartition, the
edges of $H$ correspond to vertices on the left side of the
bipartition, and two vertices are connected by an edge in the bipartite graph if
the corresponding edge in the hypergraph contains the corresponding vertex.
Note that a hypergraph $H'$ is a subgraph of a hypergraph $H$
(as defined in Section~\ref{hashing:sec:basics})
if and only if $\text{bi}(H')$ is a subgraph of $\text{bi}(H)$ in the standard sense.

We will use a rather strict notion of cycles in hypergraphs. 
A connected hypergraph is called a \emph{hypertree} if $\text{bi}(H)$ is a tree. 
A connected hypergraph is called \emph{unicyclic} if $\text{bi}(H)$ is unicyclic.
A connected hypergraph that is neither a hypertree nor unicyclic is called \emph{complex}.
Using the standard formula to calculate the cyclomatic number of a graph\footnote{The cyclomatic number 
of a connected graph $G$ 
with $m$ vertices and $n$ edges is $n - m + 1$.} \cite{diestel}, we get the following (in)equalities for a
connected $d$-uniform hypergraph $H$ with
$n$ edges and $m$ vertices: $(d - 1) \cdot n = m - 1$ if $H$ is a hypertree, $(d - 1) \cdot n = m$ if $H$
is unicyclic, and $(d-1) \cdot n > m$ if $H$ is complex.

We remark that there are different notions with respect to cycles in 
hypergraphs. In other papers,
e.g., \cite{CzechHM97,Dietzfelbinger07,BotelhoPZ13}, a hypergraph is called
\emph{acyclic} if and only if there exists a sequence of repeated deletions of
edges containing at least one vertex of degree $1$ that yields a hypergraph
without edges. (Formally, we can arrange the edge set $E = \{e_1, \ldots, e_n\}$ 
of the hypergraph in a 
sequence $(e'_1, \ldots, e'_n)$ 
such that $e'_j - \bigcup_{s < j} e'_s \neq \emptyset$, for $1 \leq j \leq n$.) 
We will call this process of repeatedly removing edges incident to a vertex
of degree 1 the \emph{peeling process}, see, e.g., \cite{Molloy05}. With respect
to this definition, a hypergraph $H$ is acyclic if and only if the $2$-core of
$H$ is empty, where the $2$-core of $H = (V, E)$ is the largest set $E' \subseteq E$ such that 
each vertex in $(V,E')$ has minimum degree $2$, disregarding isolated vertices.
An acyclic
hypergraph in this sense can have unicyclic and complex
components according to the definition from above. In the analysis, we will
point out why it is important for our work to use the concepts introduced above. 

\paragraph{Hypergraph-Related Additions to the Framework} 
When working with hypergraphs, it will be helpful to allow removing single vertices from 
hyperedges. This motivates considering the following generalizations of the notions 
``peelability'' and ``reducibility'' for hypergraph properties. 

\begin{definition}[Generalized Peelability]
A hypergraph property
$\mathsf{A}$ is called \textbf{peelable} if for all $G=(V,E)
    \in \mathsf{A}$, $|E| \geq 1$, 
		there exists an edge $e \in E$ such that 
        \begin{enumerate}
            \item $(V,E\setminus \{e\}) \in \mathsf{A}$ or 
            \item there exists an edge $e' \subseteq e$ with $|e'| < |e|$ and $|e'| \geq 2$,
		where $e$ and $e'$ have the same label,
		such that $(V,(E\setminus\{e\})\cup\{e'\}) \in \mathsf{A}$.
\end{enumerate}
\label{def:peelability:gen}
\end{definition}

\begin{definition}[Generalized Reducibility]
Let $c \in \mathbb{N}$, and let $\mathsf{A}$ and $\mathsf{B}$ be hypergraph properties.
$\mathsf{A}$ is called
\textbf{$\mathsf{B}$-$2c$-reducible} if for all graphs $(V,E) \in \mathsf{A}$ and sets $E^\ast \subseteq E$ with
$|E^\ast| \leq 2c$ we have the following:
There exists a subgraph $(V, E')$ of $(V,E)$ with $(V, E') \in \mathsf{B}$ such that 
each edge $e' \in E'$ is a subset of some $e\in E$ with the same label and 
for each edge $e^\ast \in E^\ast$ there exists an edge $e' \in E'$ with $e' \subseteq e^\ast$ 
and $e'$ and $e^\ast$ having the same label.
\label{def:reducible:gen}
\end{definition}

In contrast to Definition~\ref{def:peelability} and Definition~\ref{def:reducible}, we can remove vertices from edges
in a single peeling or reduction step.
A proof analogous to the proof of Lemma~\ref{lem:fail_prob} shows that the statement of that lemma
is true in the hypergraph setting as well. 

\subsection{Generalized Cuckoo Hashing}\label{hashing:sec:generalized:cuckoo:hashing}

The obvious extension of cuckoo hashing is to use a sequence $\vec{h} = (h_1, \ldots, h_d)$ of $d \geq 3$ hash functions. 
For a given integer $d \geq 3$ and a key set $S \subseteq U$ with $|S| = n$, our
hash table consists of $d$ tables $T_1, \ldots, T_d$, each of size $m = O(n)$, and uses $d$ hash functions
$h_1, \ldots, h_d$ with $h_i\colon U \rightarrow [m]$, for $i \in \{1, \ldots, d\}$. A key $x$ must be stored
in one of the cells $T_1[h_1(x)], T_2[h_2(x)], \ldots,$ or $T_d[h_d(x)]$. Each table cell contains at most one key. Searching and removing a key works in the 
obvious way. For the insertion procedure, note that evicting a key $y$ from a table $T_j$ leaves, in contrast
to standard cuckoo hashing, $d-1$ other choices where to put the key.  To think about 
insertion procedures, it helps to introduce the concept of a certain directed graph. 
Given a set $S$ of keys stored in a cuckoo 
hash table with tables $T_1, \ldots, T_d$ using $\vec{h}$, 
we define the following (directed) cuckoo allocation graph $G = (V,E)$, see, e.g., \cite{Khosla13}:
The vertices $V$ correspond to the memory cells in $T_1, \ldots, T_d$. 
The edge set $E$ consists of all edges $(u,v) \in V \times V$ such that 
there exists a key $x \in S$ so that $x$ is stored in the table cell which
corresponds to vertex $u$ ($x$ \emph{occupies} $u$)  and $v$ corresponds to one of the $d-1$ other
choices of key $x$. If $u \in V$ has out-degree $0$, we call $u$ \emph{free}. 
(The table cell which corresponds to vertex $u$ does not contain a key.)
Standard methods proposed for inserting a key (see \cite{FotakisPSS05}) are \emph{breadth-first search}
to find a shortest eviction sequence or a \emph{random walk} approach.
In the next section, we will study two alternative insertion strategies that
were suggested recently. If an insertion fails, a new sequence of hash
functions is chosen and the data structure is built anew.

If the hash functions used are fully random 
it is now fully understood what table sizes $m$ make it possible w.h.p. to store 
a key set according to the cuckoo hashing rules for a given number of hash functions. In $2009$,\footnote{%
Technical report versions of all papers cited here
were published at \url{www.arxiv.org}. We refer to the final publications.} this case was
settled independently by Dietzfelbinger \emph{et al.}
\cite{DietzfelbingerGMMPR10}, Fountoulakis and Panagiotou
\cite{FountoulakisP10}, and Frieze and Melsted \cite{FriezeM12}.  
Later, the random walk insertion algorithm was partially analyzed by Frieze,
Melstedt, and Mitzenmacher \cite{FriezeMM11} and Fountoulakis, Panagiotou, and Steger 
\cite{FountoulakisPS13}. 

Here, we study the static setting in which we ask if $\vec{h}$ 
allows accommodating a given key set $S \subseteq U$ in the
hash table according to the cuckoo hashing rules. 
This is equivalent to the
question whether the hypergraph $G = G(S, \vec{h})$ built from $S$ and $\vec{h}$ is
$1$-orientable or not, i.e., whether there is an injective function
that maps each edge $e$ to a vertex on $e$ or not. 
If $G(S, \vec{h})$ is
$1$-orientable, we call $\vec{h}$ \emph{suitable} for $S$.

We now discuss some known results for random hypergraphs.  As for simple random
graphs \cite{ErdosR60} there is a sharp transition phenomenon for random
hypergraphs \cite{karonski02}. When a random hypergraph with $m$ vertices 
has at most $ (1
- \varepsilon) m/(d(d-1))$ edges, all components are small and all components
are \emph{hypertrees} or \emph{unicyclic} with high probability.  On the other
hand, when it has at least $ (1 + \varepsilon) m/(d(d-1))$ edges, there exists
one large, \emph{complex} component. We will analyze generalized cuckoo hashing
under the assumption that each table has size $m \geq (1 +  \varepsilon) (d-1)
n$, for $\varepsilon > 0$.  Note that this result is rather weak: The load of
the hash table is at most $1/(d(d-1))$, i.e., the more hash functions we use,
the weaker the edge density bounds we get for the hash functions to provably work. At the end of
this section, we will discuss whether this result can be improved or not with the
methodology used here.

We will show the following theorem.
\begin{theorem}
Let $\varepsilon > 0, 0 < \delta <1$, $d \geq 3$  be
given. Assume $c\ge  2/\delta$.  
    For $n\ge 1$, consider $m\ge(1+\varepsilon)(d-1)n$ and $\ell=n^\delta$.  
    Let $S \subseteq U$ with $|S|=n$.
    Then for $\vec{h} = (h_1, \ldots, h_d)$ chosen at random from $\RR=\RR^{c,d}_{\ell,m}$ the
following holds:
    $$\Pr\left( \vec{h} \text{ is not suitable for $S$}  \right) = O(1/n).$$
    \label{thm:generalized:cuckoo:hashing:static}
\end{theorem}

Most of this subsection is devoted to the proof of this theorem. 

\begin{lemma}
    Let $H$ be a hypergraph.
    If $H$ contains no complex component then $H$ is $1$-orientable.
\label{lem:no:complex:component}
\end{lemma}

\begin{proof}
    We may consider each connected component of $H$ separately.
    First, observe that a hypertree and a unicyclic component always contains an edge
    that is incident to a vertex of degree $1$. 

    Suppose $C$ is such a hypertree or a unicyclic component.
    A $1$-orientation of $C$ is obtained via the 
    well-known ``peeling process'', see, e.g., \cite{Molloy05}. It works
    by iteratively peeling edges incident to a vertex of degree $1$ and
    orienting each edge towards such a vertex. \qquad
\end{proof}

In the light of Lemma~\ref{lem:no:complex:component} we bound the probability of
$G(S, \vec{h})$ being $1$-orientable by the probability that $G(S,\vec{h})$
contains no complex component. A connected complex component of $H$ causes
$\text{bi}(G)$ to have at least two cycles. So, minimal obstruction hypergraphs that show
that a hypergraph contains a complex component are very much like the obstruction graphs 
that showed that a graph contains more than one cycle, see Figure~\ref{hashing:fig:obstruction:sets}
on Page~\pageref{hashing:fig:obstruction:sets}. For a clean definition of obstruction hypergraphs,
we will first introduce the concept of a \emph{strict path} in a hypergraph. 
A sequence 
$(e_1, \ldots, e_t)$ with $t \geq 1$ and $e_i \in E$, for $1 \leq i \leq t$, 
is a \emph{strict path} in $H$ if $| e_i \cap e_{i
+ 1}| = 1$ for $1 \leq i \leq t - 1$ and  $|e_{i} \cap e_{j}| = 0$ for $1 \leq i \leq t - 2$ and $i + 2 \leq j \le t$.  
According to \cite{karonski02}, a complex
connected component contains a subgraph of one of the following two types:

\begin{enumerate}
    \item[Type 1:] A strict path $e_1, \ldots, e_t, t \geq 1$, and an edge $f$ such that $\left| f \cap e_1 \right| \geq 1$,
        $\left| f \cap e_t \right| \geq 1$, and $\left| f \cap \bigcup_{i = 1}^t e_i \right| \geq 3$.
    \item[Type 2:] A strict path $e_1, \ldots, e_{t - 1}, t \geq 2$, and
        edges $f_1$, $f_2$ such that $\left| f_1 \cap e_1 \right| \geq 1$,
        $\left| f_2 \cap e_{t - 1} \right| \geq 1$, and $\left| f_j \cap
        \bigcup_{i = 1}^{t-1} e_i \right| \geq 2$, for $j \in \{1,2\}$.
\end{enumerate}
The  bipartite representation of a hypergraph of Type~$1$ contains a cycle with a chord,
the  bipartite representation of a hypergraph of Type~$2$ contains two cycles connected by a path of 
some nonnegative length.
We call a hypergraph $H$ in $\GG^{d}_{m,n}$ of Type $1$ or Type $2$
a \emph{minimal complex obstruction hypergraph}.  Let $\MOS$ denote the set of all
minimal complex obstruction hypergraphs in $\GG^d_{m, n}$,
with edges labeled by distinct elements from $\{1,\dots,n\}$. In the following, our
objective is to apply
Lemma~\ref{lem:good:bad}, which says that
\begin{align}
    \Pr\left(N^\MOS_S > 0  \right) \leq \Pr\left(B^\MOS_S\right) + \E^\ast\left(N^\MOS_S\right).
    \label{hashing:eq:generalized:1}
\end{align}

\paragraph{Bounding $\E^\ast\left(N^\MOS_S\right)$}

We prove the following lemma:
\begin{lemma}
    Let $S \subseteq U$ with $|S| = n, d\geq 3$, and $\varepsilon > 0$ be given. 
    Assume $m \geq (1 + \varepsilon) (d - 1)n$. Then 
    \begin{align*}
        \E^\ast\left(N^\MOS_S\right) = O(1/n).
    \end{align*}
    \label{lem:MOS:fully:random}
\end{lemma}

\emph{Proof}.
    From the proof of \cite[Theorem 4, P. 128]{karonski02}, we know that 
    the number $w_1(d, t+ 1)$ of \emph{unlabeled} minimal complex obstruction hypergraphs with $t + 1$ edges is at most
    \begin{align}
        w_1(d, t+ 1) \leq d m^d  \left( \left( d - 1 \right) m^{d - 1}
        \right)^{t - 2} t^2d^4   \left( \left(  d -
        1 \right) m^{d - 1} m^{d - 3} + m^{2(d - 2)} \right)
    \label{eq:w1}
    \end{align}
    So,
    the number of minimal complex obstruction hypergraphs with $t + 1$ edges and edge labels from $\{1,\dots,n\}$
    is not larger than
%
%
   \begin{align*}
       &\left(\binom{n}{t + 1} \cdot (t + 1)!\right)  d m^d \left( \left( d - 1 \right) m^{d - 1} \right)^{t - 2} \cdot t^2d^4 \cdot
       \left( \left(  d - 1 \right) m^{d - 1} \cdot m^{d - 3} + m^{2(d - 2)} \right)\\
       &\leq n^{t + 1} \cdot d^6 \cdot m^{(d - 1)(t + 1) -1} \cdot t^2 \cdot (d-1)^{t-2}\\
       &\leq n^{ d (t+1) - 1} \cdot d^6 \cdot t^2 \cdot ( 1 + \varepsilon)^{(d-1)(t + 1) - 1} \cdot (d-1)^{(d-1) (t + 1) + t - 3} \\
       &= n^{ d (t+1) - 1}\cdot d^6 \cdot t^2 \cdot (1 + \varepsilon)^{(d-1)(t + 1) - 1} \cdot (d - 1)^{d  (t + 1) - 4}.
   \end{align*}
    Let $H$ be a labeled minimal complex obstruction hypergraph with $t+1$ edges.

   Draw $t+1$ edges at random from $[m]^d$, one for each labeled edge in $H$. 
   The probability that the hash values
   realize $H$ is $1/m^{d(t+1)} \leq 1/\left( (1
   + \varepsilon)  \left( d - 1\right)n \right)^{d(t + 1)}$. So,
   \begin{align*}
       \E^\ast\left(N^\MOS_S\right) &\leq \sum_{t = 1}^{n} \frac{d^6 \cdot t^2
       \cdot (1 + \varepsilon)^{(d - 1) \cdot (t+1) - 1} \cdot \left(
       d - 1 \right)^{d \cdot (t + 1)  - 4} \cdot n^{d \cdot
       \left( t + 1 \right) - 1}}{\left( (1+ \varepsilon) \left( d - 1 \right)n \right)^{d(t
       + 1)}}\\
   &\leq \frac{d^6}{(d-1)^4 n} \cdot \sum_{t + 1}^{n} \frac{t^2}{(1 + \varepsilon)^{t - 1}} = O\left(\frac{1}{n}\right).\qquad \endproof
   \end{align*}

\paragraph{Bounding $\Pr\left(B^\MOS_S\right)$}

We will now prove the following lemma.

\begin{lemma}
    Let $S \subseteq U$ with $|S| = n, d \geq 3,$ and  $\varepsilon > 0$ be given. 
    Assume $m \geq (1+\varepsilon)(d-1)n$.
    Let $\ell, c \geq 1$. Choose $\vec{h} \in \RR^{c, d}_{\ell, m}$ at random. Then
    \begin{align*}
        \Pr\left(B^\MOS_S\right) = O\left(\frac{n}{\ell^c}\right).
    \end{align*}
    \label{lem:mos:fail:prob}
\end{lemma}
To use our framework from Section~\ref{hashing:sec:basics}, we have to find a suitable peelable 
hypergraph property that
contains $\MOS$. Since minimal complex obstruction hypergraphs are path-like, we relax 
that notion in the following way.
\begin{definition}
    Let $\mathsf{PX}$ 
		be the set of all hypergraphs $H$ from $\GG^d_{m,n}$ that fall in one of the following
    categories:
    \begin{enumerate}
        \item $H$ has hypergraph property $\MOS$.
        \item $H$ is a strict path.
        \item $H$ consists of a strict path $e_1, \ldots,e_t$, $t\geq 1$, 
				and an edge $f$ such that $|f \cap (e_1 \cup e_t)| \geq 1$ and
            $\left| f \cap
        \bigcup_{i = 1}^{t} e_i \right| \geq 2.$
    \end{enumerate}
    \label{hashing:def:p:ast}
\end{definition}
All hypergraphs in $\mathsf{PX}$ are extensions of paths. 
Note that property $3$ is somewhat artificial to deal with the case that a single
edge of a minimal complex obstruction hypergraph of Type $2$ is removed. Obviously, $\MOS$ is contained in 
$\mathsf{PX}$, and property $\mathsf{PX}$
is peelable. We can now prove Lemma~\ref{lem:mos:fail:prob}.  

\emph{Proof of Lemma~\ref{lem:mos:fail:prob}}.
We apply Lemma~\ref{lem:fail_prob} which says that 
\begin{align*}
    \Pr\left(B^{\mathsf{PX}}_S\right) \leq \frac{1}{\ell^c} \sum_{t = 1}^{n} t^{2c} \mu_t^{\mathsf{PX}}.
\end{align*}
    We start by counting unlabeled hypergraphs $G \in \mathsf{PX}$
    having exactly $t + 1$ edges.  For the hypergraphs having Property 1 of
    Definition~\ref{hashing:def:p:ast}, we may use the bound \eqref{eq:w1} in the proof of
    Lemma~\ref{lem:MOS:fully:random}.  Let $w_2(d, t + 1)$ be the number of
    such hypergraphs which are strict paths, i.e., have Property $2$ of
    Definition~\ref{hashing:def:p:ast}. We obtain the following bound:
    \begin{align*}
        w_2(d, t + 1) \leq &d m^d \cdot \left( \left(d - 1\right) \cdot 
        m^{d - 1} \right)^{t}.
    \end{align*}
    Let $w_3(d, t + 1)$ be the number of hypergraphs having Property $3$ of Definition~\ref{hashing:def:p:ast}.
    We observe that 
    \begin{align*}
        w_3(d, t + 1) \leq & d m^d \cdot \left( \left(d - 1\right) \cdot 
        m^{d - 1} \right)^{t - 1} \cdot 2d^2 \cdot t \cdot
        m^{d-2}.
    \end{align*}
    So, the number of fully labeled hypergraphs having exactly $t + 1$ edges is at most 
    \begin{align*}
        &\left(\binom{n}{t + 1} \cdot (t + 1)!\right)\cdot \left( w_1(d, t + 1)
       + w_2(d, t + 1) + w_3(d, t+1)\right)\\
       &\leq n^{t + 1} \cdot d m^d \cdot \left( (d - 1) \cdot 
        m^{d - 1} \right)^{t - 2} \, \cdot \\
        &\quad\left( d^2  
        m^{2(d - 1)}  + d^4 t 
        m^{2 (d - 2)} + d^4 t^2 m^{2(d  - 2)}
        + 2d^3 t
        m^{2d - 3} \right)\\
        &\leq 4\cdot d^5 \cdot t^2 \cdot n^{t+1} \cdot m^{(d - 1)(t + 1) + 1} \cdot (d-1)^{t-2}\\
        & \leq 4 \cdot d^5 \cdot t^2 \cdot n^{d(t+1) + 1} \cdot (1 + \varepsilon)^{(d - 1)( t + 1) + 1} \cdot (d - 1)^{d ( t + 1) - 2}.
    \end{align*}
    We may thus calculate:
    \begin{align*}
        \Pr\left(B^{\mathsf{PX}}_S\right) &\leq \frac{1}{\ell^c} \sum_{t = 1}^{n} t^{2c} \cdot \frac{4 \cdot d^5 \cdot t^2 \cdot (1 + \varepsilon)^{(d - 1)( t + 1) + 1} \cdot (d - 1)^{d ( t + 1) - 2} \cdot n^{d (t + 1) + 1}}{( ( 1 + \varepsilon) (d-1)n)^{d(t + 1)}}\\
                                              &\leq \frac{n}{\ell^c}\sum_{t = 1}^{n}\frac{4 \cdot d^5 \cdot t^{2 + 2c}}{(d-1)^2 \cdot (1+\varepsilon)^{t - 1}} = O\left( \frac{n}{\ell^c}\right).\qquad \endproof
    \end{align*}

 \paragraph{Putting Everything Together} Substituting the results of
Lemma~\ref{lem:MOS:fully:random} and Lemma~\ref{lem:mos:fail:prob} into
 \eqref{hashing:eq:generalized:1} yields
\begin{align*}
    \Pr\left(N^\MOS_S > 0 \right) = O\left(\frac{1}{n}\right) + O\left(\frac{n}{\ell^c}\right).
\end{align*}
Theorem~\ref{thm:generalized:cuckoo:hashing:static} follows by
setting $c \geq 2/\delta$ and $\ell = n^\delta$.

Theorem~\ref{thm:generalized:cuckoo:hashing:static} shows that 
given a key set of size $n$, if we use 
$d$ tables of size at least $(1+\varepsilon)(d-1)n$ and hash functions from $\RR$
there exists w.h.p. an assignment of the keys to memory 
cells according to the cuckoo hashing rules. Thus,
the load of the hash table is smaller than $1/(d(d-1))$. In the 
fully random case, the load of the hash table rapidly grows towards $1$,
see, e.g., the table on Page~$5$ of \cite{DietzfelbingerGMMPR10}. For example,
using $5$ hash functions allows the hash table load to be $\approx 0.9924$. 
The approach followed in this section cannot yield such bounds for the following
reason. When we look back at the proof of Lemma~\ref{lem:no:complex:component},
we notice that it 
gives a stronger result: It shows that when a graph does not contain
a complex component, it has an empty two-core, i.e., it does not contain
a non-empty subgraph in which each vertex has minimum degree $2$. It is known from 
random hypergraph theory that the appearance of a non-empty two-core becomes
increasingly likely for $d$ getting larger.\footnote{According to \cite[p. 418]{MezardM09} (see also \cite{MajewskiWHC96}) 
for large $d$ the $2$-core of a random $d$-uniform 
hypergraph with $m$ vertices and $n$ edges is empty with 
high probability if $m$ is bounded from below by $dn/\log d$.}  So, we cannot rely on hypergraphs with empty two-cores to 
prove bounds for generalized cuckoo hashing that improve for increasing values
of $d$.

Karo{\'n}ski and 
{\L}uczak showed in \cite{karonski02} that if a random $d$-uniform hypergraph has $m$ vertices and 
at most $(1
- \varepsilon) m/(d(d-1))$ edges, then all connected components have size $O(\log m)$
with high probability. 
On the other hand, if a random $d$-uniform hypergraph has at least $(1
+ \varepsilon) m/(d(d-1))$ edges, then there exists a unique connected component of size
$\Theta(m)$. So, for $d \geq 3$ the analysis of general cuckoo hashing takes
place in the presence of the giant component, which differs heavily from the analysis of standard 
cuckoo hashing. Whether or not hash class $\RR$ allows suitable bounds for the general case will
depend on whether or not there exist small subgraphs in the giant component which are sufficiently unlikely 
to occur in the fully random case.

Now that we turned our focus to hypergraphs, we again see that the analysis is
made without exploiting details of the hash function construction, only using
the general framework developed in Section~\ref{hashing:sec:basics} together
with random graph theory.

\subsection{Labeling-based Insertion Algorithms For Generalized Cuckoo Hashing}

In the previous section we showed that when the tables are large enough, the hash
functions allow storing $S$ according to  the cuckoo hashing rules with high
probability. In this section we prove that such an assignment
can be obtained (with high probability) with hash functions 
from $\RR$ using two recently described 
insertion algorithms. 

In the last section, we pointed out two natural insertion
strategies for generalized cuckoo hashing: breadth-first search and random walk,
described in \cite{FotakisPSS05}.
Very recently, Khosla \cite{Khosla13} (2013) and Eppstein \emph{et al.}
\cite{Eppstein14} (2014) presented two new insertion strategies, which will be
described next. 
In both algorithms, each table cell $i$ in  table $T_j$   has
a label (or counter) $l(j, i) \in \mathbb{N}$, where initially $l(j, i) = 0$ for all
$j \in \{1,\ldots,d\}$ and $i \in \{0, \ldots,m - 1\}$.  The insertion of a key $x$
works as follows: Both strategies find the table index $$j = \argmin_{j \in
\{1,\ldots,d\}} \{l(j, h_j(x))\}.$$ If $T_j[h_j(x)]$ is free then $x$ is stored in this cell and the insertion terminates successfully. Otherwise, let $y$ be the key that resides in $T_j[h_j(x)]$.
Store $x$ in $T_j[h_j(x)]$. The difference between the two algorithms 
is how they adjust the labeling. The algorithm of Khosla sets $$l(j, h_j(x)) \leftarrow
\min\{l(j', h_{j'}(x)) \mid j' \in \left(\{1,\ldots,d\} \setminus \{j\}\right)\}
+ 1,$$ while the algorithm of Eppstein \emph{et al.} sets $l(j, h_j(x))
\leftarrow l(j, h_j(x)) + 1$. Now insert $y$ in the same way. This is iterated until an empty
cell is found or it is noticed that the insertion cannot be performed successfully.\footnote{
Neither in \cite{Eppstein14} nor in \cite{Khosla13} it is described how 
this should be done in the cuckoo hashing setting. From the analysis presented there, 
when deletions are forbidden, one should do the following: Both algorithms  
have a counter \text{MaxLabel}, and if there exists a label $l(j,i) \geq \text{MaxLabel}$, then
one should choose new hash functions and re-insert all items. For Khosla's algorithm, 
$\text{MaxLabel} = \Theta(\log n)$; for the algorithm of Eppstein \emph{et al.}, one should 
set $\text{MaxLabel} = \Theta(\log \log n)$.}
In Khosla's algorithm, the content of the label $l(j, i)$ is a lower bound
for the minimal length of an eviction sequence that makes it possible to store a new element
into $T_j[i]$ by moving other elements around~\cite[Proposition 1]{Khosla13}. 
In the algorithm of Eppstein \emph{et al.}, the label
$l(j, i)$ contains the number of times the memory cell $T_j[i]$ has been overwritten.
According to \cite{Eppstein14}, it aims to minimize the 
number of write operations to a memory cell. 
Here we show that in the sparse setting with $m \geq (1+\varepsilon)(d-1)n$,
using class $\RR$ 
the maximum label in the algorithm of Eppstein \emph{et al.}
is $\log \log n + O(1)$ with high probability and the maximum label in the 
algorithm of Khosla is $O(\log n)$ with high probability.
This corresponds to results proved in these papers for fully random hash functions. 

Our result when using hash functions from $\RR$ is as follows. We only study
the case that we want to insert the keys from  a set $S$ sequentially without
deletions.
\begin{theorem}
Let $\varepsilon > 0$, $0 < \delta <1$, $d \geq 3$  be
given. Assume $c\ge  2/\delta$.  
    For $n\ge 1$ consider $m\ge(1+\varepsilon)(d-1)n$ and $\ell=n^\delta$.  
    Let $S \subseteq U$ with $|S|=n$.
    Choose $\vec{h} \in \RR^{c,d}_{\ell,m}$ at random. Insert all keys from $S$
		according to $\vec{h}$ in an arbitrary order, using the algorithm of Khosla. 
    Then with probability $1-O(1/n)$ (i) all key insertions are successful and (ii)
    $\max\{l(j,i) \mid i \in \{0, \ldots,m-1\}, j \in \{1, \ldots, d\}\} = O(\log n)$.
    \label{thm:khoslas:algorithm}
\end{theorem}
\begin{theorem}
Let $\varepsilon > 0$, $0 < \delta <1$, $d \geq 3$  be
given. Assume $c\ge  2/\delta$.  
    For $n\ge 1$ consider $m\ge(1+\varepsilon)(d-1)n$ and $\ell=n^\delta$.  
    Let $S \subseteq U$ with $|S|=n$.
    Choose $\vec{h} \in \RR^{c,d}_{\ell,m}$ at random. Insert all keys from $S$
		according to $\vec{h}$ 
    in an arbitrary order, using the algorithm of Eppstein \emph{et al.} 
    Then with probability $1-O(1/n)$ (i) all key insertions are successful and (ii)
    $\max\{l(j,i) \mid i \in \{0, \ldots,m-1\}, j \in \{1, \ldots, d\}\} = \log
    \log n + O(1)$.
    \label{thm:goodrichs:algorithm}
\end{theorem}

For the analysis of both algorithms we assume that the insertion of an element 
fails if 
there exists a label of size $n + 1$. (In this case, new hash functions
are chosen and the data structure is built anew.) Hence,
to prove Theorem~\ref{thm:khoslas:algorithm} and Theorem~\ref{thm:goodrichs:algorithm} 
it suffices to show that statement (ii) holds. (An unsuccessful insertion yields
a label with value $>  n$.)

\paragraph{Analysis of Khosla's Algorithm} We first analyze the algorithm of Khosla. We remark that in our setting, Khosla's algorithm finds an assignment with high probability. (In \cite[Section 2.1]{Khosla13} Khosla gives an easy argument why her
algorithm always finds an assignment when this is possible. In the previous
section, we showed that such an assignment exists with probability $1-O(1/n)$.) 
It remains to prove that the maximum label has size $O(\log n)$.
We first introduce the notation used by 
Khosla in \cite{Khosla13}. Recall the definition of the cuckoo allocation graph 
from the beginning of Section~\ref{hashing:sec:generalized:cuckoo:hashing}. Let $G$ be a cuckoo allocation graph. Let
$F_G \subseteq V$ consist of all free vertices in $G$. Let $d_G(u,v)$ be the 
distance between $u$ and $v$ in $G$. Define
\begin{align*}
    d_G(u,F) := \min\left(\left\{d_G(u,v) \mid v \in F\right\} \cup \{\infty\}\right).
\end{align*}
Now assume that the key set $S$ is inserted in an arbitrary order. Khosla defines a 
\emph{move} as every action that writes an element into a table cell. 
(So, the $i$-th insertion is decomposed into $k_i \geq 1$ moves.)  The allocation graph at the end of the $p$-th move is denoted by $G_p = (V, E_p)$. Let $M$ denote the number of moves
necessary to insert $S$. (Recall that we assume that $\vec{h}$ is 
suitable for $S$.)
Khosla shows the following connection between labels and distances to a free
vertex.
\begin{proposition}[{\cite[Proposition 1]{Khosla13}}]
    For each $p \in \{0, 1, \ldots, M\}$ and each $v \in V$ it holds that 
        $d_{G_p}(v, F_{G_p}) \geq l(j, i)$,
    where  $T_j[i]$ is the table cell that corresponds to vertex $v$.
\end{proposition}

Now fix an integer $L \geq 1$.
Assume that there exists an integer $p$ with $0 \leq p \leq M$ and a vertex $v$ such that 
$d(v,F_{G_{p}}) = L$. Let $(v = v_0, v_1, \ldots, v_{L-1}, v_L)$ be a simple
path $p$ of length $L$ in $G_p$ such that $v_L$ is free. Let $x_0, \ldots, x_{L - 1} \subseteq S$ be 
the keys which occupy $v_0, \ldots, v_{L-1}$. Then the hypergraph $G(S, \vec{h})$ contains a subgraph $H$ that corresponds to $p$ in the obvious way.

\begin{definition}
    For given integers $L \geq 1, m \geq 1, n \geq 1, d \geq 3$, let $\SP$ (``simple path'') consist of all 
    hypergraphs $H = (V, \{e_1,\ldots,e_L\})$ in $\GG^d_{m,n}$ with the following
    properties:
    \begin{enumerate}
        \item For all $i \in \{1, \ldots, L\}$ we have that $|e_i| = 2$. (So, $H$ is a graph.)
        \item For all $i \in \{1, \ldots, L - 1\}$, $|e_i \cap e_{i + 1}| = 1$.
        \item For all $i \in \{1, \ldots, L - 2\}$ and $j \in \{i + 2, \ldots, L\}$,  $|e_i \cap e_j| = 0$.
    \end{enumerate}
\end{definition}
Our goal in the following is to show that there exists a constant $c$ such that
for all $L \geq c \log n$ we have $\Pr\left(N^{\SP}_S > 0\right) = O(1/n)$. From Lemma~\ref{lem:good:bad} we obtain the bound
\begin{align}
    \Pr\left(N^{\SP}_S > 0\right) \leq \E\phantom{}^\ast\left(N^{\SP}_S\right) + \Pr\left(B^{\SP}_S\right).
    \label{eq:khosla:1}
\end{align}

\paragraph{Bounding $\E^\ast\left(N^{\SP}_S\right)$}
We show the following lemma.
\begin{lemma}
    Let $S \subseteq U$ with $|S| = n, d \geq 3,$ and $\varepsilon >0$ be given. 
    Consider $m \geq (d - 1) (1 + \varepsilon) n$. Then
    \begin{align*}
        \E^\ast\left(N^{\SP}_S\right) \leq \frac{md}{(1+\varepsilon)^L}.
    \end{align*} 
    \label{lem:khosla:fully:random}
\end{lemma}

\emph{Proof}.
    We count fully labeled hypergraphs with property $\SP$. Let $P$
    be an unlabeled simple path of length $L$. There are
    $d \cdot (d-1)^L$ ways to label the vertices on $P$ with $\{1, \ldots, d\}$
    to fix the class of the partition they belong to. Then there are not more 
    than $m^{L + 1}$ ways to label the vertices with labels from $[m]$. There are fewer than $n^{L}$ ways to label the edges with labels from $\{1, \ldots, n\}$. Fix such a fully labeled path $P'$. Now draw
    $2L$ hash values from $[m]$ according to
    the labels of $P'$. The probability that these random choices realize $P'$ is $1/m^{2L}$. We calculate:
    \begin{align*}
        \E^\ast\left(N^{\SP}_S\right) 
            &\leq \frac{d \cdot \left(d - 1\right)^L \cdot m^{L + 1} \cdot n^L}
            {m^{2L}}
            = \frac{m \cdot d \cdot \left(d - 1\right)^L}
            {\left(\left(d - 1\right)\left(1 + \varepsilon\right)\right)^L}
        = \frac{md}{\left(1 + \varepsilon\right)^L}.\qquad \endproof
    \end{align*} 

\paragraph{Bounding $\Pr\left(B^{\SP}_S\right)$}
Note that $\SP$ is not peelable. We relax $\SP$ in the obvious way and define $\RSP = \bigcup_{0 \leq i \leq L} \SP$.
Graph property $\RSP$ is peelable. 
\begin{lemma}
    Let $S \subseteq U$ with $|S| = n$ and $d \geq 3$ be given. 
    For an $\varepsilon > 0$, set $m \geq (1+\varepsilon)(d-1)n$.
    Let $\ell, c \geq 1$. Choose $\vec{h} \in \RR^{c, d}_{\ell, m}$ at random.
    Then 
    \begin{align*}
        \Pr\left(B^{\SP}_S\right) = O\left(\frac{n}{\ell^c}\right).
    \end{align*}
    \label{lem:khosla:fail:prob}
\end{lemma}

\emph{Proof}.
    Since $\SP \subseteq \RSP$ and $\RSP$ is peelable, we may apply Lemma~\ref{lem:fail_prob} and obtain the bound
    \begin{align*}
        \Pr\left(B^{\SP}_S\right) &\leq \frac{1}{\ell^c} \cdot \sum_{t = 1}^n t^{2c} \cdot \mu^{\RSP}_t.
    \end{align*}
    By the definition of $\RSP$ and using the same counting argument as in
     the proof of Lemma~\ref{lem:khosla:fully:random}, we calculate:
     \begin{align*}
     \Pr\left(B^{\SP}_S\right) \leq \frac{1}{\ell^c} \cdot \sum_{t = 1}^n t^{2c} \cdot \frac{md}{(1+\varepsilon)^t} = O\left(\frac{n}{\ell^c}\right).\qquad \endproof
     \end{align*}

\paragraph{Putting Everything Together} Plugging the results of Lemma~\ref{lem:khosla:fully:random} and Lemma~\ref{lem:khosla:fail:prob} into
\eqref{eq:khosla:1} shows that
\begin{align*}
    \Pr\left(N^{\SP}_S > 0\right) \leq \frac{md}{(1+\varepsilon)^L} + O\left(\frac{n}{\ell^c}\right).
\end{align*}
Setting $L = 2\log_{1+\varepsilon}(n), \ell = n^\delta$, and $c \geq 2/\delta$ finishes the proof of Theorem~\ref{thm:khoslas:algorithm}.

\paragraph{Analysis of the Algorithm of Eppstein et al}

We now analyze the algorithm of Eppstein \emph{et al.} \cite{Eppstein14}.
We use the \emph{witness tree technique} to prove Theorem~\ref{thm:goodrichs:algorithm}.
This proof technique was introduced by Meyer auf der Heide, Scheideler, and Stemann
\cite{MadHSS96} in the context of shared memory simulations, and is one of the
main techniques to analyze load balancing processes (see, e.g., 
\cite{ColeFMMRSU98,ColeMHMRSSV98,stemann96,SchickingerS00,voecking}), which
will be the topic of the next section.

Central to our analysis is the notion of a  \emph{witness tree for wear $k$}, for an integer
$k \geq 1$.  (Recall that in the algorithm of Eppstein \emph{et al.}, the label $l(j, i)$ 
denotes the number of times the algorithm has put a key into the cell $T_j[i]$. This is also called the
\emph{wear} of the table cell.) For given values $n$ and $m$, a witness tree for wear $k$ is
a $(d-1)$-ary tree with $k + 1$ levels in which each non-leaf node is labeled
with a tuple $(j, i, \kappa)$, for $1 \leq j \leq d$, $0 \leq i \leq  m  - 1$, and
$1 \leq \kappa \leq n$, and each leaf is labeled with a tuple $(j, i)$, $1
\leq j \leq d$ and $0 \leq i \leq  m  - 1$.  Two children of a non-leaf node $v$ 
must have different first components
($j$-values) and, if they exist, 
third components ($\kappa$-values).
In addition, the $\kappa$-values of a node and its children must differ.

We call a witness tree
\emph{proper} if no two different non-leaf nodes have the same labeling. Further, we say
that a witness tree $T$ can be \emph{embedded into} $G(S, \vec{h})$ if for each
non-leaf node $v$ with label $(j_0, i_0, \kappa)$ with children labeled $(j_1, i_1),
\ldots, (j_{d-1}, i_{d-1})$ in the first two label components in $T$,
$h_{j_k}(x_\kappa) = i_k$, for each $0 \leq k \leq d - 1$. 
We can think of a proper witness tree 
as an edge-labeled hypergraph from $\GG^d_{m,n}$ by building from each non-leaf node labeled ($j_0, i_0, \kappa$) 
together with its $d-1$ children with label components $(j_1, i_1),\ldots,(j_{d-1}, i_{d-1})$ 
a hyperedge $(i'_0, \ldots, i'_{d-1})$ labeled ``$\kappa$'', where $i'_0, \ldots, i'_{d-1}$ are ordered according
to the $j$-values.

Suppose that there exists a label $l(j, i)$ with content $k$ for an integer $k
> 0$. We now argue about what must have happened that $l(j, i)$ has such
a label. In parallel, we construct the witness tree for wear $k$. Let $T$ be an
unlabeled $(d - 1)$-ary tree with $k + 1$ levels.  Let $y$ be
the key residing in $T_j[i]$. Label the 
root of $T$ with $(j, i, \kappa)$, where $y = x_\kappa \in S$. 
Then for all other choices of $y$ in tables $T_{j'}, j' \in \{1, \ldots, d\},  j' \neq j,$
we have $l(j',h_{j'}(y)) \geq k - 1$. (When $y$ was written into $T_j[i]$, $l(j,i)$
was $k - 1$ and this was minimal among all choices of key $y$. Labels are never
decreased.)
Let $x_1,\ldots,x_{d-1}$ be the keys in these $d - 1$ other choices of $y$. 
Label the children of the root of $T$ with the $d - 1$ tuples
$(j', h_{j'}(y)), 1 \leq j' \leq d, j' \neq j,$ and the respective key indices. 
Arguing in the same way as above, we see that for each key $x_i, i \in \{1,
\ldots, d - 1\}$, its $d - 1$ other table choices must 
have had a label of at least $k - 2$.  Label the children of the node corresponding to key $x_i$ 
on the second level of $T$ with the $d - 1$ other choices, for each $i \in \{1, \ldots, d - 1\}$. 
(Note that already the third level may include nodes with the same label.) Proceeding with 
this construction on the levels $3, \ldots, k$ gives the witness tree $T$ for wear $k$. By
construction, this witness tree can be embedded into $G(S, \vec{h})$.  

So, all we have to do to prove Theorem~\ref{thm:goodrichs:algorithm} is to obtain a (good enough) bound on the probability that a witness tree for wear $k$ can be embedded into 
$G(S, \vec{h})$. If a witness tree is not proper, it seems difficult to calculate the probability that 
this tree can be embedded into $G(S, \vec{h})$, because different parts of the witness tree
correspond to the same key in $S$, which yields dependencies among hash values.  
However, we know from the last section
that when $G(S, \vec{h})$ is sparse enough, i.e., $m \geq  (1+ \varepsilon) (d
- 1) n$, it contains only hypertrees and unicyclic components with probability
$1 -  O(1/n)$. Using a basic pruning argument, Eppstein \emph{et al.} show that this
simplifies the situation in the following way.
\begin{lemma}[{\cite[Observation 2 and Observation 3]{Eppstein14}}]
    Let $H$ be a  hypergraph that consists only of hypertrees and unicyclic components.
    Suppose $H$ contains an embedded witness tree for wear $k$. Then there exists 
    a proper witness tree for wear $k-1$ that can be embedded into $H$.
\end{lemma}

Let $W_{S,k}$ 
be the event that there exists a witness tree 
for wear $k$ that can be embedded into $G(S, \vec{h})$.   
To prove Theorem~\ref{thm:goodrichs:algorithm}, we have to
show that for the parameter choices in Theorem~\ref{thm:goodrichs:algorithm}  the probability that $W_{S,k}$ 
occurs is $O(1/n)$.

We separate the cases that $G(S, \vec{h})$ contains a complex component
and that this is not so.
Let
$\PWT{k}$  
be the set of all hypergraphs in $G^d_{m, n}$ that correspond to 
proper witness trees for wear $k$. 
Using Theorem~\ref{thm:generalized:cuckoo:hashing:static}, we may bound:
\begin{align}
    \label{eq:pwt:1}
    \Pr\left(W_{S,k}\right) 
    &\leq \Pr\left(N^{\PWT{k-1}}_S
    > 0\right) + \Pr\left(N^{\MOS}_S > 0\right)\notag\\
                                            &\leq
                                            \Pr\left(B^{\PWT{k-1}}_S\right)
                                            + \E^\ast\left(N^{\PWT{k-1}}_S\right)
                                            +  \Pr\left(N^{\MOS}_S > 0\right).
\end{align}
The last summand on the right-hand side of this inequality is handled by
Theorem~\ref{thm:generalized:cuckoo:hashing:static}, so we may concentrate
on the hypergraph property $\PWT{k-1}$.

\paragraph{Bounding $\E^\ast\left(N^{\PWT{k-1}}_S\right)$}
We start by proving that the expected number of proper witness trees in $G(S, \vec{h})$ is
$O(1/n)$ for the parameter choices in Theorem~\ref{thm:goodrichs:algorithm}. We use a different
proof method than Eppstein \emph{et al.} \cite{Eppstein14}, because we cannot use 
the statement of \cite[Lemma~1]{Eppstein14}. 
We remark here that the following analysis could be extended to obtain
bounds of $O(1/n^{s})$, for $s \geq 1$. However, the last summand of \eqref{eq:pwt:1} is 
$O(1/n)$, so this does not improve the bounds for \eqref{eq:pwt:1}. 
\begin{lemma}
    Let $S \subseteq U$ with $|S| = n$ and $d \geq 3$ be given. 
    For an $\varepsilon > 0$, set $m \geq (1+\varepsilon)(d-1)n$.
    Then there exists a value $k = \log \log n + \Theta(1)$ such that 
        \begin{align*}
            \E^\ast\left(N^{\PWT{k-1}}_S\right) = O\left(\frac{1}{n}\right).
        \end{align*}
    \label{lem:uwt:fully:random}
\end{lemma}
\begin{proof}
    We first obtain a bound on the number of proper witness trees for wear $k-1$. 
Let $T$ be an unlabeled $(d-1)$-ary tree with $k$ levels.
The number $v_{k-1}$ of vertices of such a tree is 
$$v_{k-1} = \sum_{i = 0}^{k-1} (d-1)^i = \frac{(d-1)^k - 1}{d - 2}.$$
For the number $e_{k-1}$ of non-leaf nodes of such a tree, we have 
$$e_{k-1} = \sum_{i = 0}^{k - 2} (d-1)^i = \frac{(d-1)^{k - 1} - 1}{d - 2}.$$
There are $n \cdot d \cdot m$ ways to label the root of $T$. There are not more than
$n^{d-1} \cdot m^{d-1}$ ways to label the second level of the tree. Labeling the
remaining levels in the same way, we see that in total there are fewer than
    $n^{e_{k - 1}} \cdot d \cdot m^{v_{k-1}}$
proper witness trees for wear $k - 1$. Fix such a fully labeled witness tree $T$. Now
draw $d \cdot e_{k - 1} = v_{k - 1}  + e_{k - 1} - 1$
values randomly from $[m]$ according to the labeling of the nodes in $T$. The
probability that these values realize $T$ is exactly $1/m^{v_{k - 1} + e_{k - 1}
- 1}$.  We obtain the following bound:
\begin{align*}
    \E^\ast\left(N^{\PWT{k-1}}_S\right) &\leq \frac{
    n^{e_{k - 1}} \cdot d \cdot \left( \left( 1 + \varepsilon \right) (d - 1) n  \right)^{v_{k-1}}}
    {\left( \left( 1 + \varepsilon \right) (d - 1) n \right)^{v_{k-1} + e_{k-1} - 1}}
    = \frac{n \cdot d }{\left( (1 + \varepsilon) (d - 1) \right)^{e_{k-1} - 1}}\\
    &\leq \frac{n \cdot d}{\left( (1 + \varepsilon) (d - 1) \right)^{(d - 1)^{k - 2}}},
\end{align*}
which is $O(1/n)$ for $k = \log \log n + \Theta(1)$. \qquad
\end{proof}

\paragraph{Bounding $\Pr\left(B^{\PWT{k-1}}_S\right)$} 
We first relax the notion of a witness tree in the following way.
\begin{definition}
    Let $\RWT{k-1}$ (relaxed witness trees) be the set of all hypergraphs
    which can be obtained in the following way:
    \begin{enumerate}
	\item Let $T \in \PWT{k'}$ be an arbitrary proper witness tree for wear $k'$, with $k'
	    \leq k - 1$. Let $\ell$ denote the number of nodes on level $k' - 1$, i.e., the level
            prior to the leaf level of $T$.
        \item Arbitrarily choose $\ell' \in \mathbb{N}$ with  $\ell' \leq \ell - 1$.
	\item Choose $\kappa = \lfloor \ell' / ( d - 1)\rfloor $ arbitrary distinct non-leaf nodes on level $k' - 2$. 
            For each such node, remove all its children together with their $d - 1$ children from $T$.
            Then remove from a group of $d - 1$ siblings on level $k' - 1$ the
            $\ell' - (d - 1) \cdot \kappa$ siblings with the largest $j$-values together with their leaves.
    \end{enumerate}
\end{definition}
Note that $\RWT{k-1}$ is a peelable hypergraph property, for we can iteratively remove non-leaf nodes that 
correspond to edges in the hypergraph until 
the whole leaf level is removed. Removing these nodes as described in the third property makes sure that 
there exists at most one non-leaf node at level $k' - 2$ that has fewer than $d - 1$ children. 
Also, it is clear what the first
components in the labeling of the children of this node are. Removing nodes in a more arbitrary fashion would give 
more labeling choices and thus more trees with property $\RWT{k-1}$.
\begin{lemma}
    Let $S \subseteq U$ with $|S| = n$ and $d \geq 3$ be given. 
    For an $\varepsilon > 0$, set $m \geq (1+\varepsilon)(d-1)n$.
    Let $\ell, c \geq 1$. Choose $\vec{h} \in \RR^{c, d}_{\ell, m}$ at random.
    Then 
    \begin{align*}
        \Pr\left(B^{\RWT{k-1}}_S\right) = O\left(\frac{n}{\ell^c}\right).
    \end{align*}
    \label{lem:rwt:failure:term}
\end{lemma}

\emph{Proof.}
We apply Lemma~\ref{lem:fail_prob}, which says that 
\begin{align*}
    \Pr\left(B^{\RWT{k-1}}_S\right) \leq \frac{1}{\ell^c} \cdot \sum_{t = 2}^{n} t^{2c} \mu^{\RWT{k-1}}_t.
\end{align*}
Using the same line of argument as in the bound for $\E^\ast\left(N^{\PWT{k-1}}_S\right)$, the expected number of 
witness trees with property $\RWT{k-1}$ with exactly $t$ edges, i.e., exactly $t$ non-leaf nodes, 
is at most $n \cdot d \cdot / \left( (1 + \varepsilon) (d - 1) \right)^{t - 1}$. We calculate:
\begin{align*}
    \Pr\left(B^{\RWT{k-1}}_S\right) &= \frac{1}{\ell^c}\cdot\sum_{t = 1}^{n}
    t^{2c} \frac{n \cdot d}{\left( (1 + \varepsilon) (d - 1) \right)^{t - 1}}
= O\left(\frac{n}{\ell^c}\right). \qquad \endproof
\end{align*}

\paragraph{Putting Everything Together} Using \eqref{eq:pwt:1} and
Lemmas~\ref{lem:uwt:fully:random} and~\ref{lem:rwt:failure:term}, we conclude that 
\begin{align*}
    \Pr\left(W_{S,k}\right)  &\leq
    \Pr\left(B^{\PWT{k-1}}_S\right)
    + \E^\ast\left(N^{\PWT{k-1}}_S\right) +  \Pr\left(N^{\MOS}_S > 0\right)\\
    &= O(1/n) + O(n/\ell^c).
\end{align*}
Theorem~\ref{thm:goodrichs:algorithm} follows for $\ell=n^{\delta}$ and $c \geq 2/\delta$. 


With respect to the algorithm of Eppstein \emph{et al}., our result shows that $n$ insertions take time
$O(n \log \log n)$ with high probability when using hash functions from
$\RR$.  With an analogous argument to the one given by Khosla in \cite{Khosla13}, the algorithm of Eppstein \emph{et al}. of course
finds an assignment of the keys whenever this is possible. However,
the bound of $O(\log \log n)$ 
on the maximum label is only known for $m \geq (1 + \varepsilon)
(d - 1) n$ and $d \geq 3$, even in the fully random case.  Extending the analysis on the 
maximum label size to denser hypergraphs is an interesting open question. 

\subsection{Load Balancing}\label{sec:load_balancing}
In this section we apply hash class $\RR$ in the area
of load balancing schemes.  In the discussion at the end of this section, we will present 
a link of our results w.r.t. load balancing to the space utilization of generalized cuckoo hashing
in which each memory cell can hold $\kappa \geq 1$ items.

In randomized load balancing we want to allocate a set of jobs
$J$ to a set of machines $M$ such that a condition, e.g., there exists no machine
with ``high'' load, is satisfied with high probability. To be consistent with the notation
used in our framework and previous applications, $S$ will denote the set of jobs, and 
the machines will be numbered $1, \ldots, m$. In this section we
assume $|S| = n = m$, i.e., we allocate $n$ jobs to $n$ machines.  

We use the following approach to load balancing: For an integer $d \geq 2$, we
split the $n$ machines into groups of size $n/d$ each. For simplicity, we assume
that $d $ divides $n$. Now a job chooses $d$ candidate machines by choosing
exactly one machine from each group. This can be
modeled by using $d$ hash functions $h_1,\dots,h_d$ with $h_i\colon
S \rightarrow [n/d], 1 \leq i \leq d$, such that machine $h_i(j)$ is the
candidate machine in group $i$ of job $j$. 

In load balancing schemes, the
arrival of jobs has been split into two models: parallel and sequential arrival.
We will focus on parallel job arrivals and come back to the sequential case at the
end of this section.

In the parallel arrival model, all jobs arrive at the same time.
They communicate with
the machines in synchronous rounds. In these rounds, decisions on the allocations of jobs
to machines are made. The $\tau$-collision protocol is one algorithm 
to find such an assignment.
This protocol was studied in the context of distributed memory machines  
by Dietzfelbinger and Meyer auf der Heide \cite{DietzfelbingerMadH93}. 
As a method for load balancing the allocation algorithm was analyzed by
Stemann in \cite{stemann96}. The $\tau$-collision protocol works in the following way:
First, each job chooses one candidate machine from each of the $d\geq 2$
groups. Then the following steps are repeated until all jobs are assigned to
machines:
\begin{enumerate}
\item Synchronously and in parallel, each unassigned job sends an allocation
request to each of its candidate machines.
\item Synchronously and in parallel, each machine sends an acknowledgement to all
requesting jobs if and only if it got at most $\tau$ allocation requests in this
round. Otherwise, it does not react.
\item Each job that gets an acknowledgement is assigned to one of the machines 
    that has sent an acknowledgement. Ties are broken arbitrarily.
\end{enumerate}
Note that the number of rounds is not bounded. However, 
in \cite{DietzfelbingerMadH93} and~\cite{stemann96} it was shown that w.h.p. the
$\tau$-collision protocol will terminate after a small number
of rounds, if suitable hash classes are used. We will show that this
also holds when class $\RR$ is used.

There exist several analysis techniques for load balancing, e.g., 
layered induction, fluid limit
models and witness trees \cite{handbook_of_randomized_computing}. 
We will focus on the witness
tree technique. 
We use the variant studied by
Schickinger and Steger in \cite{SchickingerS00} in
connection with hash class $\RR$. The main contribution of \cite{SchickingerS00} 
is that it provides a unified analysis for several load balancing algorithms. This
allows us to show that hash class $\RR$ is suitable in all of these situations
as well, with only little additional work. 

At the core of the analysis in \cite{SchickingerS00} is the so-called \emph{allocation graph}. In our
setting, where each job chooses exactly one candidate machine in each of the $d$
groups, the allocation graph is a bipartite graph $G = ([n],[n],E)$, where the
jobs are on the left side of the bipartition, and the machines are on the right
side, split into groups of size $n/d$. Each job vertex is adjacent to its $d$
candidate machines. As already discussed in
Section~\ref{hashing:sec:generalized:cuckoo:hashing}, the allocation graph is
equivalent to the hypergraph $G(S, \vec{h})$.
Recall that we refer to the bipartite representation of 
a hypergraph $G = (V,E)$ 
as bi$(V,E)$. We 
call the vertices on the left side \emph{job
vertices} and the vertices on the right side \emph{machine vertices}.

If a machine has high load we can find a subgraph in the allocation graph that
shows the chain of events in the allocation process that led to this situation,
hence ``witnessing'' the high load of this machine. (This is similar to the wear of
a table cell in the algorithm of Eppstein \emph{et al.}~\cite{Eppstein14} in the previous
section.) Such witness trees might differ greatly in structure, depending on the
load balancing scheme.

In short, the approach of Schickinger and Steger works as follows.\footnote{This approach
has a lot in common with our analysis of insertion algorithms for generalized cuckoo
hashing. However, the analysis will be much more complicated here, since the 
hypergraph $G(S, \vec{h})$ has exactly as many vertices as edges.}
\begin{enumerate}
\item They show that high load leads to the existence of a ``witness graph''
and describe the properties of such a graph for a given load balancing scheme.
\item For their analysis to succeed they demand that the witness graph from
above is a tree in the standard sense. They show that with
high probability a witness graph can be turned into a cycle-free witness tree by
removing a small number of edges at the root.
\item They show that it is unlikely that the
allocation graph contains such a witness tree.
\end{enumerate}
We will give a detailed description of this approach after stating the main
result of this section.

The following theorem represents one
selected result from \cite{SchickingerS00}, replacing the full randomness
assumption with hash functions from $\RR$ to choose candidate machines for
jobs.  We simplify the theorem by
omitting the exact parameter choices calculated in \cite{SchickingerS00}. All
the other examples considered in \cite{SchickingerS00} can be analyzed in an
analogous way, resulting in corresponding theorems. We discuss this claim
further in the discussion part of this section. 

\begin{theorem}
For each constant $\alpha > 0, d \geq 2$, there exist constants $\beta, c > 0$ (depending on $\alpha$ and $d$), such that for
each $t$ with $2
\leq t \leq (1/\beta) \ln \ln n,
\ell = n^{1/2}$ and $\vec{h} = (h_1, \ldots, h_d) \in \RR^{c,d}_{\ell, n}$,
the
$\tau$-collision protocol described above with threshold $\tau = O\left(
((\ln n)/(\ln \ln n))^{1/(t-2)}/(d-1)\right)$ finishes after $t$ rounds with
probability $1 - O(n^{-\alpha})$.
\label{thm:parallel_schickinger}
\end{theorem}

We will now analyze the $\tau$-collision protocol using hash functions from class
$\RR$. Most importantly, we have to describe the probability of the event that
the $\tau$-collision protocol does not terminate after $t$ rounds in the form of a
hypergraph property. To achieve this, we start by describing the structure of witness
trees.

In the setting of the $\tau$-collision protocol in parallel arrival, a witness tree
has the following structure. Using the notation of \cite{SchickingerS00},
a machine is \emph{active in round $t$} if there exists at least one job
that sends a request to this machine in round $t$. If no such job exists, 
the machine is \emph{inactive in round $t$}. 
Assume that after round $t$ the collision protocol has
not yet terminated. Then there exists a machine $y$ that is active 
in round $t$ and that received  more
than $\tau$ allocation requests. Arbitrarily choose $\tau$ of these requests.  
These requests were sent by $\tau$ unallocated jobs in
round $t$. The vertex that corresponds to machine $y$ is the root of the witness tree, 
the $\tau$ job vertices are its children. In round $t$, each of
the $\tau$ unallocated jobs sent allocation requests to $d-1$ other machines.
The corresponding machine vertices are the children
of each of the $\tau$ job vertices in the witness tree.
By definition, the machines that correspond to these machine vertices are also active in round $t$, and so they 
were active in round $t - 1$ as well. So, there are 
$\tau \cdot (d-1)$ machine vertices that correspond to machines that are active
and receive requests from one of the jobs on level $t$ in round $t - 1$. We must
be aware that among these
machine vertices the same machine might appear more than once, because unallocated jobs may have 
chosen the same candidate machine. So, there may exist vertices in the witness tree that
correspond to the same machine. 
For all these $\tau \cdot(d-1)$ machines the same argument holds in round $t-1$. Proceeding
with the construction for rounds $t-2, t-3,\ldots, 1$, we build the \emph{witness tree $T_t$
with root $y$}.  It exhibits a regular recursive
structure, depicted abstractly in Figure~\ref{fig:witness_tree}. Note that all
leaves, i.e., vertices on level 0, correspond to machine vertices, since no allocation requests 
are sent in round $0$.
\begin{figure}[tb]
\centering
\def\svgwidth{\textwidth}
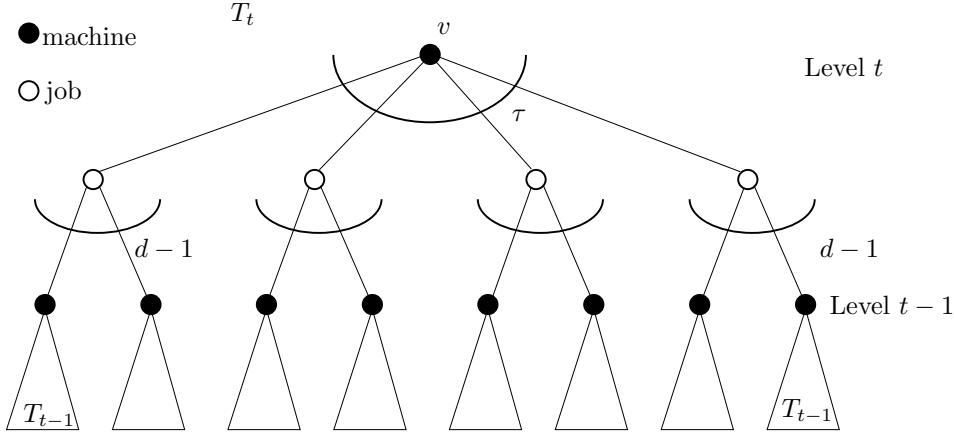
\caption{Structure of a witness tree $T_t$ with root $v$ after $t$ rounds if the $\tau$-collision
protocol with $d$ candidate machines has not yet terminated.}
\label{fig:witness_tree}
\end{figure}

As we have seen, such regular witness trees do not need to be subgraphs of
the allocation graph since two vertices of a witness tree might be embedded to the same
vertex. Hence, the witness tree is ``folded together'' to a subgraph in the
allocation graph. In the embedding of a witness tree as a subgraph
of the allocation graph, edges do not occur independently and the analysis becomes
difficult, even in the fully random case.

Schickinger and Steger found the following way to analyze this situation. 
For a connected undirected graph $G = (V, E)$ a \emph{shortest path tree $T$ 
rooted at node $s \in V$} is 
a tree in $G$ in which the unique paths from $s$ to all other vertices are
shortest paths in $G$.  (Such a tree can be obtained by starting a breadth-first search in $G$ from $s$.)
Schickinger and Steger introduced the notion of a \emph{multicycle} that describes an ``almost
tree-like'' graph. 
\begin{definition}
    Let $k,t \geq 1$. Let $G = (V,E)$ be an undirected graph. 
    Let $s \in V$.
    A \emph{$(k,t)$-multicycle of depth at most $t$ at node $s$
    in $G$} is a connected subgraph $G' = (V', E')$ of $G$ together with a shortest path tree $(V', T')$ of $G'$ 
    rooted at node $s$ with the following properties:
    \begin{enumerate}
        \item $G'$ includes vertex $s$.
        \item $G'$ has cyclomatic number $k$ (\emph{cf.} Section~\ref{hashing:sec:leafless}).
        \item For each vertex $v$ in $T'$, the distance between $s$ and $v$ in $T'$ is at most $t$. 
	\item Each leaf in $T'$ is incident to an edge in $G'$ that is not in $T'$. 
    \end{enumerate}
\end{definition}
Multicycles will be used to reason in the following way: When $G$ does not contain a certain
$(k,t)$-multicycle at node $s$ as a subgraph, removing only a few edges in $G$ incident to node $s$ makes the neighborhood that
includes all vertices of distance at most $t$
of $s$ in the remaining graph acyclic.  As we shall see, the proof that this is possible
will be quite easy when 
using shortest path trees. 

One easily checks that a $(k,t)$-multicycle $M$ with $m$
vertices and $n$ edges satisfies  
$n = m + k - 1$, because the cyclomatic number of a connected graph is exactly 
$n - m  + 1$ \cite{diestel}. Furthermore, it has at most $2kt + 1$ vertices, because there can be at most
$2k$ leaves that each have distance at most $t$ from $s$, and all vertices of the spanning tree lie on 
the unique paths from $s$ to the leaves. 
We will later see that for the parameters given in
Theorem~\ref{thm:parallel_schickinger}, a $(k,t)$-multicycle is with high
probability not a subgraph of the allocation graph.
\begin{lemma}[{\cite[Lemma 2]{SchickingerS00}}]
Let $k,t \geq 1$. Assume that a graph $G=(V,E)$ contains no $(k',t)$-multicycle,
for $k' > k$. 
Furthermore, given a vertex $v \in V$, consider the induced subgraph $H=(V',E')$ of $G$
that contains all vertices $w\in V$ with distance at most $t$ from $v$ in $G$.
Then we can remove at most $2k$ edges incident to $v$ in $H$ to get a graph $H^\ast$ such that
the connected component of $v$ in $H^\ast$ is a tree. 
\end{lemma}
%

In the light of this lemma, we set $\tau \geq 2k + 1$.
Then we know that if the allocation graph contains a witness tree after $t$ rounds, 
it contains a $(k,t)$-multicycle or a regular witness tree $T_{t-1}$. This observation motivates considering the following hypergraph property:
\begin{definition}
    Let $k, t \in \mathbb{N}$. Then $\MCWT{k,t} \subseteq \GG^{d}_{n/d,n}$
    is the set of all hypergraphs $H$ such that \emph{bi}$(H)$ forms either a $(k,t)$-multicycle
    or a witness tree $T_{t-1}$.
\end{definition}

If we use hash class $\RR$ and set $\tau \geq 2k + 1$, for a set $S$ of jobs, 
we have, by the discussion above:
\begin{align}
    \Pr(\text{the $\tau$-collision protocol does not terminate after $t$ rounds})
    \leq \Pr\left(N^{\MCWT{k,t}}_{S} > 0\right).
    \label{hashing:eq:lb:1}
\end{align}
By Lemma~\ref{lem:good:bad} we may bound the probability on the right-hand side of \eqref{hashing:eq:lb:1} by 
    \begin{align}
        \Pr\left(N^{\MCWT{k,t}}_{S} > 0\right) \leq
    \Pr\left(B^{\MCWT{k,t}}_S\right) + \E^\ast\left(N^{\MCWT{k,t}}_S\right).
    \label{hashing:eq:load:balancing:goal}
\end{align}

\paragraph{Bounding $\E^\ast(N^{\MCWT{k,t}}_S)$} 
We first bound the expected number of subgraphs that form multicycles or 
witness trees when the hash functions are fully random. The following lemma 
is equivalent to Theorem~$1$ in \cite{SchickingerS00}. However,
our parameter choices are slightly different because in \cite{SchickingerS00}
each of the $d$ candidate machines is chosen from the set $[n]$, while
we split the $n$ machines into $d$ groups of size $n/d$. The proof of the lemma 
can be found in Appendix~\ref{app:schickinger:steger:proof}.
\begin{lemma}
    Let $\alpha \geq 1$ and $d \geq 2$.  Set $\beta = 2d(\alpha + \ln d + 3/2)$ and $k = \alpha + 2$.
    Consider $t$ with $2 \leq t \leq (1/\beta) \ln \ln n$. Let 
    \begin{align*}
	\tau = \max\left\{\frac{1}{d-1}\left(\frac{\beta t \ln n}{\ln
	\ln n}\right)^{\frac{1}{t-2}}, d^{d+1}e^d + 1, 2k + 1\right\}. 
\end{align*}
    Then
    $$\E^\ast\left(N^{\MCWT{k,t}}_S\right) = O(n^{-\alpha}).$$
\label{lem:fully_random_100}
\end{lemma}

\paragraph{Bounding $\Pr(B^{\MCWT{k,t}}_S)$}
To apply
Lemma~\ref{lem:fail_prob}, we need a peelable hypergraph property that contains
$\MCWT{k,t}$. We will first calculate the size of witness trees to see that
they are small for the parameter settings given in
Theorem~\ref{thm:parallel_schickinger}.

\paragraph{The Size of Witness Trees}\label{subsec:load_balancing_size_of_wt}

Let $T_t$ be a witness tree after $t$ rounds. 
As above, the number of job vertices $j_t$ in $T_t$ is given by
\begin{align*}
j_t = \frac{\tau^t(d-1)^{t-1}-\tau}{\tau(d - 1) - 1}.
\end{align*}
We bound the size of the witness tree. 
\begin{lemma}
If $\alpha$, $d$, $\beta$, $k$, $t$, and $\tau$ are as in Lemma~\ref{lem:fully_random_100}, we have
$j_t < \log n$.
\label{lem:witness_tree_size}
\end{lemma}

\begin{proof}
Observe the following upper bound for the number of jobs in a witness tree after $t$ rounds:
\begin{align*}
    j_t = \frac{\tau^t (d-1)^{t-1}-\tau}{\tau(d-1) - 1} 
    \leq \frac{\tau(\tau(d-1))^{t-1}}{2\tau-1} \leq
(\tau(d-1))^{t-1}.
\end{align*}
Now observe
that for $\tau \in \{d^{d+1}e^d + 1, 2k + 1\}$ we have
\begin{align*}
(\tau(d-1))^{t-1} \leq (\tau(d-1))^{\frac{1}{\beta}\ln \ln n} \leq (\ln n)^{\frac{\ln \tau + \ln d}{\beta}}
\leq \ln n,
\end{align*}
since $\frac{\ln \tau}{\beta} \leq 1$ for the two constant choices for $\tau$ in Lemma~\ref{lem:fully_random_100}.
Furthermore, for $\tau = (\beta t (\ln n)/\ln \ln n)^{\frac{1}{t-2}}/(d-1)$ we have
\begin{align*}
((d-1)\tau)^{t - 1} \leq \frac{\beta t \ln n}{\ln
\ln n} \leq \ln n,
\end{align*}
and hence $j_t \leq \ln n < \log n$. \qquad
\end{proof}

A $(k,t)$-multicycle has at most $2kt + k - 1$ edges,
hence such multicycles are
smaller than witness trees for
$t = O(\log \log n)$ and a constant $k \geq 1$.

\paragraph{A Peelable Hypergraph Property}\label{subsec:load_balancing_graph_prop}

To apply Lemma~\ref{lem:fail_prob}, we have to find a peelable hypergraph property
that contains all subgraphs that have property $\MCWT{k,t}$ 
(multicycles or witness trees for $t-1$ rounds).
Since we know from above 
that witness trees and multicycles are contained in small connected subgraphs of
the hypergraph $G(S, \vec{h})$, we will use the following hypergraph property.

\begin{definition}
Let $K > 0$ and $d \geq 2$ be constants. Let $\Csmall(K,d)$ contain all connected
$d$-partite hypergraphs $(V,E) \in \GG^d_{n/d, n}$ with $|E| \leq K
\log n$ disregarding isolated vertices.
\end{definition}

The following central lemma shows how we can bound the failure term of 
$\RR$.

\begin{lemma}
Let $K > 0$, $c \geq 1$, $\ell \geq 1$, and $d \geq 2$ be constants. Let $S$ be the set of jobs with $|S|
= n$. Then 
\begin{align*}
    \Pr\left(B^{\MCWT{k,t}}_S\right) \leq \Pr\left(B^{\Csmall(K,d)}_S\right) =  O\left(\frac{n^{K (d + 1) \log d+2}}{\ell^{c}}\right).
\end{align*}\label{lem:fail_prob_csmall}
\end{lemma}
For proving this bound, we need the following
auxiliary hypergraph property. Note that some of the considered graphs are not $d$-uniform. 

\begin{definition}
Let $K > 0$,  $\ell \geq 1$, and $d \geq 2$ be constants. Let $n \geq 1$ be
given.  Then $\HT(K, d,\ell)$ (``hypertree'') is the set of all $d$-partite hypergraphs
$G=(V,E)$ in $\GG^d_{n/d, n}$ with $|E| \leq K \log n$ for which
$\textnormal{\textrm{bi}}(G)$ (disregarding isolated vertices) is a tree, has at most $\ell$ leaf edges and
has leaves only on the left (job) side.

\end{definition}
We will now establish the following connection between
$\Csmall(K,d)$ and $\HT(K,d,\ell)$.

\begin{lemma}
Let $K > 0$, $d \geq 2$ and $c \geq 1$ be constants. Then $\Csmall(K,d)$ is
$\HT(K,d,2c)$-$2c$-reducible, \emph{cf.} Definition~\ref{def:reducible:gen}.
\label{lem:csmall_reducibility}
\end{lemma}

\begin{proof}
Assume $G = (V,E)\in \Csmall(K,d)$.
Arbitrarily choose $E^\ast \subseteq E$ with $|E^\ast| \leq 2c$. We have to show
that there exists an edge set $E'$ such that $(V,E') \in
\HT(K,d,2c)$, $(V, E')$ is a subgraph of $(V,E)$, and for each edge $e^\ast \in E^\ast$
there exists an edge $e' \in E'$ such that $e' \subseteq e^\ast$ and $e'$ and $e^\ast$
have the same label.

Identify an arbitrary spanning tree $T$ in bi$(G)$. Now repeatedly remove leaf
vertices with their incident edges, as long as these leaf vertices do not correspond to edges from
$E^\ast$.  Denote the resulting tree by $T'$. In the hypergraph representation,
$T'$ satisfies all 
properties from above. \qquad 
\end{proof}

From Lemma~\ref{lem:fail_prob} it follows that we have
$$\Pr\left(B^{\MCWT{k,t}}_S\right) \leq \Pr\left(B^{\Csmall(K,d)}_S\right) \leq \ell^{-c} \cdot \sum_{t = 2}^{n}
t^{2c} \mu^{\HT(K,d,2c)}_t.$$

\begin{lemma}
Let $K > 0, d \geq 2,$ and $c \geq 1$ be constants. If $t \leq K
\cdot \log n,$ then
\begin{align*}
\mu^{\HT(K,d,2c)}_t \leq t^{O(1)}\cdot n \cdot d^{(d + 1)t}. 
\end{align*}
For $t > K \log n$ it holds that $\mu^{\HT(K,d,2c)}_t = 0$.
\label{lem:prob_t(c,d,2k)}
\end{lemma}

\emph{Proof}.
It is trivial that $\mu^{\HT(K,d,2c)}_t = 0$ for $t > K \log n$, since
a hypergraph with more than $K \log n$ edges contains too many edges to have property $\HT(K,d,2c)$.

Now suppose $t \leq K \log n$. We first count labeled hypergraphs having property $\HT(K, d, 2c)$
consisting of $t$ job vertices and 
$z$ edges, for some fixed $z \in \{t, \ldots, dt\}$, in the \emph{bipartite representation}.
(Note that hypergraphs with property $\HT(K, d, 2c)$ may not be $d$-uniform. Thus, in the bipartite
representation not all vertices on the left side have $d$ neighbors on the right side.) 

There are at most $z^{O(2c)} = z^{O(1)}$ unlabeled trees with $z$ edges and
at most $2c$ leaf edges (Lemma~\ref{lem:num_graphs}). Fix one such tree $T$. 
There are not more than
$n^t$ ways to label
the job vertices of $T$, and there are at most $d^z$ ways to assign each 
edge a label from $\{1, \ldots, d\}$. Once these labels are fixed, 
there are at most $(n/d)^{z + 1 - t}$ ways to assign the
right vertices to machines. Fix such a fully labeled tree $T'$.

Now draw $z$ hash values at random from $[n/d]$ and build a graph
according to these hash values and the labels of $T'$. 
The probability that these random choices realize $T'$ is exactly $1/(n/d)^z$.
Thus we may estimate:

\begin{align*}
	\mu^{\HT(K,d,2c)}_t &\leq \sum_{z = t}^{dt}\frac{(n/d)^{z + 1 -
	t} \cdot z^{O(1)} \cdot n^t \cdot d^z}{(n/d)^{z}}
        = \sum_{z = t}^{dt}z^{O(1)} \cdot n \cdot d^{z - 1 + t}\\
        &< dt \cdot (dt)^{O(1)} \cdot n \cdot d^{(d + 1)t}
        =t^{O(1)}\cdot n \cdot d^{(d + 1) t}.\qquad \endproof
    \end{align*}

We can now proceed with the proof our main lemma.

\emph{Proof of Lemma~\ref{lem:fail_prob_csmall}}.
By Lemma~\ref{lem:fail_prob}, we know that 
$$\Pr\left(B^{\MCWT{k,t}}_S\right) \leq \ell^{-c} \cdot \sum_{t = 2}^{n}
t^{2c} \mu^{\HT(K,d,2c)}_t.$$

Applying the result of Lemma~\ref{lem:prob_t(c,d,2k)}, we calculate
   \begin{align*}
       \Pr\left(B^{\MCWT{k,t}}_S\right)      &\leq \ell^{-c} \cdot \sum_{t = 2}^{K \log n}
       t^{2c}
        t^{O(1)}\cdot n \cdot d^{(d + 1) t}
       = n \cdot \ell^{-c} \cdot (K \log n)^{O(1)}
       \cdot d^{(d+1) K \log n}\\
       &= O(n^2 \cdot \ell^{-c})\cdot  d^{K(d+1) \log n} = 
       O(n^{K(d+1)\log d+2} \cdot \ell^{-c}).\qquad \endproof
   \end{align*}

\paragraph{Putting Everything Together}

The previous lemmas allow us to complete the proof of the main theorem.

\emph{Proof of Theorem~\ref{thm:parallel_schickinger}}.

    We plug the results of  Lemma~\ref{lem:fully_random_100} and Lemma~\ref{lem:fail_prob_csmall} into \eqref{hashing:eq:load:balancing:goal} and get
    $$\Pr\left(N^{\MCWT{k,t}}_S > 0\right) \leq O\left(\frac{n^{K (d + 1)\log d + 2}}{\ell^{c}}\right) + O\left(\frac{1}{n^{\alpha}}\right).$$
For the case of parallel arrival with $d \geq 2$ hash functions, we calculated that witness trees do
not have more than $\log n$ edges (Lemma~\ref{lem:witness_tree_size}). So, we set $K = 1$.
Setting $\ell = n^{1/2}$ and $c = 2 (2 + \alpha + (d + 1) \log d)$ finishes the
proof of the  theorem. \qquad \endproof

\paragraph{Discussion on Load Balancing} We remark that the graph property $\Csmall(K,d)$ provides a very general result
on the failure probability of $\RR$ on hypergraphs $G(S, \vec{h})$. It can be applied for 
all results from \cite{SchickingerS00}. We will exemplify this statement by discussing what needs to be 
done to show that $\RR$ works in the setting of Voecking's ``Always-Go-Left'' sequential allocation algorithm \cite{voecking}. 
By
specifying explicitly how to break ties (always allocate the 
job to the ``left-most'' machine), Voecking's algorithm decreases the 
maximum bin load (w.h.p.)
in sequential load balancing with $d \geq 2$ hash functions from $\ln \ln n/\ln d + O(1)$ (arbitrary
tie-breaking)
\cite{AzarBKU99} to
$\ln \ln n/(d \cdot \ln \Phi_d) + O(1)$, which is an exponential improvement in $d$.
Here $\Phi_d$ is defined as follows. Let $F_d(j) = 0$ for
$j \leq 0$ and $F_d(1) = 1$. For $j \geq 2$, $F_d(j) = \sum_{i = 1}^d F_d(j - i)$.
(This is a generalization of the Fibonacci numbers.) Then 
$\Phi_d = \lim_{j \to \infty} F_d(j)^{1/j}$.
It holds that $\Phi_d$ is a constant with $1.61 \leq \Phi_d \leq 2$, see \cite{voecking}.
(We refer to \cite[Section 5.2]{SchickingerS00} and \cite{voecking} for
details about the description of the algorithm.) In the unified witness tree approach
of Schickinger and Steger, the main difference between the analysis of parallel arrivals and the sequential algorithm of Voecking
is in the definition of the
witness tree. 
Here, the analysis in \cite{SchickingerS00} also assumes that 
the machines are split into $d$ groups of size $n/d$. This means that we can just re-use their 
analysis in the fully random case. For bounding the failure term of hash class $\RR$, we have to show 
that the witness trees in the case of Voecking's ``Go-Left'' algorithm (see \cite[Fig. 6]{SchickingerS00})
have at most $O(\log n)$ jobs, i.e., that they are contained in
small connected components. Otherwise, we cannot apply Lemma~\ref{lem:fail_prob_csmall}.  

According to \cite[Page 84]{SchickingerS00}, the number of
job vertices $j_\ell$ in a witness tree for a bin with load $\ell$ is bounded by 
\begin{align}
    \label{eq:wt:leaves}
j_\ell \leq 4 h_\ell + 1,
\end{align}
where $h_\ell$ is the number of leaves in the witness tree. Following Voecking \cite{voecking}, 
Schickinger and Steger show that 
setting $\ell$ as large as $\ln \ln n / (d \ln \Phi_d) + O(1)$ is sufficient to 
bound the expected number 
of witness trees by $O(n^{-\alpha})$. Such witness trees have only $O(\log n)$ many job nodes.
\begin{lemma}
Let $\alpha > 0$ and $\ell = \log_{\Phi_d} (4 \log n^\alpha)/d$. Then $j_\ell
\leq 33 \alpha \log n$.
\end{lemma}

\begin{proof}
It holds $h_\ell = F_d(d \cdot \ell + 1)$, see \cite[Page 84]{SchickingerS00}, 
and $F_d(d \cdot \ell + 1) \leq \Phi_d^{d \cdot \ell + 1}$, since $F_d(j)^{1/j}$ is monotonically increasing. 
We obtain the bound  
\begin{align*}
j_\ell &\leq 4 \cdot h_\ell + 1 \leq 4 \cdot \Phi_d^{d \cdot \ell + 1} + 1 \leq
4\cdot \Phi_d^{\log_{\Phi_d}(4 \log n^\alpha) + 1} + 1\\ 
&= 16\cdot \Phi_d \cdot \alpha \log n + 1 \leq 33 \alpha \log n,
\end{align*}
using $\Phi_d \leq 2$ and assuming $\alpha \log n \geq 1.$
\end{proof}

Thus, we know that a witness tree in the setting of Voecking's algorithm is contained in a connected hypergraph with 
at most $33 \alpha \log n$ edges. Thus, we may apply Lemma~\ref{lem:fail_prob_csmall} in the same way as we did 
for parallel arrival. The result is that for given $\alpha > 0$ we can
choose $(h_1,\ldots,h_d) \in \RR^{c, d}_{\ell, n}$ with 
$\ell = n^\delta$, $0 < \delta < 1,$ and $c \geq (33\alpha(d+1) \log d + 2 + \alpha) / \delta$
and know that the maximum load is $(\ln \ln n) / (d \cdot \ln \Phi_d) + O(1)$ with probability $1 - O(1/n^\alpha)$. 
So, our general analysis using small connected hypergraphs makes it very easy to show that hash class $\RR$ suffices to run a specific algorithm with load guarantees.

When we are interested in making the parameters for setting up a 
hash function as small as possible, 
one should take care when bounding the constants in the logarithmic bound on
the number of edges in
the connected hypergraphs. (According to \cite{SchickingerS00}, \eqref{eq:wt:leaves}
can be improved by a more careful argumentation.) More promising
is a direct approach 
to witness trees,
as we did in the analysis of the algorithm of Eppstein \emph{et al.} in the previous subsection,
i.e., directly peeling the witness tree. 
Using such an approach,  Woelfel showed in \cite[Theorem 2.1 and its discussion]{Woelfel06}
that smaller parameters for the hash functions $\RR$ are sufficient
to run Voecking's algorithm.

\paragraph{Application to Generalized Cuckoo Hashing} We further remark that
the analysis of the $\tau$-collision protocol makes 
it possible to
analyze the space utilization of generalized cuckoo hashing using $d \geq 2$
hash functions and buckets which hold up to  $\kappa \geq 2$ keys in each table
cell, as proposed by Dietzfelbinger and Weidling in \cite{DietzfelbingerW07}.
Obviously, a suitable assignment of keys to table cells is equivalent to
a $\kappa$-orientation of $G(S, \vec{h})$.  It is well-known that any graph
that has an empty $(\kappa + 1)$-core, i.e., that has no subgraph in which all
vertices have degree at least $\kappa + 1$, has a $\kappa$-orientation, see,
e.g., \cite{DevroyeM09} and the references therein. The $(\kappa + 1)$-core of a graph can be
obtained by repeatedly removing vertices with degree at most $\kappa$ and their
incident hyperedges. The precise study of this process is due to Molloy
\cite{Molloy05}. The $\tau$-collision protocol is the parallel variant of this
process, where in each round all vertices with degree at most $\tau$ are removed
with their incident edges. (In the fully random case, properties of this
process were recently studied by Jiang, Mitzenmacher, and Thaler in \cite{JiangMT14}.) In terms of 
orientability, Theorem~\ref{thm:parallel_schickinger} with the exact parameter choices from 
Lemma~\ref{lem:fully_random_100} shows that for $\tau = \max\{e^\beta, d^{d+1}e^d + 1, 2k + 1\}$ 
there exists (w.h.p.) an assignment of the $n$ keys 
to $n$ memory cells when each cell can hold $\tau$ keys. (This is equivalent to a hash table
load of $1/\tau$.) It is open to find good space bounds for generalized cuckoo hashing
using this approach. However, we think that it suffers from the same general problem as the analysis for generalized 
cuckoo hashing with $d \geq 3$ hash functions and one key per table cell:
Since the analysis
builds upon a process which requires an empty $(\kappa + 1)$-core in the hypergraph
to succeed, space utilization seems to decrease for $d$ and $\kappa$ getting larger. Table~\ref{tab:core:thresholds}
contains space utilization bounds for static generalized cuckoo hashing with $d \geq 3$ hash functions and $\kappa$ 
elements per table cell when the assignment is obtained via a process that requires the $(\kappa + 1)$-core to be empty.
These calculations clearly support the conjecture that space utilization decreases for larger values of $d$ and $\kappa$. 

\begin{table}[t!]
    \centering
    \begin{tabular}{c | c | c| c | c | c | c }
        $_{\kappa + 1}\backslash^d$ & $3$ & $4$ & $5$ & $6$ & $7$ & $8$ \\ \hline
        $2$ & $0.818$ & $0.772$ & $0.702$ & $0.637$ & $0.582$ & $0.535$ \\ \hline
        $3$ & $0.776$ & $0.667$ & $0.579$ & $0.511$ & $0.457$ & $0.414$ \\ \hline 
        $4$ & $0.725$ & $0.604$ & $0.515$ & $0.450$ & $0.399$ & $0.359$ \\ \hline
        $5$ & $0.687$ & $0.562$ & $0.476$ & $0.412$ & $0.364$ & $0.327$ \\ \hline
        $6$ & $0.658$ & $0.533$ & $0.448$ & $0.387$ & $0.341$ & $0.305$ \\
    \end{tabular}
    \caption{Space utilization thresholds for generalized cuckoo hashing with $d
    \geq 3$ hash functions and $\kappa  + 1$ keys per cell, for $\kappa \geq 1$, based on the
    non-existence of the $(\kappa + 1)$-core. Each table cell gives the maximal
    space utilization achievable for the specific pair $(d, \kappa + 1)$. These
    values have been obtained using Maple\textsuperscript{\textregistered} to evaluate the formula from Theorem~$1$ of
    \cite{Molloy05}. }
    \label{tab:core:thresholds}
\end{table}

%% file: figures/witness_tree.pdf_tex
\begingroup%
  \makeatletter%
  \providecommand\color[2][]{%
    \errmessage{(Inkscape) Color is used for the text in Inkscape, but the package 'color.sty' is not loaded}%
    \renewcommand\color[2][]{}%
  }%
  \providecommand\transparent[1]{%
    \errmessage{(Inkscape) Transparency is used (non-zero) for the text in Inkscape, but the package 'transparent.sty' is not loaded}%
    \renewcommand\transparent[1]{}%
  }%
  \providecommand\rotatebox[2]{#2}%
  \ifx\svgwidth\undefined%
    \setlength{\unitlength}{813.66826172bp}%
    \ifx\svgscale\undefined%
      \relax%
    \else%
      \setlength{\unitlength}{\unitlength * \real{\svgscale}}%
    \fi%
  \else%
    \setlength{\unitlength}{\svgwidth}%
  \fi%
  \global\let\svgwidth\undefined%
  \global\let\svgscale\undefined%
  \makeatother%
  \begin{picture}(1,0.45033133)%
    \put(0,0){\includegraphics[width=\unitlength]{figures/witness_tree.pdf}}%
    \put(0.44011669,0.40515078){\color[rgb]{0,0,0}\makebox(0,0)[lb]{\smash{$v$}}}%
    \put(0.22942544,0.41813033){\color[rgb]{0,0,0}\makebox(0,0)[lb]{\smash{$T_t$}}}%
    \put(0.0367478,0.39330074){\color[rgb]{0,0,0}\makebox(0,0)[lb]{\smash{machine}}}%
    \put(0.04072053,0.33867563){\color[rgb]{0,0,0}\makebox(0,0)[lb]{\smash{job}}}%
    \put(0.01774098,0.00794549){\color[rgb]{0,0,0}\makebox(0,0)[lb]{\smash{$T_{t-1}$}}}%
    \put(0.79211268,0.01291143){\color[rgb]{0,0,0}\makebox(0,0)[lb]{\smash{$T_{t-1}$}}}%
    \put(0.51645549,0.31583242){\color[rgb]{0,0,0}\makebox(0,0)[lb]{\smash{$\tau$}}}%
    \put(0.13110026,0.17678668){\color[rgb]{0,0,0}\makebox(0,0)[lb]{\smash{$d-1$}}}%
    \put(0.83110026,0.17678668){\color[rgb]{0,0,0}\makebox(0,0)[lb]{\smash{$d-1$}}}%
    \put(0.81540376,0.36151885){\color[rgb]{0,0,0}\makebox(0,0)[lb]{\smash{Level $t$}}}%
    \put(0.84122655,0.11818887){\color[rgb]{0,0,0}\makebox(0,0)[lb]{\smash{Level $t-1$}}}%
  \end{picture}%
\endgroup%

%% file: sec-5-generalization.tex
\section{A Generalized Version of the Hash
Class}\label{hashing:sec:generalization}
In this short section we present a generalized version of our hash class that
uses arbitrary $\kappa$-wise independent hash classes as building blocks.

\subsection{The Generalized Hash Class}
The following definition is a generalization of Definition~\ref{def:family:R} to functions
with higher degrees of independence than two.
\begin{definition}
    Let $c\ge1$, $d\ge2$, and $\kappa \geq 2$.  For integers $m$, $\ell\ge 1$, and given
    $f_1,\ldots,f_d\colon U\to [m]$, $g_1,\ldots,g_c\colon U \to [\ell]$, and
    $d$ two-dimensional tables $z^{(i)}[1..c,0..\ell-1]$ with elements from $[m]$ for
    $i\in\{1,\ldots,d\}$, we let
    $\vec{h} = (h_1,\ldots,h_d) = (h_1,\ldots,h_d)\langle
    f_1, \ldots, f_d,g_1,\ldots,g_c,\allowbreak z^{(1)},\ldots, z^{(d)}\rangle$,
    where $$ {h_i(x) = \Bigl(f_i(x) +\sum_{1\le j \le c}
    z^{(i)}[j, g_j(x)]\Bigr) \bmod m\text{, for }x\in U, i\in\{1,\ldots,d\}.} $$ 

    \smallskip
    \noindent Let $\HH^{\kappa}_m$ [$\HH^{\kappa}_\ell$] be an arbitrary $\kappa$-wise independent hash class with
    functions from $U$ to $[m]$ [from $U$ to $[\ell]$].
    Then $\RR^{c,d,\kappa}_{\ell,m}(\HH^\kappa_\ell, \HH^\kappa_m)$ is the class of
    all sequences $(h_1,\ldots,h_d)\langle f_1, \allowbreak\ldots, f_d,
    g_1, \ldots, g_c,\allowbreak z^{(1)}, \ldots, z^{(d)}\rangle$ for $f_i \in \HH^\kappa_m$ with $1 \leq i \leq d$ and
    $g_j \in \HH^\kappa_\ell$ with $1 \leq j \leq c$.
    \label{def:family:R:gen}
\end{definition}

We consider $\RR^{c,d,2k}_{\ell,m}\left(\HH^{2k}_\ell,
\HH^{2k}_m\right)$ for some fixed $k \in \mathbb{N}, k \geq 1$. For the
parameters $d = 2$ and $c = 1$, this is the hash class used by Dietzfelbinger
and Woelfel in \cite{DW2003a}. We first analyze the properties of this hash
class by stating a definition similar to Definition~\ref{def:T:bad} and a lemma
similar to Lemma~\ref{lem:random}. 
We hope that comparing the proofs of
Lemma~\ref{lem:random} and Lemma~\ref{lem:random:gen} shows the (relative)
simplicity of the original analysis.
\begin{definition}
For $T \subseteq U$, define the random variable 
$d_T$, the ``deficiency'' of $\vec{h} = (h_1,\ldots,h_d)$ with respect to $T$, by 
$d_T(\vec{h})=|T| - \max\{k, |g_1(T)|,\ldots, |g_c(T)|\}$.
\emph{(}Note\emph{:} $d_T$ depends only on the $g_j$-components of
$(h_1,\ldots,h_d)$.\emph{)} 
Further, define
\begin{enumerate}
    \item[\textnormal{(i)}] \mbox{\emph{$\text{bad}_T$} as the event that $d_T >
	k$\emph{;}}
    \item[{\textnormal{(ii)}}] \emph{$\text{good}_T$}{} as \emph{$\overline{\text{bad}_T}$}, 
	           i.e., the event that $d_T \le k$\emph{;}
    \item[{\textnormal{(iii)}}] \emph{$\text{crit}_T$} as the event that $d_T =
	k$.
\end{enumerate}
Hash function sequences $(h_1,\ldots,h_d)$ in these events are called 
\emph{``$T$-bad''}, \emph{``$T$-good''}, and \emph{``$T$-critical''}, resp.
\label{def:T:bad:gen}
\end{definition}

The following lemma is identical to \cite[Lemma 1]{AumullerDW14}.
\begin{lemma}
  Assume $d\ge2$, $c\ge 1,$ and $k\ge 1$. 
  For $T\subseteq U$, the following holds\emph{:}
  \begin{enumerate}
      \item[\textnormal{(a)}] \emph{$\Pr(\text{bad}_T\cup\text{crit}_T) \le
	  \bigl(\left|T\right|^{2}/\ell\bigr)^{ck}$}.
      \item[\textnormal{(b)}] Conditioned on \emph{$\text{good}_T$} \emph{(}or on \emph{$\text{crit}_T$}\emph{)},
  the hash values  $(h_1(x), \ldots, h_d(x))$, $x\in T$, are 
  distributed uniformly and independently in $[r]^d$.
  \end{enumerate}
\label{lem:random:gen}
\end{lemma}

\subsection{Application of the Hash Class} 

The central lemma to bound the impact of using our hash class in contrast to
fully random hash functions was Lemma~\ref{lem:fail_prob}. One can reprove this
lemma in an analogous way for the
generalized version of the hash class, using the probability bound from
Lemma~\ref{lem:random:gen}(a) to get the following result.

\begin{lemma}
    Let $c \geq 1$, $k \geq 1$, $S\subseteq U$ with $|S| = n,$ and let $\mathsf{A}$ be a graph property. 
    Let $\mathsf{B} \supseteq \mathsf{A}$ be a peelable graph property. Let
    $\mathsf{C}$ be a graph property such that $\mathsf{B}$ is
    $\mathsf{C}$-$2ck$-reducible. Then  
    $$\Pr\left(B^{\mathsf{A}}_S\right) \leq \Pr\left(B^{\mathsf{B}}_S\right) \leq
    \ell^{-ck} \sum_{t = 2k}^{n} t^{2ck} \cdot
   \mu^\mathsf{C}_t.$$
\label{lem:fail_prob:gen}
\end{lemma}

\subsection{Discussion} One can now redo all the calculations from
Section~\ref{hashing:sec:applications:simple:graphs} and Section~\ref{hashing:sec:applications:hypergraphs}. We discuss the
differences. Looking at Lemma~\ref{lem:fail_prob:gen}, we notice the
$(t^{2ck})$-factor in the sum instead of $t^{2c}$. Since $k$ is fixed, this factor does not
change anything in the calculations that always used $t^{O(1)}$ (see, e.g., the
proof of Lemma~\ref{lem:failure_term_bound}). The factor $1/\ell^{ck}$ (instead
of $1/\ell^c$) leads to lower values for $c$ if $k \geq 2$. E.g., in cuckoo hashing
with a stash, we have to set $c \geq (s+2)/(\delta k)$ instead of $c \geq
(s+2)/\delta$. This improves the space usage, since we need less tables
filled with random values. However, the higher degree of independence needed for
the $f$- and $g$-components leads to a higher evaluation time of a single
function.

%% file: appendix.tex
\section{Proof of Lemma~\ref{lem:fully_random_100}}
\label{app:schickinger:steger:proof}
\begin{proof}
    For the sake of the analysis, we regard $\MCWT{k,t}$ as the union of
    two graph properties $\MC{k,t}$, hypergraphs that form $(k,t)$-multicycles, 
    and $\WT{t-1}$, hypergraphs that form witness trees for the parameter
    $t-1$. We show the lemma by proving $\E^\ast\left(N^{\MC{k,t}}_S\right) =
    O(n^{-\alpha})$ and $\E^\ast\left(N^{\WT{t-1}}_S\right) = O(n^{-\alpha})$. 
    In both cases, we consider the bipartite representation of hypergraphs. 
    Our proofs follow \cite[Section~$4$]{SchickingerS00}.

    We start by bounding $\E^\ast\left(N^{\MC{k,t}}_S\right)$. As we have seen, a
    $(k,t)$-multicycle is a connected graph that has at most $2kt$ vertices and 
    cyclomatic number $k$. We start by counting
    $(k,t)$-multicycles with exactly $s$ vertices and $s+k-1$ edges, for 
    $s \leq 2kt$. In this case, we have to
    choose $j$ and $u$ (the number of jobs and machines, resp.) such that $s = j
    + u$. By Lemma~\ref{lem:num_graphs} there are at most $(s + k - 1)^{O(k)}$
    unlabeled $(k,t)$-multicycles. Fix such an unlabeled $(k,t)$-multicycle 
    $G$. There are two ways
    to split the vertices of $G$ into the two sides of the
    bipartition. ($G$ can be assumed to be bipartite since we consider $(k,t)$-multicycles
    that are subgraphs of the allocation graph.) Once this bipartition is fixed,
    we have $n^{j}$ ways to choose the job vertices and label vertices of $G$
    with these jobs. There are $d^{s+k - 1}$ ways to label the edges of $G$ with
    labels from $1,\dots,d,$ which represent the request modeled by an edge
    between a job vertex and a machine vertex. Once this labeling is fixed,
    there are $(n/d)^{u}$ ways to choose machine vertices and label the
    remaining vertices of $G$. Fix such a fully labeled graph $G'$.

    For each request $r$ of a job $w = 1, \ldots, j$, choose  a machine from
    $[n/d]$ at random and independently.
    The probability that this machine is the same machine that $w$ had chosen 
in $G'$ is $d/n$. Thus, the
    probability that $G'$ is realized by the random choices is $(d/n)^{s+k-1}$.
    By setting $k = \alpha + 2$ and using the parameter choice $t = O(\ln \ln n)$ we calculate
    \begin{align*}
	\E^\ast\left(N^{\MC{k,t}}_S\right) 
	&\leq \sum_{s = 1}^{2kt} \sum_{u + j = s} 2 \cdot n^{j} \cdot (n/d)^u \cdot d^{s+k-1}
    \cdot (s + k - 1)^{O(k)} \cdot (d/n)^{s + k - 1}\\
    &\leq n^{1 - k} \sum_{s = 1}^{2kt} 2s \cdot d^{2(s+k-1)} \cdot (s+k-1)^{O(1)}\\
    &\leq n^{1 - k} \cdot 2kt \cdot 4kt \cdot d^{2(2kt + k -1)} \cdot (2kt+k -
    1)^{O(1)}\\
    &\leq n^{1 - k} \cdot (\ln \ln n)^{O(1)} \cdot (\ln n)^{O(1)}
    = O(n^{ 2 - k}) = O(n^{-\alpha}).
    \end{align*}
Now we consider $\E^\ast\left(N^{\WT{t-1}}_S\right)$. 
    By the simple recursive structure of witness trees, a witness tree of depth $t-1$ has $j =
    \frac{\tau^{t-1}(d-1)^{t-2} - \tau}{\tau(d-1) - 1}$ job vertices and $u =
    \frac{\tau^{t-1}(d-1)^{t-1} -
    1}{\tau(d-1) - 1}$ machine vertices. Let $T$ be an unlabeled witness tree
    of depth $t - 1$. $T$ has 
    $r = d \cdot j$ edges. There are at most
    $n^j$ ways to choose $j$ jobs from $S$ and label the job vertices of $T$ and
    at most $d^r$ ways to label the edges with a label from
    $\{1,\dots,d\}$. Once this labeling is fixed, there are at most $(n/d)^u$ ways to choose
    the machines and label the machine vertices in the witness tree. With these rough estimates,
    we over-counted the number of witness trees by at least a factor of $(\tau!)^{j/\tau} \cdot
    ((d-1)!)^j$. (See Figure~\ref{fig:witness_tree}. For each job vertex, there are $(d-1)!$ 
    labelings which result in the same witness tree. Furthermore, for each non-leaf machine vertex, there
    are $\tau!$ many labelings which yield the same witness tree.)
    Fix such a fully
    labeled witness tree $T'$.  

    For each request of a job $w = 1, \ldots, j$ choose at random 
    a machine from $[n/d]$.
    The probability that the edge matches the edge in $T'$ is $d/n$. Thus, the
    probability that $T'$ is realized by the random choices is $(d/n)^{r}$. 
    We calculate
    \begin{align*}
        \E^\ast\left(N^{\WT{t-1}}_S\right) 
&\leq n^j \cdot d^r \cdot (n/d)^u \cdot \left(\frac{1}{\tau!}\right)^{j/\tau} \cdot
\left(\frac{1}{(d-1)!}\right)^j \cdot (d/n)^r\\
&\leq n \cdot d^{2r} \cdot \left(\frac{1}{\tau!}\right)^{j/\tau} \cdot
\left(\frac{1}{(d-1)!}\right)^j
\leq n \left[\frac{e}{\tau} \cdot \left(\frac{e}{d-1}\right)^{d-1}\cdot
d^{2d}\right]^j\\
&\leq n \left(\frac{e^d \cdot d^{d+1}}{\tau} \right)^j.
    \end{align*}
    Observe that
    \begin{align*}
    j = \frac{\tau^{t-1}(d-1)^{t-2} - \tau}{\tau(d-1) - 1} \geq
    \frac{\tau^{t-2}(d-1)^{t-2} - \tau}{d} \geq \frac{\left(\tau(d-1)\right)^{t-2}}{2d}.
    \end{align*}
    For the parameter settings assumed in the lemma we get
    \begin{align*}
	\E^\ast\left(N^{\WT{t-1}}_S\right) &\leq n \left(\frac{e^d \cdot d^{d+1}}{\tau}
	\right)^{\frac{\left(\tau(d-1)\right)^{t-2}}{2d}}
	= n \left(\frac{e^d \cdot d^{d+1}}{\tau}
	\right)^{\frac{\beta t \ln n}{2d \ln \ln n}}\\
	&\leq n \left(e^d \cdot d^{d+2} \cdot \left(\frac{\beta \ln \ln n}{t \ln n}\right)^{\frac{1}{t-2}}
	\right)^{\frac{\beta t \ln n}{2d \ln \ln n}}\\
	&\leq n \cdot \left(\left(e^d \cdot d^{d+2}\right)^t \left(\frac{\beta \ln \ln n}{t \ln n}\right)\right)^
	{\frac{\beta \ln n}{2d \ln \ln n}}\\
	&\leq  n \cdot \left(\left(e^d \cdot d^{d+2}\right)^{\frac{1}{\beta}\ln \ln n} \cdot \frac{1}{\ln n}\right)^
	{\frac{\beta \ln n}{2d \ln \ln n}}\\
	&\leq n^{3/2 + \ln d - \beta / 2d}.
    \end{align*}
    Setting $\beta = 2d (\alpha + \ln d + 3/2)$ 
    suffices to show that $\E^\ast(N^{\MCWT{k,t}}_S) = O(n^{-\alpha})$.\qquad
\end{proof}